\documentclass[11pt,a4paper]{article}

\usepackage{authblk}
\usepackage{fullpage}
\usepackage{amsmath}
\usepackage{hyperref}
\usepackage{cleveref}
\usepackage{comment}
\usepackage[inline]{enumitem}
\usepackage{amsthm}
\usepackage{bm}

\usepackage{titlesec}

\titleformat{\paragraph}
{\normalfont\normalsize\bfseries}{\theparagraph}{1em}{}
\titlespacing*{\paragraph}
{0pt}{3.25ex plus 1ex minus .2ex}{1.5ex plus .2ex}

\title{Deepening the  (Parameterized) Complexity Analysis of Incremental Stable Matching Problems} 

\author{Niclas Boehmer}
\author{Klaus Heeger}
\author{Rolf Niedermeier}

\affil{\small
  Technische Universit\"at Berlin, Faculty~IV, Algorithmics and Computational 
  Complexity\protect\\
  \{niclas.boehmer,heeger\}@tu-berlin.de}

\date{\today}

\usepackage{framed}
\usepackage{tabularx}

\newlength{\RoundedBoxWidth}
\newsavebox{\GrayRoundedBox}
\newenvironment{GrayBox}[1]%
   {\setlength{\RoundedBoxWidth}{.93\columnwidth}
    \def\boxheading{#1}
    \begin{lrbox}{\GrayRoundedBox}
       \begin{minipage}{\RoundedBoxWidth}}%
   {   \end{minipage}
    \end{lrbox}
    \begin{center}
    \begin{tikzpicture}%
       \node(Text)[draw=black!20,fill=white,rounded corners,inner sep=2ex,text width=\RoundedBoxWidth]
             {\usebox{\GrayRoundedBox}};
        \coordinate(x) at (current bounding box.north west);
        \node [draw=white,rectangle,inner sep=3pt,anchor=north west,fill=white]
        at ($(x)+(6pt,.75em)$) {\boxheading};
    \end{tikzpicture}
    \end{center}}

\newenvironment{defproblemx}[1]{\noindent\ignorespaces%
                                \FrameSep=6pt%
                                \parindent=0pt%
                \begin{GrayBox}{#1}%
                \begin{tabular*}{\columnwidth}{!{\extracolsep{\fill}}@{\hspace{.1em}}
                 >{\itshape} p{1.4cm} p{0.88\columnwidth} @{}}%
            }{
                 \vspace*{-1em}
                 \end{tabular*}%
                \end{GrayBox}%
                \ignorespacesafterend
            }

\usepackage{xspace}

\usepackage[textwidth=20mm,obeyFinal,colorinlistoftodos]{todonotes}
\usepackage{marginnote}

\usepackage{multirow}
\usepackage{booktabs}
\usepackage{tabularx,environ}

\usepackage{thmtools}
\usepackage{thm-restate}

\newcommand{\mytodo}[2]{\xspace}
\newcommand{\myrevtodo}[2]{{%
		\let\marginpamarginnote
		\reversemarginpar
		\renewcommand{\baselinestretch}{0.8}%
		}}
\newcommand{\myinlinetodo}[2]{\todo[size=\small, color=#1!50!white, inline,
	caption={}]{#2}\xspace}
\newcommand{\registerAuthor}[3]{%
	\expandafter\newcommand\csname #2com\endcsname[1]{\mytodo{#3}{\textsc{#2}:
			##1}}%
	\expandafter\newcommand\csname
	#2revcom\endcsname[1]{\myrevtodo{#3}{\textsc{#2}: ##1}}%
	\expandafter\newcommand\csname
	#2inline\endcsname[1]{\myinlinetodo{#3}{\textsc{#2}: ##1}}%
	\expandafter\newcommand\csname
	#2inlineLater\endcsname[1]{\lv{\myinlinetodo{#3}{\textsc{#2}: ##1}}}%
}
\registerAuthor{Niclas Boehmer}{nb}{orange}
\registerAuthor{Klaus Heeger}{kh}{green}

\DeclareMathOperator{\lb}{\text{lb}}
\DeclareMathOperator{\lm}{\text{lm}}
\DeclareMathOperator{\lt}{\text{lt}}
\DeclareMathOperator{\rb}{\text{rb}}
\DeclareMathOperator{\rrm}{\text{rm}}
\DeclareMathOperator{\rt}{\text{rt}}

\newcommand{\temp}{\text{temp}}

\crefname{claim}{Claim}{Claims}
\crefname{observation}{Observation}{Observations}
\Crefname{observation}{Observation}{Observations}
\crefname{paragraph}{Paragraph}{Paragraphs}
\Crefname{paragraph}{Paragraph}{Paragraphs}

\tikzstyle{vertex}=[draw, circle, fill, inner sep = 2pt]
\tikzstyle{squared-vertex}=[draw, fill, inner sep = 2pt]
\usetikzlibrary{arrows.meta}

\DeclareMathOperator{\single}{single}

\DeclareMathOperator{\bc}{bc}
\DeclareMathOperator{\wc}{wc}

\newcommand{\after}{\text{after}}
\newcommand{\before}{\text{before}}

\usepackage{amsmath,amsfonts,amssymb}

\include{appHandler}

\allowdisplaybreaks
\usetikzlibrary{calc}
\tikzstyle{vertex}=[draw, circle, fill, inner sep = 2pt]
\tikzstyle{bedge}=[line width=1.8pt]
\tikzstyle{squared-vertex}=[draw, fill, inner sep = 2pt]
\usetikzlibrary{arrows.meta}

\usepackage{algorithm}
\usepackage[noend]{algpseudocode}
\algnewcommand\algorithmicinput{\textbf{Input:}}
\algnewcommand\algorithmicoutput{\textbf{Output:}}
\algnewcommand\Input{\item[\algorithmicinput]}
\algnewcommand\Output{\item[\algorithmicoutput]}
\algnewcommand\algorithmicgoto{\textbf{GoTo }}
\algnewcommand\GoTo{\item[\algorithmicgoto]}
\algnewcommand{\IIf}[1]{\State\algorithmicif\ #1\ \algorithmicthen}
\algnewcommand{\EndIIf}{\unskip\ \algorithmicend\ \algorithmicif}
\algnewcommand{\IfThenElse}[3]{
	\State \algorithmicif\ #1\ \algorithmicthen\ #2\ \algorithmicelse\ #3}

\usepackage{stackengine}
\stackMath
\newcommand\tsup[2][2]{%
	\def\useanchorwidth{T}%
	\ifnum#1>1%
	\stackon[-.5pt]{\tsup[\numexpr#1-1\relax]{#2}}{\scriptscriptstyle\sim}%
	\else%
	\stackon[.5pt]{#2}{\scriptscriptstyle\sim}%
	\fi%
}

\newcommand{\largeConst}{\nu m}

\DeclareMathOperator{\pend}{
    \succ 
    \overset{\raise0.3em\hbox{\text{\scriptsize{(rest)}}}}{\ldots}}

\makeatletter
\providecommand*{\cupdot}{%
	\mathbin{%
		\mathpalette\@cupdot{}%
	}%
}
\newcommand*{\@cupdot}[2]{%
	\ooalign{%
		$\m@th#1\cup$\cr
		\hidewidth$\m@th#1\cdot$\hidewidth
	}%
}
\makeatother

\DeclareMathOperator{\Ac}{Ac}

\newcommand{\decprob}[3]{%
  \begin{center}%
    \begin{minipage}{0.9\linewidth}%
      \textsc{#1}\\
      \textbf{Input:} #2\\
      \textbf{Question:} #3
    \end{minipage}%
  \end{center}%
}

\newtheorem{observation}{Observation}
\newtheorem{claim}{Claim}

\newtheorem{lemma}{Lemma}

\usepackage{url}

\newcommand{\ISMTwithoutspace}{\textsc{ISM-T}}
\newcommand{\ISRTwithoutspace}{\textsc{ISR-T}}
\newcommand{\ISR}{\textsc{ISR}\xspace}

\newcommand{\SRT}{\textsc{SR-T}\xspace}
\newcommand{\ISRT}{\ISRTwithoutspace\xspace}
\newcommand{\ISM}{\textsc{ISM}\xspace}
\newcommand{\ISMT}{\ISMTwithoutspace\xspace}

\newenvironment{claimproof}{%
	\proof}{\endproof}

\begin{document}
\maketitle

\begin{abstract}
 	When computing stable matchings, it is usually 
 	assumed that the preferences of the agents in 
the matching market are fixed. However, in many realistic scenarios, preferences change over time.
 	Consequently, an initially stable matching may become 
 	unstable.
 	Then, a natural goal is to find a  
 	matching which is stable with respect to the modified preferences and 
 	as close as possible to the initial one.
 	For \textsc{Stable Marriage/Roommates}, this problem was 
formally defined as \textsc{Incremental Stable Marriage/Roommates} by Bredereck et 
al.\ [AAAI '20].
 	As they showed that \textsc{Incremental Stable Roommates} and 
\textsc{Incremental Stable Marriage with Ties} are NP-hard, we focus on
    the parameterized complexity of these problems.
 	We answer two open questions of Bredereck et 
al.~[AAAI~'20]: We show that \textsc{Incremental Stable Roommates} is W[1]-hard 
 	parameterized by the number of changes in the preferences, yet admits an intricate 
 	XP-algorithm, and we show that \textsc{Incremental 
 	Stable Marriage with Ties} is W[1]-hard parameterized by the number 
 	of ties.
 	Furthermore, we analyze the influence of the degree of ``similarity'' between the agents'
preference lists, identifying several polynomial-time solvable and fixed-parameter 
tractable cases, but also proving that \textsc{Incremental Stable Roommates} and 
\textsc{Incremental Stable Marriage with Ties} parameterized by the number of 
different preference lists are W[1]-hard. 
\end{abstract}

\section{Introduction} 
Efficiently adapting solutions to changing inputs is a core issue in modern algorithmics~\cite{DBLP:conf/icalp/BhattacharyaHHK15,DBLP:conf/atal/BoehmerN21,DBLP:journals/siamcomp/CharikarCFM04,DBLP:conf/icalp/EisenstatMS14,HMS21}. 
In particular, in \emph{incremental} combinatorial problems, roughly speaking, 
the goal is to build new solutions incrementally while adapting to changes 
in the input. Typically, one wants to avoid (if possible) too radical 
changes in the solution relative to perhaps moderate changes in the input.
The corresponding study of incremental algorithms attracted 
research on numerous problems and scenarios~\cite{DBLP:conf/approx/GoemansU17}, including among many others
shortest path computation~\cite{DBLP:journals/jal/RamalingamR96}, flow computation~\cite{DBLP:journals/jmma/KumarG03},
clustering problems~\cite{DBLP:journals/siamcomp/CharikarCFM04,LMNN21}, and graph coloring~\cite{DBLP:journals/tcs/HartungN13}.

In this paper, we study the problem of adapting stable matchings under preferences to change. Consider for instance the following two scenarios:
First, as reported by Feigenbaum et al.~\cite{Feigenbaum17}, school seats in public schools are centrally assigned in New York.
Ahead of the start of the new year, all interested students are asked to submit their preferences over public schools. Then, a stable matching of students to public schools is computed and transmitted. 
However, in the past, shortly before the start of the new year typically around 10\% of students changed their preferences and decided to attend a private school instead, leaving the initially implemented matching unstable and triggering lengthy decentralized ad hoc updates.
Second, consider the assignment of freshmen to double bedrooms 
in college accommodation. 
After the orientation weeks, it is quite likely that students got to know each other (and in particular their roommates) better and thus their initially uninformed preferences changed, making the matching unstable. 

In our work, we focus on the problem of finding a stable matching after the ``change'' that is as close as possible to a given initially stable matching. 
The closeness condition here is due to the fact that in most applications reassignments come at some cost which we want to minimize (e.g., in the above New York example, reassigning students might make it necessary for the family to reallocate within the city). 
We build upon the work of Bredereck et al.~\cite{DBLP:conf/aaai/BredereckCKLN20}, who
performed a first systematic study of incremental versions of stable matching 
problems, and the recent (partially empirical) follow-up work by Boehmer et al.~\cite{uschanged}, who proved among others that different types of changes are ``equivalent'' to each other. 
The central focus of our studies lies on the \textsc{Stable Roommates}~(SR) problem: given a set of agents with each agent having preferences over other agents,
the task is to find a stable matching, i.e., a matching so that there are no two agents preferring each other to their assigned partner.
We also consider a famous special case of SR, called \textsc{Stable Marriage} (SM), where the set of agents is partitioned into two sets, and each agent may only be matched to an agent from the other set.
Formally, in the incremental versions of SR and SM, called \textsc{Incremental Stable Roommates}~(\ISR) and \textsc{Incremental Stable Marriage} (\ISM), we are given two preference profiles containing the preferences of each agent before and after the ``change'' and a matching that is stable in the preference profile before the change.
Then, the task is to find a matching that is stable after the change and as close as possible to the given matching, i.e., has a minimum symmetric difference to it.

	\subparagraph{Related Work.}
 	Bredereck et al.~\cite{DBLP:conf/aaai/BredereckCKLN20} formally introduced 
\textsc{Incremental Stable Marriage [with ties]} (\ISM/[\ISMTwithoutspace]) and \textsc{Incremental Stable Roommates [with ties]} (\ISR/[\ISRTwithoutspace]).
They showed that \ISM without ties (in the preference lists) is solvable in polynomial time by a simple reduction to finding a stable matching maximizing the weight of the included agent pairs (which is solvable in polynomial time~\cite{DBLP:journals/jcss/Feder92}). In contrast to this, \ISR is NP-complete \cite[Theorem~4.2]{DBLP:journals/disopt/CsehM16}, yet admits an FPT-algorithm for the parameter~$k$, that is, the maximum allowed size of 
the symmetric difference between the two matchings~\cite{DBLP:conf/aaai/BredereckCKLN20}.
With ties, Bredereck et al.~\cite{DBLP:conf/aaai/BredereckCKLN20} showed that \ISMT and \ISRT are NP-complete and W[1]-hard for $k$ even if the two preference profiles only differ in a single swap in some agent's preference list.
As \ISRT can be considered as a generalization of \ISMT, their results motivate us to focus on the NP-hard \ISR and \ISMT problems (which are somewhat incomparable problems).
Recently, Boehmer et al.~\cite{uschanged} followed up on the work of Bredereck et al.~\cite{DBLP:conf/aaai/BredereckCKLN20}, proving that different types of changes such as deleting agents or performing swaps of adjacent agents in some preference list are ``equivalent''.
Moreover, they introduced incremental variants of further stable matching problems and performed empirical~studies.

More broadly considered, matching under preferences problems in the presence of change are of high current interest in several application domains. Many such works fall into the category of ``dynamic matchings''~\cite{ALO20,baccara2020optimal,DBLP:journals/geb/DamianoL05,DBLP:journals/corr/abs-1906.11391,DBLP:journals/corr/abs-2007-03794}. 
However, different from our work, there the focus is typically  on adapting classic stability notions to dynamic settings while we rather aim for ``reestablishing'' 
(classic) stability at minimal change cost.

Closer to our work, there are several papers on adapting a given matching to change (while minimizing the number of reassignments): 
First, Gajulapalli et 
 	al.~\cite{DBLP:conf/fsttcs/GajulapalliLMV20} designed a polynomial-time (and incentive-compatible) algorithm for an incremental variant of the many-to-one version of \textsc{Stable Marriage} (known as \textsc{Hospital Residents}) where new agents are added. 
 	Second, Feigenbaum et al.~\cite{Feigenbaum17} considered an incremental variant of \textsc{Hospital 
 	Residents} where some agents may leave the 
 	system. They designed a ``fair'', Pareto-efficient, and strategy-proof 
algorithm for finding a matching before and after the change. 
The problems studied in both works are closest related to the polynomial-time solvable \ISM problem, which we do not study. 
Third, Bhattacharya et 
 	al.~\cite{DBLP:conf/icalp/BhattacharyaHHK15} studied one-to-one matching markets where 
 	agents are added and deleted over time and for some agents the set of 
 	acceptable partners may change. Their focus is on updating the matching in each 
 	step such that the number of reassignments is small while maintaining a 
 	small unpopularity factor. So in contrast to our work, they do not maintain that the matching is stable but (close to)~popular. 

Also motivated by temporally evolving preferences, several papers study the robustness of stable matchings subject to changing preferences~\cite{DBLP:conf/sagt/BoehmerBHN20,DBLP:journals/teco/ChenSS21,DBLP:conf/ijcai/Genc0OS17,DBLP:conf/aaai/Genc0OS17,DBLP:journals/tcs/GencSSO19,DBLP:conf/esa/MaiV18}.
By selecting a robust initial stable matching, one can increase the odds that it remains stable after some~changes. 

\subparagraph{Our Contributions.}

\begin{table}
	\begin{center}
	\resizebox{\textwidth}{!}{\begin{tabular}{ c|c|c } 
		
		 & \ISR & \ISMT \\  \hline
		$|\mathcal{P}_1\oplus \mathcal{P}_2|$ & W[1]-h. (Th. \ref{th:ISR-WP1P2}) \& XP (Th. \ref{th:ISR-XPP1P2})  & NP-h. for $|\mathcal{P}_1\oplus \mathcal{P}_2|=1$ (Th. \ref{co:ISMWties})   \\ \hline \hline
		\multirow{2}{*}{\#ties$+ k$} & \multirow{2}{*}{FPT wrt. $k$ (Th. 1 in \cite{DBLP:conf/aaai/BredereckCKLN20})} & W[1]-h. even if $|\mathcal{P}_1\oplus \mathcal{P}_2|=1$ (Th. \ref{co:ISMWties}) \\
		& & XP (even for parameter \#agents with ties) (Pr. \ref{pr:ISMNT})\\ \hline\hline
		\#outliers & FPT (Th. \ref{pr:ISRFPTO}) & ? \\  \hline 
		\#master lists 	& \multicolumn{2}{c}{W[1]-h. even for complete preferences (Th. \ref{thm:isr-master-lists}/\ref{pr:ISMWML})}  
	\end{tabular}}
\end{center}
\caption{Overview of our main results where each row contains results for one parameterization. Note that \ISM is polynomial-time solvable as proven by Bredereck et al.~\cite{DBLP:conf/aaai/BredereckCKLN20}.
} \label{table:ov}
\end{table}

 	Focusing on the two NP-hard problems \ISR and \ISMT, we significantly extend the work of  
 	Bredereck et al.~\cite{DBLP:conf/aaai/BredereckCKLN20} on incremental stable matchings.
 	In particular, we answer their two main open questions. Moreover, we strengthen several of their results. In addition, we analyze the impact of the degree of ``similarity'' between the agents' preference lists. Doing so, from a conceptual perspective, we complement work of Meeks and Rastegari~\cite{DBLP:journals/tcs/MeeksR20}.
 	They studied the influence of the number of agent types on the computational complexity of stable matching problems (two agents are of the same type if they have the same preferences and all other agents are indifferent between them).
 	By way of contrast, we consider the smaller so far unstudied parameter ``number of different~preference~lists''.

 	Next, we present a brief summary of the structure of the paper (for each section marking the main studied problem(s)) and our main contributions (see \Cref{table:ov} for an overview):
 	 
 	 \begin{description}
 		\item[\Cref{ISR} (\ISR)] Motivated by the observation that \ISMT is NP-hard even if just one swap has been performed, Bredereck et al.\ \cite{DBLP:conf/aaai/BredereckCKLN20} asked for 
 		the parameterized complexity 
 		of 
 		\ISR with respect to the  
 		difference $|\mathcal{P}_1 
 		\oplus \mathcal{P}_2|$ between the two given preference profiles $\mathcal{P}_1$ and~$\mathcal{P}_2$. We design and analyze an involved algorithm solving \ISR in polynomial time if $|\mathcal{P}_1 
 		\oplus \mathcal{P}_2|$ is constant (in other words, this is an XP-algorithm).
 		Our algorithm relies on the observation that if we know how certain
 		agents are matched in the matching to be constructed and we adapt the given matching
 		accordingly, then we can find an optimal solution by propagating these changes through
 		the matching until a new stable matching is reached; a general approach that might be of independent interest.
 		We complement this result by proving 
 		that \ISR parameterized by~$|\mathcal{P}_1 
 		\oplus \mathcal{P}_2|$ is W[1]-hard.
 		\item[\Cref{ISM} (\ISMT)] Bredereck et al.\ 
 		\cite{DBLP:conf/aaai/BredereckCKLN20} considered the total number of ties 
 		to be a promising parameter to potentially achieve 
 		fixed-parameter tractability results. We prove that this is 
        not the case as \ISMT is W[1]-hard with 
 		respect to $k$ \emph{plus} the total number of ties even if 
 		$|\mathcal{P}_1\oplus \mathcal{P}_2|=1$. Notably, this result  
 		strengthens  the W[1]-hardness with respect to 
 		$k$ for $|\mathcal{P}_1\oplus \mathcal{P}_2|=1$ of Bredereck et al.\ 
 		\cite{DBLP:conf/aaai/BredereckCKLN20} for \ISMT, while presenting a fundamentally different yet less technical proof. 
 		On the positive side, we devise an XP-algorithm for the number of agents with at least one tie in their preferences.
 		\item[\Cref{se:ML} (\ISR; \ISMT)] We study different cases where the agents have ``similar'' preferences.
 		For instance, we consider the case where all but $s$ agents have the same preference list (we call these $s$ agents \emph{outliers}), or the case where each agent has one out of only $p$ different \emph{master preference lists}.
 		We devise an algorithm that enumerates all stable matchings in an \textsc{SR} instance in FPT time with respect to $s$, implying an FPT algorithm for \ISR parameterized by~$s$.
 		In contrast to this and to a simple FPT algorithm for the number of agent types \cite{uschanged}, we prove that \ISR and \ISMT are W[1]-hard with respect to the number~$p$ of different preference lists.
 	 \end{description}

\section{Preliminaries}

The input of \textsc{Stable Roommates with Ties} (\SRT) is a set
 $A=\{a_1,\dots, a_{2n}\}$ of agents. 
 Each agent $a\in A$ has a 
 subset $\Ac(a)\subseteq A\setminus \{a\}$ of agents it \emph{accepts} and a
 preference relation~$\succsim_a$ which is a weak order over the agents 
 $\Ac(a)$.
 We assume that acceptance is symmetric, i.e., for two agents $a,a'\in A$, we have $a'\in \Ac(a)$ if and only if~$a\in \Ac(a')$.  
 The collection $\mathcal{P} = (\succsim_a)_{a\in A}$ of the preferences of all agents is called \emph{preference 
profile}.
 For two agents $a',a''\in \Ac(a)$, agent~$a$ \emph{weakly prefers} $a'$ to~$a''$ if $a'\succsim_a a''$.  If $a$ 
 weakly prefers $a'$ to $a''$ but does not weakly prefer~$a''$ to~$a'$, then $a$ 
 \emph{strictly prefers} $a'$ to~$a''$, and we write $a'\succ_a a''$. If $a$ 
weakly but not strictly prefers~$a'$ to~$a''$, then $a$ 
 is \emph{indifferent} between $a'$ and $a''$ and we write $a'\sim_a a''$; in other 
words, $a'$ and~$a''$ are \emph{tied}.
If~$a$ strictly prefers~$a'$ to~$a''$ or~$a' = a''$ holds, then we write~$a'\succeq_a a''$.
We 
say that an agent~$a$ has strict preferences, which we denote as $\succ_a$, 
if $\succsim_a$ is a strict order, and, in this case, we use the terms ``strictly 
prefer'' and ``prefer'' interchangeably. 
The \textsc{Stable Roommates} (\textsc{SR}) problem is the special case of \textsc{SR-T} where the preference relation of each agent is strict.

For two preference 
 relations $\succsim$ and~$\succsim'$ defined 
over the same set, the swap distance between~$\succsim$ and 
 $\succsim'$ is the number of agent pairs that are ordered differently by the two 
 relations, i.e., $|\{(a,b): a\succ b \wedge b\succsim' 
 a\}|+|\{(a,b): a\sim b \wedge \neg (a\sim'b)\}|$;
 for two preference relations over different sets, we define the swap distance 
to be infinite.
 For two preference profiles~$\mathcal{P}_1$ and~$\mathcal{P}_2$ containing the preferences of the same agents, 
$|\mathcal{P}_1\oplus 
 \mathcal{P}_2|$ denotes the sum over all agents~$a \in A$ of the swap distance of the two preference of~$a$ in $\mathcal{P}_1$ and $\mathcal{P}_2$.\footnote{Notably, by the equivalence theorem of Boehmer et al. \cite[Theorem 1]{uschanged}, all our results (except for \Cref{co:ISMWties} where the constant $|\mathcal{P}_1\oplus \mathcal{P}_2|$ increases slightly) still hold if $|\mathcal{P}_1\oplus 
 \mathcal{P}_2|$ instead denotes the number of agents whose preferences changed, the number of deleted agents, or the number of added agents. }
  
 A \emph{matching} $M$ is a set of pairs $\{a,a'\}$ with $a\neq a'\in A$, $a \in \Ac (a')$, and $a' \in \Ac (a)$, where 
each 
 agent appears in at most one pair. In a matching $M$, an agent $a$ 
 is \emph{matched} if $a$ is part of one pair from~$M$; otherwise, $a$ is 
 \emph{unmatched}. 
  A \emph{perfect matching} is a matching in which all agents are 
matched. 
  For a matching $M$ and an agent $a\in A$, we denote by 
$M(a)$ the 
  partner of~$a$ in~$M$, i.e., $M(a)=a'$ if $\{a,a'\}\in M$ and  
$M(a):=\square$ if $a$ is unmatched in $M$. All agents~$a\in A$ 
strictly prefer any agent from $\Ac(a)$ to being unmatched, i.e., 
$a'\succ_a\square$ for all~$a\in A$ and~$a'\in \Ac(a)$.  
  
  Two agents $a\neq a'\in A$ \emph{block} a 
  matching~$M$ if $a$ and~$a'$ accept each other and strictly prefer each other to their 
  partners in $M$, i.e., $a\in \Ac(a')$, $a'\in \Ac(a)$, $a'\succ_a M(a)$, and 
  $a\succ_{a'} M(a')$. A matching~$M$ is \emph{stable} if it is not 
  blocked by any 
  agent pair.\footnote{This definition of stability in the presence of ties is the by far most frequently studied variant known as \emph{weak stability}. \emph{Strong stability} 
  and \emph{super stability}
  are the two most popular alternatives. Notably, \ISMT (as defined later) becomes polynomial-time solvable for both strong and super stability, as for these two stability notions a stable matching maximizing a given weight function on all pairs of agents can be found in polynomial time~\cite{DBLP:journals/tcs/FleinerIM07,DBLP:conf/isaac/Kunysz18,DBLP:conf/soda/KunyszPG16}.
}
  An agent pair~$\{a,a'\}$ is called a \emph{stable pair} if there is a 
  stable matching $M$ with $\{a,a'\}\in M$. For two matchings~$M$ and~$M'$, we denote by~$M 
  \triangle M'$ the set of pairs that appear only in $M$ or only in $M'$, 
  i.e., $M \triangle M'=\{\{a,a'\}\mid \big(\{a,a'\}\in M \wedge \{a,a'\}\notin M' 
  \big) \vee \big(\{a,a'\}\notin M \wedge \{a,a'\}\in M' \big)\}$. 
  
  The incremental variant of \textsc{Stable Roommates [with Ties]} is defined 
as follows.
  \decprob{\textsc{Incremental Stable Roommates [with Ties]} (\ISR/[\ISRTwithoutspace])}{A 
  set $A$ of agents, two preference profiles $\mathcal{P}_1$ 
	and $\mathcal{P}_2$ containing the strict [weak] preferences of all agents, 
	a stable 
	matching~$M_1$ in $\mathcal{P}_1$, and 
	an integer~$k$.}{Is there a matching~$M_2$ that 
	is stable in $\mathcal{P}_2$ such that  $|M_1 \triangle 
	M_2| \le k$?}
  
  We also consider the incremental variant of \textsc{Stable Marriage (SM)}. 
Instances of \textsc{Stable Marriage} are instances of \textsc{Stable Roommates} 
where the set of agents is partitioned into two sets $U$ and $W$ such that agents from one of the 
sets 
only accept agents from the other set, i.e., $\Ac(m)\subseteq W$ for all $m\in 
U$ and $\Ac(w)\subseteq U$ for all $w\in W$. Following traditional conventions, 
we refer to the agents from~$U$ as men and to the 
agents from~$W$ as women.
All definitions from above still analogously apply to 
\textsc{Stable Marriage}. Thus, in \textsc{Incremental Stable Marriage 
[with Ties]} (\ISM/[\ISMTwithoutspace]), we are given a set~$A=U\cupdot W$ of 
agents and two preference profiles $\mathcal{P}_1$ 
	and $\mathcal{P}_2$ containing the strict [weak] preferences of all agents, 
where each~$m\in U$ accepts only agents from $W$ and the other way round.
The preferences of an 
agent $a\in A$ are \emph{complete} if $\Ac(a)=A\setminus \{a\}$. In 
\textsc{Stable Marriage}, the preferences of an agent~$a\in U 
\cupdot W$ are \emph{complete} if $\Ac(a)=W$ for $a\in U$ or if $\Ac(a)=U$ for 
$a\in W$. 
If the preferences of an agent are not complete, then they are 
\emph{incomplete}.

Note that for \textsc{SM} and \textsc{SM-T}, we know that a stable matching always exists and that we can check whether an \textsc{SR} instance admits a stable matching in $\mathcal{O}(|A|^2)$ time. 
Thus, in our algorithms for the incremental variants of these problems we always assume without loss of generality that there exists a stable matching in $\mathcal{P}_2$.
Moreover, the stable matchings in \textsc{SM} and \textsc{SR} instances have some structure:
If the preferences of agents in an \textsc{SM} or \textsc{SR} instance are complete, then all stable matchings are perfect. 
If the preferences are incomplete, then in both variants by the Rural Hospitals Theorem \cite{DBLP:books/daglib/0066875}, all stable matchings match the same set of agents. 
This in particular implies that for \ISR and \ISM instead of minimizing $|M_1\triangle M_2|$, we can alternatively minimize the number of agents that are matched differently in $M_1$ and $M_2$. 

 \subparagraph*{Parameterized Complexity Theory.}
A \emph{parameterized problem}~$L$ consists of a problem instance~$\mathcal{I}$ and a parameter value~$k\in \mathbb{N}$.
An \emph{XP-algorithm} for $L$ with respect to $k$ is an algorithm deciding $L$ in $ | \mathcal{I} |^{f(k)}$ time for some computable function $f$. 
A \emph{fixed-parameter tractable algorithm} for $L$ with respect to $k$ is an algorithm  deciding~$L$ in $f(k) \cdot \vert \mathcal{I} \vert ^{\mathcal{O}(1)}$ time for some computable function~$f$. The corresponding complexity classes are called XP, resp., FPT.
There is also a theory of hardness for parameterized problems. 
For our purposes, the central class here is W[1], where it holds FPT $\subseteq$ W[1] $\subseteq$ XP and it is commonly believed that both inclusions are strict. 
Thus, W[1]-hard problems are commonly believed not to be in FPT. 
To show that a problem is W[1]-hard, one typically constructs a parameterized reduction from a known W[1]-hard problem~$L'$.
Such a reduction from a parameterized problem~$L'$ to another parameterized problem~$L$ is a function that maps instances of $L'=(\mathcal{I}',k')$ to equivalent instances of $L=(\mathcal{I},k)$ with $k \le f(k')$ running in $f(k')\cdot  \vert \mathcal{I} \vert ^{\mathcal{O}(1)}$ time for some 
computable function $f$.
Finally, a parameterized problem~$L$ is \emph{para-NP-hard} if it is NP-hard even when the parameter is bounded by a constant.

\section{Incremental Stable Roommates Parameterized by~$|\mathcal{P}_1 \oplus 
 \mathcal{P}_2|$} \label{ISR}
Bredereck et al.~\cite{DBLP:conf/aaai/BredereckCKLN20} showed that \ISRT and \ISMT are  
NP-hard even if $\mathcal{P}_1$ and $\mathcal{P}_2$ differ only by a single 
swap.
While Bredereck et al.\ showed that \ISR (without ties) is NP-hard, they 
asked whether it is fixed-parameter tractable
parameterized by $|\mathcal{P}_1 \oplus \mathcal{P}_2|$.
We show that \ISR is W[1]-hard with respect to~$|\mathcal{P}_1 
\oplus \mathcal{P}_2|$ (\Cref{sec:swap-dist-ISR}), yet admits an intricate polynomial-time algorithm for constant  $|\mathcal{P}_1 
\oplus \mathcal{P}_2|$ (\Cref{sub:XP}), thus still clearly 
distinguishing it from the case with ties.

 \subsection{W[1]-Hardness}
\label{sec:swap-dist-ISR}

This section is devoted to proving that \ISR with respect to $|\mathcal{P}_1 
\oplus \mathcal{P}_2|$ is W[1]-hard:
\begin{restatable}{theorem}{WPP}
	\label{th:ISR-WP1P2}
	\ISR parameterized by $|\mathcal{P}_1 
	\oplus \mathcal{P}_2|$ is W[1]-hard.
\end{restatable}
To prove \Cref{th:ISR-WP1P2}, we reduce from the W[1]-hard \textsc{Multicolored Clique} problem 
parameterized 
by the solution size~$\ell$ 
\cite{DBLP:journals/jcss/Pietrzak03}. In \textsc{Multicolored Clique}, we are given an $\ell$-partite graph $G = (V^1 \cupdot V^2 \cupdot \dots \cupdot V^\ell, E)$  and the question is whether there is a clique~$X$ of size $\ell$ in~$G$ with $X\cap V^c\neq \emptyset$ for 
all $c\in [\ell]$.
To simplify notation, we assume that $V^c = \{v^c_1, \dots, v^c_\nu\}$ for all~$c\in [\ell]$ and that the given graph $G$ is $r$-regular for some~$r\in \mathbb{N}$.
We refer to the elements of~$[\ell]$ as \emph{colors} and say that a vertex $v$ has
color~$c\in [\ell]$ if~$v\in V^c$. 
The structure of the reduction is as follows.
For each color $c\in [\ell]$, there is a vertex-selection gadget, encoding 
which vertex from~$V^c$ is part of 
the multicolored clique.
Furthermore, there is one edge gadget for each edge.
Unless both endpoints of an edge are selected by the corresponding 
vertex-selection gadgets, the matching $M_2$ selected in the edge gadget 
contributes to the difference $M_1 \triangle M_2$ between $M_1$ and $M_2$.
We set $k$ (that is, the upper bound on~$|M_1 \triangle M_2|$) such that at least $\binom{\ell}{2}$ edges need to have both endpoints 
in 
the selected set of vertices, implying that the selected set of vertices forms a clique.

\subparagraph{Vertex-Selection Gadget.}
For each color~$c\in [\ell]$, we add a vertex selection gadget.
For each vertex $v^c_i\in V^c$, we add one 4-cycle consisting of agents $a_{i,1}^c$, $a_{i,2}^c$, $a_{i,3}^c$, and $a_{i,4}^c$. 
Further, in~$\mathcal{P}_2$, two agents $s^c$ and $\bar{s}^c$ are 
``added'' to the gadget (more specifically, $s^c$ and $\bar s^c$ are matched to 
dummy agents~$t^c$ and $\bar t^c$ 
in all stable matching in $\mathcal{P}_1$ but cannot be matched to~$t^c$ and $\bar t^c$ in a stable matching in~$\mathcal{P}_2$).
We construct the vertex-selection gadget such that the agents~$s^c$ and~$\bar{s}^c$ have to be matched to agents from the same  
4-cycle in a stable matching in~$\mathcal{P}_2$. This encodes the selection 
of the 
vertex corresponding to this 4-cycle to be part of the multicolored clique. 
Lastly, we add a second 4-cycle consisting of agents $\bar{a}_{i,1}^c$, $\bar{a}_{i,2}^c$, $\bar{a}_{i,3}^c$, and~$\bar{a}_{i,4}^c$ for each vertex~$v_i^c\in V^c$ to achieve that $M_1 \triangle M_2$ contains the same number of pairs from the vertex-selection gadget, independent of which vertex is selected to be part of the~clique. 
See \Cref{fig:vsg} for an example. 
\begin{figure}[t]
	
	\resizebox{\textwidth}{!}{\begin{tikzpicture}
		\node[vertex, label={[xshift=-0.05cm]180:$s^c$}] (si) at (0, 0) {};
		\node[vertex, label=0:$\bar s^c$] (sib) at ($(si) + (-6, 0)$) {};
		\node[vertex, label=90:$t^c$] (ti) at ($(si) + (1.5, 0)$) {};
		\node[vertex, label=90:$\bar t^c$] (tib) at ($(sib) + (-1.5, 0)$) {};
		\node[vertex, label=0:$u^c$] (ui) at ($(ti) + (1.5, 0)$) {};
		\node[vertex, label=180:$\bar u^c$] (uib) at ($(tib) + (-1.5, 0)$) {};
		
		\node (a11) at ($(si) + (3.5, 1.75)$) {};
		
		\node[vertex, label=0:$a^c_{2,1}$] (a21) at ($(a11) + ( -4, 0)$) {};
		\node[vertex, label=0:$a^c_{2,2}$] (a22) at ($(a21) + ( 0, 1.5)$) {};
		\node[vertex, label=180:$a^c_{2,3}$] (a23) at ($(a21) + (-1.5, 0)$) {};
		\node[vertex, label=180:$a^c_{2,4}$] (a24) at ($(a21) + (-1.5, 1.5)$) {};
		
		\node (dots) at ($0.5*(a21)+0.5*(a22) + (3,0)$) {\Huge \dots};
		
		\node[vertex, label=270:$a^c_{1,1}$] (a31) at ($(a21) + (-6, 0)$) {};
		\node[vertex, label=0:$a^c_{1,2}$] (a32) at ($(a31) + (0, 1.5)$) {};
		\node[vertex, label=180:$a^c_{1,3}$] (a33) at ($(a31) + (-1.5, 0)$) {};
		\node[vertex, label=180:$a^c_{1,4}$] (a34) at ($(a31) + (-1.5, 1.5)$) {};
		
		\node (ba11) at ($(si) + (3.5, -1.75)$) {};
		
		\node[vertex, label=0:$\bar a^c_{2,1}$] (ba21) at ($(ba11) + (-4, 0)$) {};
		\node[vertex, label=0:$\bar a^c_{2,2}$] (ba22) at ($(ba21) + (0, -1.5)$) {};
		\node[vertex, label=180:$\bar a^c_{2,3}$] (ba23) at ($(ba21) + (-1.5, 0)$) {};
		\node[vertex, label=180:$\bar a^c_{2,4}$] (ba24) at ($(ba21) + (-1.5, -1.5)$) {};
		
		\node[vertex, label=90:$\bar a^c_{1,1}$] (ba31) at ($(ba21) + (-6, 0)$) {};
		\node[vertex, label=0:$\bar a^c_{1,2}$] (ba32) at ($(ba31) + (0, -1.5)$) {};
		\node[vertex, label=180:$\bar a^c_{1,3}$] (ba33) at ($(ba31) + (-1.5, 0)$) {};
		\node[vertex, label=180:$\bar a^c_{1,4}$] (ba34) at ($(ba31) + (-1.5, -1.5)$) {};
		
		\node (dots) at ($0.5*(ba21)+0.5*(ba22) + (3,0)$) {\Huge \dots};
		
		\draw[bedge] (ti) edge node[pos=0.2, fill=white, inner sep=2pt] 
		{\color{red}\scriptsize $1$}  node[pos=0.76, fill=white, inner sep=2pt] 
		{\scriptsize $1$} (si);
		\draw[bedge] (tib) edge node[pos=0.2, fill=white, inner sep=2pt] 
		{\color{red}\scriptsize $1$}  node[pos=0.76, fill=white, inner sep=2pt] 
		{\scriptsize $1$} (sib);
		\draw (ti) edge node[pos=0.2, fill=white, inner sep=2pt] {\color{red}\scriptsize $2$}  node[pos=0.76, fill=white, inner sep=2pt] {\scriptsize $1$} (ui);
		\draw (tib) edge node[pos=0.2, fill=white, inner sep=2pt] {\color{red}\scriptsize $2$}  node[pos=0.76, fill=white, inner sep=2pt] {\scriptsize $1$} (uib);
		
		\draw (a21) edge node[pos=0.2, fill=white, inner sep=2pt] {\scriptsize $2$}  node[pos=0.8, fill=white, inner sep=2pt] {\scriptsize $r + 4$} (si);
		\draw (a24) edge node[pos=0.2, fill=white, inner sep=2pt] {\scriptsize $2$}  node[pos=0.85, fill=white, inner sep=2pt] {\scriptsize $x$} (sib);
		
		\draw (a31) edge[bend left = 15] node[pos=0.2, fill=white, inner sep=2pt] {\scriptsize $2$}  node[pos=0.7, fill=white, inner sep=2pt] {\scriptsize $2$} (si);
		\draw (a34) edge node[pos=0.2, fill=white, inner sep=2pt] {\scriptsize $2$}  node[pos=0.9, fill=white, inner sep=2pt] {\scriptsize $y$} (sib);
		
		\draw[bedge] (a21) edge node[pos=0.2, fill=white, inner sep=2pt] 
		{\scriptsize $1$}  node[pos=0.76, fill=white, inner sep=2pt] 
		{\scriptsize $2$} (a22);
		\draw (a21) edge node[pos=0.2, fill=white, inner sep=2pt] {\scriptsize $3$}  node[pos=0.76, fill=white, inner sep=2pt] {\scriptsize $1$} (a24);
		\draw (a23) edge node[pos=0.2, fill=white, inner sep=2pt] {\scriptsize $2$}  node[pos=0.76, fill=white, inner sep=2pt] {\scriptsize $1$} (a22);
		\draw[bedge] (a23) edge node[pos=0.2, fill=white, inner sep=2pt] 
		{\scriptsize $1$}  node[pos=0.76, fill=white, inner sep=2pt] 
		{\scriptsize $3$} (a24);
		
		\draw[bedge] (a31) edge node[pos=0.2, fill=white, inner sep=2pt] 
		{\scriptsize $1$}  node[pos=0.76, fill=white, inner sep=2pt] 
		{\scriptsize $2$} (a32);
		\draw (a31) edge node[pos=0.2, fill=white, inner sep=2pt] {\scriptsize $3$}  node[pos=0.76, fill=white, inner sep=2pt] {\scriptsize $1$} (a34);
		\draw (a33) edge node[pos=0.2, fill=white, inner sep=2pt] {\scriptsize $2$}  node[pos=0.76, fill=white, inner sep=2pt] {\scriptsize $1$} (a32);
		\draw[bedge] (a33) edge node[pos=0.2, fill=white, inner sep=2pt] 
		{\scriptsize $1$}  node[pos=0.76, fill=white, inner sep=2pt] 
		{\scriptsize $3$} (a34);
		
		\draw (ba21) edge node[pos=0.2, fill=white, inner sep=2pt] {\scriptsize $2$}  node[pos=0.8, fill=white, inner sep=2pt] {\scriptsize $r + 5$} (si);
		\draw (ba24) edge node[pos=0.2, fill=white, inner sep=2pt] {\scriptsize $2$}  node[pos=0.85, fill=white, inner sep=2pt] {\scriptsize $x+1$} (sib);
		
		\draw (ba31) edge[bend right = 15] node[pos=0.2, fill=white, inner sep=2pt] {\scriptsize $2$}  node[pos=0.7, fill=white, inner sep=2pt] {\scriptsize $3$} (si);
		\draw (ba34) edge node[pos=0.2, fill=white, inner sep=2pt] {\scriptsize $2$}  node[pos=0.9, fill=white, inner sep=2pt] {\scriptsize $y+1$} (sib);
		
		\node at ($(a34)+ (-3, 0)$) {$ x = (\nu -2) (r + 2) + 2$};
		\node at ($(a34)+ (-3, -1)$) {$ y = (\nu -1) (r + 2) + 2$};
		
		\draw (ba21) edge node[pos=0.2, fill=white, inner sep=2pt] {\scriptsize $1$}  node[pos=0.76, fill=white, inner sep=2pt] {\scriptsize $2$} (ba22);
		\draw[bedge] (ba21) edge node[pos=0.2, fill=white, inner sep=2pt] 
		{\scriptsize $3$}  node[pos=0.76, fill=white, inner sep=2pt] 
		{\scriptsize $1$} (ba24);
		\draw[bedge] (ba23) edge node[pos=0.2, fill=white, inner sep=2pt] 
		{\scriptsize $2$}  node[pos=0.76, fill=white, inner sep=2pt] 
		{\scriptsize $1$} (ba22);
		\draw (ba23) edge node[pos=0.2, fill=white, inner sep=2pt] {\scriptsize $1$}  node[pos=0.76, fill=white, inner sep=2pt] {\scriptsize $3$} (ba24);
		
		\draw (ba31) edge node[pos=0.2, fill=white, inner sep=2pt] {\scriptsize $1$}  node[pos=0.76, fill=white, inner sep=2pt] {\scriptsize $2$} (ba32);
		\draw[bedge] (ba31) edge node[pos=0.2, fill=white, inner sep=2pt] 
		{\scriptsize $3$}  node[pos=0.76, fill=white, inner sep=2pt] 
		{\scriptsize $1$} (ba34);
		\draw[bedge] (ba33) edge node[pos=0.2, fill=white, inner sep=2pt] 
		{\scriptsize $2$}  node[pos=0.76, fill=white, inner sep=2pt] 
		{\scriptsize $1$} (ba32);
		\draw (ba33) edge node[pos=0.2, fill=white, inner sep=2pt] {\scriptsize $1$}  node[pos=0.76, fill=white, inner sep=2pt] {\scriptsize $3$} (ba34);
		\end{tikzpicture}}
	
	\caption{An example for the vertex-selection gadget from \Cref{th:ISR-WP1P2}. 
		For an edge between two agents $a$ and $a'$, the number~$x$ closer to agent $a$ means that $a$ ranks~$a'$ at position~$x$, i.e., there are ${x -1}$~agents which $a$ prefers to~$a'$. 
		For example, the preferences of~$a_{1,4}^c$ are $a_{1,1}^c \succ \bar s^c \succ a_{1,3}^c$.
		The depicted preferences are those in $\mathcal{P}_1$. The preferences 
		in~$\mathcal{P}_2$ arise from swapping the red numbers.
		The matching 
		$M_1$ is marked in bold.
	}
	\label{fig:vsg}
\end{figure}
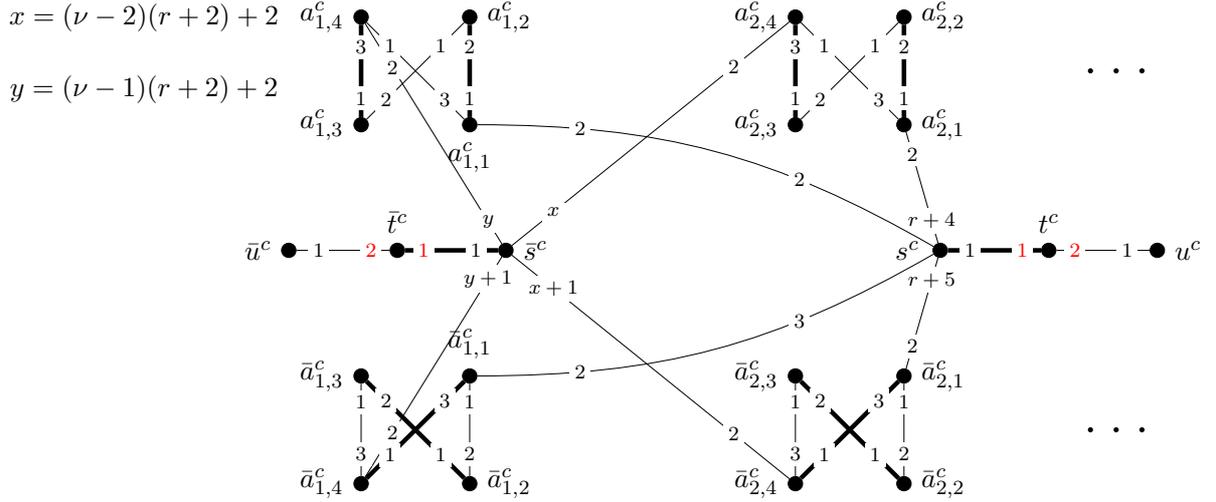

Apart from agents $s^c$ and $\bar s^c$, all agents from the vertex-selection gadget only find agents from the gadget acceptable, while $s^c$ and $\bar s^c$ also find agent~$a_{e, 1}$ (this agent will be introduced in the next paragraph ``Edge Gadget'') for each edge~$e$ incident to a vertex from~$V^c$ acceptable.
For 
$v_i^c \in V^c$, 
let 
$A_{\delta(v^c_i),1}$ denote the set of agents~$a_{e,1}$ such that~$e$ is an edge 
incident to $v^c_i$, i.e., $A_{\delta (v^c_i), 1 } := \{a_{e,1}\mid e\in E \wedge 
e\cap v^c_i \neq \emptyset \}$, and let $[A_{\delta(v^c_i),1}]$ denote an arbitrary strict order of~$A_{\delta (v^c_i), 1}$. 
For all $c\in [\ell]$ and $i\in [n]$, each vertex-selection gadget contains the following agents with the indicated preferences in $\mathcal{P}_1$: 
\begin{align*}
s^c & : t^c \succ a^c_{1, 1} \succ \bar{a}^c_{1, 1} \succ  
[A_{\delta(v^c_1),1}] \succ a^c_{2, 1}\succ \bar{a}^c_{2, 1} \succ 
[A_{\delta(v^c_2),1}]  \succ \dots \succ a^c_{\nu, 1}
 \succ \bar{a}^c_{\nu,1}\succ [A_{\delta(v^c_{\nu}),1}] \\
\bar{s}^c & : \bar{t}^c  \succ a^c_{\nu, 4} \succ \bar{a}^c_{\nu,4} \succ 
[A_{\delta(v^c_{\nu}),1}] \succ a^c_{\nu-1, 4}\succ \bar{a}^c_{\nu-1,4} \succ  
[A_{\delta(v^c_{\nu-1}),1}] 
\succ \dots \succ a^c_{1,4}
 \succ \bar{a}^c_{1, 4} 
\succ [A_{\delta(v^c_{1}),1}] \\
t^c & : s^c \succ u^c, \qquad \qquad \qquad \bar t^c  : \bar s^c \succ \bar 
u^c,
\qquad \qquad \qquad
u^c : t^c, \qquad \qquad \qquad
\bar u^c  : \bar t^c\\
a^c_{i, 1} & : a^c_{i, 2} \succ s^c \succ a^c_{i, 4}, \qquad a^c_{i, 2}  : a^c_{i, 3} \succ a^c_{i, 1}, \qquad a^c_{i, 3} : a^c_{i, 4} \succ a^c_{i, 2}, \qquad a^c_{i, 4}  : a^c_{i, 1} \succ \bar{s}^c \succ a^c_{i, 3} \\
\bar{a}^c_{i,1} & : \bar a^c_{i,2} \succ s^c \succ\bar a^c_{i,4}, \qquad \bar{a}^c_{i,2}  : \bar a^c_{i,3} \succ \bar a^c_{i,1}, \qquad \bar{a}^c_{i,3}  : \bar a^c_{i,4} \succ \bar a^c_{i,2}, \qquad \bar{a}^c_{i,4}  : \bar a^c_{i,1} \succ \bar s^c \succ \bar a^c_{i,3}
\end{align*}

In $\mathcal{P}_2$, only the preferences of agents $t^c$ and $t^{\bar{c}}$ change to $u^c\succ s^c$, respectively, $\bar{u}^c\succ \bar{s}^c$.
Notably, in each of the added 4-cycles, there exist two matchings of the four agents that are stable within the cycle in both~$\mathcal{P}_1$ and~$\mathcal{P}_2$ (i.e., $\bigl\{\{a^c_{i,1},a^c_{i,2}\},\{a^c_{i,3},a^c_{i,4}\}\bigr\}$ or $\bigl\{\{a^c_{i,1},a^c_{i,4}\},\{a^c_{i,2},a^c_{i,3}\}\bigr\}$ and $\bigl\{\{\bar{a}^c_{i,1},\bar{a}^c_{i,2}\},\{\bar{a}^c_{i,3},\bar{a}^c_{i,4}\}\bigr\}$ or $\bigl\{\{\bar{a}^c_{i,1},\bar{a}^c_{i,4}\},\{\bar{a}^c_{i,2},\bar{a}^c_{i,3}\}\bigr\}$  for $c\in [\ell]$ and~$i\in [\nu]$).
From each 4-cycle~$a^c_{i,1}$-$ a^c_{i,2}$-$ a^c_{i,3}$-$ a^c_{i,4}$-$ a^c_{i,1}$, matching~$M_1$ contains edges~$\{a^c_{i,1},a^c_{i,2}\}$ and~$\{a^c_{i,3},a^c_{i,4}\}$, and from each 4-cycle~$\bar a^c_{i,1}$-$ \bar a^c_{i,2}$-$ \bar a^c_{i,3}$-$ \bar a^c_{i,4}$-$ \bar a^c_{i,1}$, matching~$M_1$ contains edges~$\{\bar a^c_{i,1},\bar a^c_{i,4}\}$ and~$\{\bar a^c_{i,2}, \bar a^c_{i,3}\}$.

\subparagraph{Edge Gadget.}
For each edge $e = \{v_i^c, v^{\hat{c}}_j\}\in E$, we add an edge 
gadget.
This gadget consists of a 4-cycle with agents~$a_{e, 1}$, $a_{e, 2}$, $a_{e, 3}$, and $a_{e, 4}$, admitting two different 
matchings that are stable \emph{within} the gadget in both $\mathcal{P}_1$ and
$\mathcal{P}_2$.
The matching~$M_1$ contains $\{a_{e,1}, a_{e,4}\}$ and~$\{a_{e,3}, a_{e,2}\}$ 
in 
this 4-cycle and remains stable in~$M_2$ if all of~$s^c$, 
$\bar{s}^c$, $s^{\hat{c}}$, and~$\bar{s}^{\hat{c}}$ are matched at least 
as good as the respective agents corresponding to
$v_i^c$ and 
$v^{\hat{c}}_j$. This notably only happens if the vertex-selection 
gadgets of 
$V^c$ and 
$V^{\hat{c}}$ ``select'' the endpoints of $e$. Otherwise, the agents in this 
component 
need to be matched as  $\{a_{e,1}, a_{e,2}\}$ and $\{a_{e,3}, a_{e,4}\}$ in $M_2$, 
thereby contributing four pairs to $M_1 \triangle M_2$. For each~$e= \{v_i^c, v^{\hat{c}}_j\}\in E$, the 
agent's preferences are as follows:

\begin{align*}
a_{e, 1} & : a_{e, 2} \succ s^c \succ \bar s^c \succ s^{\hat{c}} \succ 
\bar 
s^{\hat{c}} 
\succ a_{e, 4}, \qquad &
a_{e, 2} & : a_{e, 3} \succ a_{e, 1}, \\
a_{e, 3}  &: a_{e, 4} \succ a_{e, 2}, \qquad &
a_{e, 4}  &: a_{e, 1} \succ a_{e, 3}.
\end{align*}

\subparagraph{The Reduction.}

To complete the description of the parameterized reduction, we
set $M_1 := \{ \{s^c, t^c\}, \{\bar s^c, \bar t^c\} \mid c \in [\ell]\} \cup 
\{ \{a^c_{i, 1}, a^c_{i, 2}\}, \{a^c_{i,3 } , a^c_{i,4}\}, \{\bar 
a^c_{i, 1}, \bar a^c_{i, 4}\}, \{\bar a^c_{i,3} , \bar a^c_{i,2}\} \mid c\in 
[\ell ], i\in [\nu]\}
\cup \{\{a_{e,1}, a_{e,4}\}, \{a_{e,3}, a_{e,2}\} \mid e\in E\}$ and $k 
:=  \ell \cdot (4 \nu + 5) + 4 (|E| - \binom{\ell}{2})$.

For the correctness of the reduction one can show that in $M_2$ for each $c\in [\ell]$  there is some~$i^*\in [\nu]$ such that the
matching $M_2$ contains pairs $\{s^c, a^c_{i^*, 1}\}$, $\{\bar{s}^c, a^{c}_{i^*, 
	4}\}$ (this corresponds to selecting vertex $v^c_{i^*}$ for color $c$). Then, the only agents $a_{e,1}$ for an edge~$e\in E$ incident to some vertex from $V^c$ that both $s^c$ and $\bar s^c$ do not prefer to their partner in~$M_2$ are those in~$A_{\delta(v^c_{i^*}),1}$.
This implies that for all edges $e=\{v_i^c, v^{\hat{c}}_j\}\in E$ with both endpoints selected we can match $a_{e,1}$ worse than $s^c$, $\bar s^c$, $s^{\hat{c}}$, and $\bar s^{\hat{c}}$. Thus, we can select $\{a_{e,1}, a_{e,4}\},\{a_{e,2}, a_{e,3}\}$ as in~$M_1$ in the respective edge gadget. In contrast, for all other edges we have to select the other matching in the edge gadget. To upper-bound the overall symmetric difference, one needs to further prove that for all $j<i^*$, matching~$M_2$ contains~$\{\{\bar a^c_{j, 1}, \bar a^c_{j, 2}\}, \{\bar a^c_{j,3}, \bar a^c_{j, 4}\}\}$, and that for all $j>i^*$, matching~$M_2$ contains~$\{\{ a^c_{j, 1}, a^c_{j, 4}\}, \{ a^c_{j,2}, a^c_{j, 3}\}\}$.
Thus, independent of the selected vertex, each vertex-selection gadget contributes~$4\nu + 5$ pairs to~$M_1 \triangle M_2$.

\subparagraph{Proof of Correctness.}

 We start the proof of correctness by showing that the constructed instance is indeed a feasible 
instance of \textsc{Incremental Stable Roommates}, i.e., that $M_1$ is 
stable in $\mathcal{P}_1$.

\begin{lemma}
	$M_1$ is stable in $\mathcal{P}_1$.
\end{lemma}

\begin{proof}
	Since $s^c$, $\bar s^c$, $t^c$, and $\bar t^c$ are matched to their first 
	choice, they are not part of a blocking pair.
	Thus, the only possible remaining blocking pairs are inside a 4-cycle 
	$a^c_{i, 1}$-$a^c_{i, 2}$-$a^c_{i,3}$-$a^c_{i, 4}$-$a^c_{i,1}$ or \mbox{$\bar a^c_{i, 
			1}$-$\bar a^c_{i, 2}$-$\bar a^c_{i,3}$-$\bar a^c_{i, 4}$-$\bar a^c_{i,1}$} for some $i\in 
	[\nu]$ and $c\in [\ell]$, or $a_{e, 1}$-$a_{e, 2}$-$a_{e,3}$-$a_{e, 4}$-$a_{e, 1}$ for 
	some $e\in E$.
	It is easy to verify that there is no blocking pair inside such a 4-cycle.
	Thus, $M_1$ is stable.
\end{proof}

 Next, we classify the stable matchings in $\mathcal{P}_2$.
 In order to simplify notation, we define for every vertex~$v^c_i$ with $i\in 
 [\nu]$ and $c\in [\ell]$,  
 matchings~$M^c_i := \{ \{a^c_{i,1 }, a^c_{i, 2}\},\allowbreak \{a^c_{i, 3}, a^c_{i, 4}\}, \allowbreak
 \{\bar a^c_{i,1 }, \bar a^c_{i, 2}\}, \allowbreak\{\bar a^c_{i, 3}, \bar a^c_{i, 4}\}\}$ 
 and $\bar M^c_i := \{ \{a^c_{i,1 }, a^c_{i, 4}\}, \{a^c_{i, 3}, a^c_{i, 2}\}, \allowbreak
 \{\bar a^c_{i,1 }, \bar a^c_{i, 4}\}, \allowbreak \{\bar a^c_{i, 3}, \bar a^c_{i, 2}\}\}$.

 \begin{lemma}\label{lem:stable-matchings}
	A matching $M$ is stable in $\mathcal{P}_2$ if and only if
	\begin{enumerate}
		\item for every $c\in [\ell]$, matching~$M$ contains $\{t^c, u^c\}$ and 
		$\{\bar t^c , \bar u ^c\}$,\label{item:tu}
		\item for every $c \in [\ell]$, there exists some $i\in [\nu]$ such that 
		matching $M$ contains pairs $\{s^c, a^c_{i, 1}\}$, $\{\bar{s}^c, a^{c}_{i, 
			4}\}$, $\{a^c_{i, 2}, a^c_{i,3}\}$, matching $M^c_j$ for all $j < i$, 
		matching $\bar M^c_j$ for $  j > i$, and  $\{\{\bar a^c_{i, 1}, 
		\bar a^c_{i, 2} \},\{\bar{a}^c_{i, 3}, \bar a^c_{i, 4}\} \}$ or 
		$\{\{\bar a^c_{i, 1}, \bar a^c_{i, 4}\},\{\bar a^c_{i, 3}, \bar 
		a^c_{i, 2}\}\}$, 
		\label{item:vertex-gadget}
		\item for every edge $e\in E$, matching $M$ contains pairs $\{a_{e,1}, 
		a_{e, 2}\}$ and $\{a_{e, 3}, a_{e, 4}\}$ or $\{a_{e, 1} , a_{e, 4}\}$ and 
		$\{a_{e, 3}, a_{e, 2}\}$, and
		\label{item:edge-gadget}
		\item for every edge $e = \{v^c_i, v^{\hat{c}}_j\}\in E$ with $\{s^c, 
		a^c_{i, 1}\} \notin M$ or $\{s^{\hat{c}}, a^{\hat{c}}_{j, 1}\} \notin 
		M$, we have that 
		$\{a_{e, 1}, a_{e, 2}\}$ and $\{a_{e, 3}, a_{e, 4}\}$ are contained in $M$.
		\label{item:consistency}
	\end{enumerate}
	
\end{lemma}
\begin{proof}
	We first show that every matching~$M $ fulfilling the above conditions is stable.
	So assume for a contradiction that matching~$M$ fulfills the above conditions but there exists a blocking pair~$\{a, b\}$.
	By \Cref{item:tu}, for every $c\in [\ell]$, agents~$t^c$, $\bar t^c$, $u^c$, 
	and $\bar u^c$ are 
	matched to their first choice and thus not part of a blocking pair.
	It is easy to check that for each~$c \in [\ell]$ and~${i \in [\nu]}$, none of 
	$\{a^c_{i, 1}, a^c_{i,2}\}, \{a^c_{i, 1}, a^c_{i,4}\}, \{a^c_{i, 3}, 
	a^c_{i,2}\}$, and $\{a^c_{i, 3}, a^c_{i,4}\}$ is blocking, and the same 
	holds for~$\{\bar a^c_{i, 1}, \bar a^c_{i,2}\}, \{\bar a^c_{i, 1}, \bar 
	a^c_{i,4}\}, \{\bar a^c_{i, 3}, \bar a^c_{i,2}\}$, and $\{\bar 
	a^c_{i, 3}, \bar a^c_{i,4}\}$. 
	
	We now check whether for some $c\in [\ell]$, agent~$s^c$ or $\bar{s}^c$ forms a 
	blocking pair together with an agent from the same vertex-selection gadget.
	By \Cref{item:vertex-gadget}, there exists some $i\in [\nu]$ such that $\{s^c, 
	a^c_{i, 1}\}\in M$ and $\{\bar{s}^c, 
	a^c_{i, 4}\}\in M$.
	For~$ j< i$, agent $a^c_{j, 1}$ prefers~$M(a^c_{j, 1}) = a^c_{j, 2}$ to 
	$s^c$, while for~${ j > i}$, agent $s^c$ prefers $M(s^c ) = a^c_{i, 1}$ to 
	$a^c_{j, 1}$.
	Thus, $\{s^c, a^c_{j, 1}\}$ is not blocking.
	Similarly, for each $c\in [\ell]$ and $j\in [\nu]$, $\{ s^c, 
	\bar a^c_{j, 1}\} $ is not blocking.
	For $j < i$, agent $\bar s^c$ prefers $M(\bar s^c) = a^c_{i, 4}$ to~$a^c_{j, 
		4}$, while for $j > i$, agent $a^c_{j, 4}$ prefers $M(a^c_{j, 4} ) = a^c_{j, 
		1}$ to $\bar s^c$.
	Thus, $\{\bar s^c, a^c_{j, 4}\}$ is not blocking.
	Similarly, for each $c\in [\ell]$ and $j\in [\nu]$, pair~$\{\bar s^c, \bar a^c_{j, 
		4}\}$ is not blocking.
	
	It is easy to verify that there is no blocking pair inside an edge gadget.
	Lastly, consider a pair~$\{s^c, a_{e, 1}\}$ for an edge $e = \{v^c_i, 
	v^{\hat{c}}_j\}$.
	If $\{a_{e,1}, a_{e, 2}\} \in M$, then $a_{e, 1}$ prefers $M(a_{e, 1}) = 
	a_{e, 2}$ to~$s^c$ and $\bar s^c$.
	Otherwise, \Cref{item:consistency} implies that $\{s^c, a^c_{i, 1}\} \in M$ 
	and by \Cref{item:vertex-gadget}  $\{ \bar s^c, a^c_{i,4}\}\in M$.
	Thus, $s^c$ prefers $M(s^c) = a^c_{i, 1}$ to $a_{e, 1}$ and $\bar s^c$ 
	prefers $M( \bar s^c) = a^c_{i, 4}$ to $a_{e, 1}$.
	Consequently, $\{s^c, a_{e, 1}\} $ and $\{\bar s^c, a_{e, 1}\} $ are not blocking.
	Therefore, $M$ is stable.
	
	We now show that every stable matching~$M$ fulfills the above conditions.
	Since for all~$c\in [\ell]$, agents~$t^c $ and $u^c$ as well as $\bar t^c $ and 
	$\bar u^c$ are their 
	mutual top choice, it follows that $M$ fulfills \Cref{item:tu}.
	Note that there clearly always exists a matching~$M^*$ fulfilling 
	\Cref{item:tu,item:vertex-gadget,item:edge-gadget,item:consistency}, and that 
	$M^*$ is perfect.
	As proven above, $M^*$ is stable, and by the Rural-Hospitals-Theorem for SR \cite{DBLP:books/daglib/0066875}, every stable 
	matching in $\mathcal{P}_2$ is perfect.
	It follows that no stable matching contains $\{s^c, a_{e, 1}\}$ or~$\{\bar s^c, a_{e, 1}\}$ for some~${c\in [\ell]}$ and~${e\in E}$, as otherwise 
	$a_{e,2}$ or $a_{e,4}$ remains 
	unmatched.
	As $M$ is perfect, it follows that $M$ fulfills \Cref{item:edge-gadget}.
	We now turn to \Cref{item:vertex-gadget}. Fix some~$c\in [\ell]$.
	Since $s^c$ is matched in~$M$, matching~$M$ contains 
	pair~$\{s^c, \bar a^c_{i,1}\}$ or 
	$\{s^c, a^c_{i, 1}\}$ for some $i \in [\nu]$.
	We first assume for a contradiction that $M$ contains $\{s^c ,\bar 
	a^c_{i,1}\}$ for some~$ i \in [\nu]$.
	As $M$ is perfect, this implies that $M$ also contains $\{\bar s^c, \bar 
	a^c_{i,4}\}$. We now distinguish two cases. 
	Assuming that $M$ contains $\{a^c_{i,1}, a^c_{i,2}\}$, then it also contains 
	$\{a^c_{i,3}, a^c_{i,4}\}$, and it follows that $\{ \bar s^c, a^c_{i, 4}\} $ 
	blocks $M$, a contradiction.
	Otherwise $M$ contains $\{a^c_{i,1}, a^c_{i,4}\}$, and $\{ s^c, a^c_{i,1}\}$ 
	blocks~$M$, a contradiction.
	Therefore, $\{s^c, a^c_{i,1}\}$ is part of $M$ for some $i\in [\nu]$.
	As $M$ is perfect, $M$ also contains $\{a^c_{i,2}, a^c_{i,3}\}$ and $\{ 
	\bar s^c, a^c_{i,4}\}$.
	Since, for every $j\in [\nu]\setminus \{i\}$, pairs~$\{s^c, a^c_{j, 1}\}$ and 
	$\{s^c, \bar 
	a^c_{j,1}\}$ are not blocking, 
	it follows that $M$ contains~$M^c_j$ for $j < i$.
	Since $\{\bar s^c , a^c_{j,4}\} $ and $\{ \bar s^c, \bar a^c_{j,4}\}$ are 
	not blocking, it follows that $M$ contains~$\bar M^c_j $ for $j < i$.
	Therefore, $M$ fulfills \Cref{item:vertex-gadget}.
	
	Assume for a contradiction that $M$ does not fulfill \Cref{item:consistency}.
	Then there exists some $c\in [\ell]$ and $i\in [\nu]$ such that $\{s^c, 
	a^c_{i,1}\} \in M$ and 
	$\{\bar s^c, a^c_{i,4}\} \in M$, and an edge $e = \{v^c_j, 
	v^{\hat{c}}_{j'}\}$ 
	with $i 
	\neq j$ such that $\{a_{e, 1}, a_{e, 4}\} \in M$.
	If $j < i$, then $\{s^c, a_{e, 1}\}$ blocks $M$.
	If $j > i$, then $\{ \bar s^c, a_{e, 1}\}$ blocks~$M$.
	Therefore, $M$ is not stable, a contradiction.
\end{proof}

We are now ready to prove the correctness of our reduction.

\WPP*
\begin{proof}
	In this section, we have described a reduction from 
	\textsc{Multicolored Clique} parameterized by solution size~$\ell$ to \ISR 
	parameterized by  $|\mathcal{P}_1 
	\oplus \mathcal{P}_2|$.
	This reduction clearly runs in polynomial time. Further, as the preference 
	profiles~$\mathcal{P}_1$ and $\mathcal{P}_2$ only differ  by a single swap 
	in the preference 
	profiles of $t^c $ and $\bar t^c$ for every $c\in [\ell]$, we have 
	$|\mathcal{P}_1 \oplus \mathcal{P}_2| = 2\ell$.
	It remains to show the correctness of the reduction.
	
	$(\Rightarrow):$
	Given a multicolored clique~$\{v^1_{h_1}, \dots, v^\ell_{h_\ell}\}$, we 
	construct a stable matching as follows.
	Set 
	\begin{align*}
	M_2:= &\{ \{t^c, u^c\},\{\bar t^c, \bar u^c\}, \{s^c, a^c_{h_c, 1}\}, \{\bar 
	s^c, a^c_{h_c, 4}\}, 
	\{a^c_{h_c,2}, 
	a^c_{h_c,3}\}, \{\bar a^c_{h_c, 1}, \bar a^c_{h_c, 4}\},\\
	&\{\bar a^c_{h_c, 
		3}, \bar 
	a^c_{h_c, 2}\} \mid c\in [\ell]\}\\ 
	&\cup \bigcup_{c \in [\ell], i\in [\nu]: i < h_c} M^c_i \cup \bigcup_{c \in [\ell], i > h_c} \bar M^c_i\cup \bigcup_{e = \{v^c_{h_c}, v^{\hat{c}}_{h_{\hat c}}\}\in E} \{\{a_{e, 1}, 
	a_{e, 4}\}, 
	\{a_{e, 3}, a_{e,2}\} \} \\
	&\cup \bigcup_{e \neq \{v^c_{h_c}, v^{\hat{c}}_{h_{\hat c}}\}\in E} \{\{a_{e, 1}, 
	a_{e,2}\}, 
	\{a_{e, 3}, a_{e,4}\}\}
	\end{align*}
	By \Cref{lem:stable-matchings}, $M_2$ is stable.
	Note that 
	\begin{align*}
	M_1 \triangle M_2 = &
	\bigl\{ \{a^c_{i, 1}, a^c_{i,2}\}, \{a^c_{i,1}, a^c_{i,4}\}, \{a^c_{i, 2}, 
	a^c_{i, 3}\}, \{a^c_{i, 3}, a^c_{i,4}\} \mid c\in [\ell], i > 
	h_c\bigr\}\\
	& \cup \bigl\{\{s^c, t^c\}, \{\bar s^c, \bar t^c\},  \{t^c, u^c\},\{\bar t^c, 
	\bar u^c\}, \{s^c, a^c_{h_c, 1}\}, 
	\{a^c_{h_c, 1}, a^c_{h_c, 2}\}, \{a^c_{h_c, 
		2}, a^c_{h_c, 3}\},\\ 
	&\{a^c_{h_c, 3}, a^c_{h_c, 4}\}, \{a^c_{h_c, 4}, \bar 
	s^c 
	\} \mid  c\in [\ell] \bigr\} \\
	&\cup \bigl\{ \{\bar a^c_{i, 1}, \bar a^c_{i, 2}\}, \{\bar a^c_{i,1}, \bar 
	a^c_{i,4}\}, \{ \bar a^c_{i, 2}, \bar a^c_{i,3}\}, \{\bar a^c_{i, 3}, \bar 
	a^c_{i,4}\} \mid  c\in [\ell], i < h_c\bigr\} \\
	&\cup \Bigl\{\{a_{e, 1}, a_{e, 2} \}, \{a_{e, 1}, a_{e, 4}\}, \{a_{e, 2}, 
	a_{e,3}\},\\
	&\{a_{e, 3}, a_{e, 4}\} \mid e\in E \setminus \bigr\{ \{v^c_{h_c}, 
	v^{\hat{c}}_{h_{\hat{c}}}\}\mid c,\hat{c} \in [\ell], c\neq 
	\hat{c} \bigr\}\Bigr\} 
	\end{align*}
	Thus, we have $|M_1 \triangle M_2| = \ell \cdot (4 \nu + 5) + 4 (|E| 
	-\binom{\ell}{2}) = k$.
	
	$(\Leftarrow):$
	Let $M_2$ be a stable matching in $\mathcal{P}_2$ with $|M_1 \triangle M_2 
	|\le k$.
	By \Cref{lem:stable-matchings}, it follows that the symmetric difference of 
	$M_1$ and $M_2$ in each vertex gadget is at least~$4\nu + 5$.
	Furthermore, for each color $c\in [\ell]$, there exists some $h_c$ such that 
	$\{s^c, a^c_{h_c, 1}\} \in M$.
	\Cref{lem:stable-matchings} implies that the symmetric difference of $M_1 $ 
	and $M_2$ in an edge gadget can only be smaller than four for edges~$e = 
	\{v^c_i, v^{\hat{c}}_j\}$ such that $\{s^c, a^c_{i, 1} \} \in M_2$ and 
	$\{s^{\hat{c}}, 
	a^{\hat{c}}_{j, 1}\}\in M_2$.
	Since $|M_1 \triangle M_2 | \le k$, it follows that the symmetric difference 
	of $M_1$ and $M_2$ in an edge gadget is smaller than four for at least~$\binom{\ell}{2}$ edge gadgets.
	Thus, $\{v^c_{h_c} \mid c\in [\ell]\}$ is a multicolored clique.
\end{proof}

\subsection{XP-Algorithm} \label{sub:XP} 

Complementing the above W[1]-hardness result, we now present an intricate XP-algorithm for \ISR parameterized by 
$|\mathcal{P}_1 \oplus \mathcal{P}_2|$, resulting in the following theorem: 
 \begin{restatable}{theorem}{thISRXPPP} \label{th:ISR-XPP1P2}
  \ISR can be solved in $\mathcal{O}(2^{4|\mathcal{P}_1 \oplus \mathcal{P}_2| 
   }\cdot n^{5|\mathcal{P}_1 \oplus \mathcal{P}_2| +3})$ time.
 \end{restatable}

As our algorithm and its proof of correctness are quite complicated and uses some novel ideas, we split the reminder of this subsection into two parts. 
In \Cref{subsub:XP-high-level}, we give an high-level overview of the algorithm and some intuitive explanations how and why it works. 
Subsequently, in \Cref{subsub:XP-details}, we present the algorithm in all its details and formally prove its correctness. Throughout this section, we assume that $M_1$ and $M_2$ are both perfect matchings (we will argue why we can assume this without loss of generality in \Cref{le:wlog2}). 

\subsubsection{High-Level Description} \label{subsub:XP-high-level}
Our algorithm works in two phases: 
In the initialization phase, we make some guesses how the stable matching $M_2$ looks like and accordingly change the original stable matching~$M_1$. 
These changes and guesses then impose certain constraints how good/bad some agents must be matched in $M_2$.
Subsequently, in the propagation phase, we locally resolve blocking pairs caused by the initial changes by ``propagating'' these constraints through the instance until a new stable matching is reached. This matching is then guaranteed to be as close as possible to the original stable matching.  
We believe that our technique to propagate changes through a matching is also of independent interest and might find applications elsewhere. 
In the following, we sketch the main ingredients of our algorithm.
We say that we \emph{update a matching~$M$ to contain a pair}~$e=\{a,b\}$ if we delete all pairs containing $a$ or $b$ from~$M$ and
add pair~$e$ to~$M$.

\subparagraph{Initialization Phase (First Part of Description).} 

Our algorithm maintains a matching~$M$.
At the beginning, we set~$M:=M_1$.
Before we change~$M$, we make some guesses on how the output matching $M_2$ 
shall look like.
These guesses are responsible for the exponential part of the running time (the rest of our algorithm runs in polynomial time).
The guesses result in some changes to $M$ and, for some agents $a\in A$, in a ``best case'' and ``worst case'' to which partner $a$ can be matched in~$M_2$. Consequently, we will store in $\bc (a)$ the \emph{best case} how $a$ may be matched in~$M_2$, i.e., the most-preferred (by $a$ in~$\mathcal{P}_2$) agent~$b$ for which we cannot exclude that $a $ is matched to~$b$ in a stable matching in $\mathcal{P}_2$ respecting 
 the guesses (in other words, this means that $a$ cannot be matched to a better partner that $\bc(a)$ in $M_2$).
 Similarly, $\wc (a)$ stores the \emph{worst case} to which $a$ can be matched.
 We initialize $\bc(a)=\wc(a)=\bot$ for all $a\in A$, encoding that we do not know a best or worst case yet.
 We will say that agent $a$ has a \emph{trivial} best case if $\bc (a) = \bot$ and a \emph{trivial} worst case if $\wc (a) = \bot$.
 
 To be more specific, among others, in the initialization phase we guess for each agent~$a\in A$ with 
 modified preferences as well as for~$M_1 (a)$ how they are matched in~$M_2$ 
 and update~$M$ to include the guessed pairs. Moreover, as an unmatched agent $a$ shall always have $\bc(a)\neq \bot$ or $\wc(a)\neq \bot$,  we guess for all agents $a$ that became unmatched by this whether they prefer~$M_1(a)$ to $M_2(a)$ (in which case we set $\bc(a):=M_1(a)$) or $M_2(a)$ to~$M_1(a)$ (in which case we set $\wc(a):=M_1(a)$).  
 Our algorithm also makes further guesses in the initialization phase.
 However, in order to understand the purpose of these additional guesses, it is helpful to first understand the subsequent propagation phase in some detail. 
 Thus, we postpone the description of the additional guesses to the end of this section.
 
\subparagraph{Propagation Phase.} 
After the initialization phase, blocking pairs for the current matching~$M$ 
force the algorithm to further change~$M$ and force a  propagation of best and worst cases through the instance until a stable matching is reached.
As our updates to $M$ are in some sense ``minimally invasive'' and exhaustive, once $M$ is stable in $\mathcal{P}_2$, it is guaranteed to be the stable matching in $\mathcal{P}_2$ which is closest to $M_1$ among all matchings respecting the initial guesses.
At  the core of the technique lies the simple observation that in an SR instance for each stable pair~$\{c,d\}$ and each stable matching~$N$ not including $\{c,d\}$ exactly one of~$c$ and~$d$ prefers the other to its partner in~$N$: 
\begin{restatable}[{\cite[Lemma 4.3.9]{DBLP:books/daglib/0066875}}]{lemma}{circularprefs}\label{lem:circular-prefs}
	Let~$N$ be a stable matching and $e = \{c, d\} \notin N$ be a stable pair in an SR instance.
	Then either $N(c)\succ_{c} d$ and $c\succ_{d} N(d)$ or 
	$d\succ_{c} N(c)$ and $N(d)\succ_{d} c$.
\end{restatable}
From this we can draw conclusions in the following spirit: 
Assuming that for a stable pair~$\{c, d\}$ in $\mathcal{P}_2$ we have that $\wc(c)\succ_c d$, i.e., $c$ is matched better than $d$ in $M_2$, it follows from \Cref{lem:circular-prefs} that $d$ is matched worse than $c$ in $M_2$, implying that we can safely set $\bc(d)=c$.

 \begin{algorithm}[t!]
\caption{Simplified propagation step performed for a pair $\{a,b\}$ of two matched agents blocking~$M$ with $\bc(b)=M(b)$ or for an unmatched agent $a$ with $\wc(a)\neq \bot$.}
 \begin{algorithmic}[1]
			\If{a is unmatched}
              {Let~$e=\{a, c\}$ be the stable pair in $\mathcal P_2$ such 
              that 
              $c \succ^{\mathcal{P}_2}_a \wc(a)$ and 
$\bc(c)\succ^{\mathcal{P}_2}_c a$ (or $\bc(c)=\bot$) and $c$ is 
worst-ranked by 
              $a$ 
              among 
              all such pairs.}
         \Else  
         { Let~$e=\{a, c\}$ be the stable pair in $\mathcal P_2$ such 
         	that 
         	$c \succ^{\mathcal{P}_2}_a M(a)$ and 
         	$\bc(c)\succ^{\mathcal{P}_2}_c a$ (or $\bc(c)=\bot$) and $c$ is 
         	worst-ranked by 
         	$a$ 
         	among 
         	all such pairs.}
         \EndIf
          \If{no such pair exists}
          {Reject this guess.}
          \Else 
          {~Update~$M$ such that it contains~$e$, set $\wc (a) := c$ and 
          	$\bc (c ) : = a$.}\label{l:1}
          \If{$M(a)\neq \square$}
          {$\bc (M(a)) : =  a$.}\label{l:2}
          \EndIf
          \If{$M(c)\neq \square$}
          {$\wc (M(c)) : = c$.}\label{l:3}
          \EndIf
          \EndIf
            
			\end{algorithmic}
			\label{alg:update}
			\end{algorithm}
 In the following, we will now explain simplified versions of some parts of the propagation phase, while leaving out others.
 
Assume for a moment that the current matching $M$ is perfect and that there is a blocking pair for $M$ in $\mathcal{P}_2$ (see \Cref{alg:update} for a pseudocode-description of the procedure described in the following). 
Because $M_1$ is stable in~$\mathcal{P}_1$, all pairs that currently block~$M_1$ either involve an agent with changed preferences or resulted from previous changes made to~$M$.
Using this, one can show that at least one of the two agents from a blocking pair~$\{a, b\}$, say~$b$, will have $a \succ_b \bc (b)=M(b)$.
Thus, we know that $b$ is matched worse than $a$ in any stable matching in~$\mathcal{P}_2$ respecting our current guesses. 
Accordingly, for $\{a,b\}$ not to block~$M_2$, agent~$a$ has to be matched to~$b$ or better and, in particular, better than $M(a)$ in the solution.
As a consequence, we update the worst case of~$a$ to be the next agent~$c$ which $a$ prefers to $M(a)$ such that~$\{a, c\}$ is a stable pair in $\mathcal{P}_2$, i.e., we set $\wc(a):=c$  (see Line \ref{l:1}).
Subsequently, we propagate this change further through the instance:
 Note that from \Cref{lem:circular-prefs} it follows that if $\{a',a''\}$ is a stable pair in $\mathcal{P}_2$ and agent~$a''$ is the worst possible partner of~$a'$ in a stable matching in $\mathcal{P}_2$ (or $a'$ prefers its worst possible partner to~$a''$), then agent~$a''$ cannot be matched better than agent~$a'$ in a stable matching in~$\mathcal{P}_2$.
 Thus, by setting $\wc(a'):=a''$ we also get $\bc(a''):=a'$. 
 Consequently, applying this to our previous update $\wc(a):=c$, we can also set~$\bc(c):=a$ (see Line \ref{l:1}). 
 Moreover, recall that by the definition of $c$, agent~$a$ prefers $c$ to $a$'s current partner $M(a)$ in~$M$. 
 Thus, assuming that in a stable matching $M^*$ in~$\mathcal{P}_2$ one of $a$ and $M(a)$ prefers the other to its partner in~$M^*$, whereas the other prefers its partner in~$M^*$, as $\wc(a)=c\succ_a M(a)$, we can conclude $\bc(M(a)):=a$ (see Line \ref{l:2}; we will discuss in the next part ``Initialization Phase (Second Part of Description)'' why this assumption can be made). 
 A similar reasoning applies to the update in Line~\ref{l:3}.
 
 So far, we assumed that all agents are matched.
 If there is an unmatched agent~$a$ involved in a blocking pair, then one can show that it cannot be matched to~$\bc (a)$ or $\wc (a)$.
 Thus, as each stable matching for~$\mathcal{P}_2$ is perfect, if $\wc (a) \neq \bot$, then we match $a$ to the next-better agent~$c$ before~$\wc (a)$ in its preferences such that $\{a, c\}$ is a stable pair in $\mathcal{P}_2$ and set $\wc(a)=c$. 
 Subsequently, we propagate this change as in the  above described case of a blocking pair (see \Cref{alg:update}).
 Otherwise, we have $\bc (a) \neq \bot$ and we match $a$ to the next-worse agent~$c$ after~$\bc (a)$ in the preferences of~$a$ such that $\{a, c\}$ is a stable pair in $\mathcal{P}_2$.
 Here, a slightly more complicated subsequent propagation step is needed (as described in the next section).  
 
Repeating these steps, i.e., matching so far unmatched agents and resolving blocking pairs, eventually either results in 
a conflict (i.e., 
an agent preferring its worst case to its best case, or changing a pair which we guessed to be part of~$M_2$) or in a stable matching.
In the first case, we conclude that no stable matching obeying our guesses exists, while in the latter case, we found a stable matching obeying the guesses and minimizing the symmetric difference to~$M_1$ among all such matchings.
The reason for the optimality of this matching is that every matching obeying our initial guesses has to obey the best and worst cases at the termination of the algorithm, and the computed matching~$M$ contains all pairs from~$M_1$ which comply with the best and worst cases.

\subparagraph{Initialization Phase (Second Part of Description).}  In addition to the guesses described above, the algorithm guesses the set~$F$ of pairs from $M_1$ for which both 
endpoints prefer~$M_2$ to $M_1$.
Similarly, the algorithm guesses the set~$H$ of pairs from~$M_2$ for which both endpoints prefer~$M_1$ to $M_2$.
Notably, one can prove that the cardinality of both~$F$ and~$H$ can be upper-bounded by $|\mathcal{P}_1 \oplus \mathcal{P}_2|$, ensuring XP-running time.
The reason why we need to guess the set~$F$ is that pairs~$\{a, b\}$ from~$M_1$ may not be stable pairs in~$\mathcal{P}_2$.
In this case, if $\{a,b\}\in M$ (which initially holds) and we know that $a$ prefers~$M_2(a)$ to~$b$ it does not follow that $b$ prefers~$a$ to~$M_2(b)$.
Thus, if we would treat the pairs from~$F$ as ``normal'' pairs, we would propagate an incorrect best case in Line~\ref{l:2}.
Note that all pairs from~$M\setminus M_1$  are stable pairs in~$\mathcal{P}_2$, as throughout the algorithm, we only add pairs that are stable in~$\mathcal{P}_2$ (see Line~\ref{l:1}). 
The reason why we need to guess the set~$H$ is more subtle but also due to fact that pairs from $H$ might cause problems for our propagation step. 
To incorporate our guesses, in the initialization phase, for each $\{a, b\}\in F$, we delete~$\{a, b\}$ from~$M$ and set~$\wc (a) = b$ and $\wc (b) = a$, while for each~$\{a, b\} \in H$ we update~$M$ to include~$H$.
We remark that from the proof of \Cref{th:ISR-WP1P2} it follows that \ISR is NP-complete even if we know for each agent~$a$ whose preferences changed as well as~$M_1 (a)$ how they are matched in~$M_2$ and the set of pairs~$F\subseteq M_1$ for which both endpoints prefer~$M_2$ to~$M_1$.
This indicates that guessing the set~$H$ might be necessary for an XP-algorithm.

\subsubsection{Formal Description and Proof of Correctness} \label{subsub:XP-details}
In the following, we start by giving a formal description of the XP algorithm in Section \ref{sec:formal_desc}. 
Subsequently, in Section \ref{sec:perfect}, we show that we can indeed assume that both $M_1$ and $M_2$ are perfect matchings. 
Next, in Section \ref{sec:intial_obs}, we prove some useful observations on the best and worst cases of agents, how they relate to the current matching $M$ and to the matching $M_1$, and how they change during the execution of the algorithm. 
These observations will be crucial when proving the correctness of the algorithm. 
In Section \ref{sec:defined}, we continue by establishing that the algorithm is well-defined.
Moreover, in Section \ref{sec:propagate} we prove that we propagate the best and worst cases of agents correctly, i.e., a matching that obeys our initial guesses also obeys the best and worst cases of all agents throughout the execution of the algorithm. 
We conclude by proving in Section \ref{sec:bounding} that the number of our initial guesses that we need to make is bounded in $n^{\mathcal{O}(|\mathcal{P}_1\oplus \mathcal{P}_2|)}$ and finally combine all parts of the proof in Section \ref{sec:putting}. 

 \paragraph{Formal Description of Algorithm} \label{sec:formal_desc}
 We assume that in the given \ISR instance matchings $M_1$ and $M_2$ are perfect (we later argue in Section \ref{sec:perfect} how we can preprocess instances to achieve this).
 
 A pseudocode description of our XP-algorithm for \ISR parameterized by $|\mathcal{P}_1 \oplus \mathcal{P}_2|$ can be found in \Cref{alg:XP,alg:init-XP,alg:updateapx}.
 \Cref{alg:XP} contains the main part of the algorithm while \Cref{alg:init-XP} concerns the initialization and \Cref{alg:updateapx} the propagation step. 
 
 We start the algorithm with a call of the \textsc{Initialization}$(\cdot)$ function described in \Cref{alg:init-XP} in which we initialize the matching $M$, the best and worst cases of agents, and make some guesses as described in \Cref{subsub:XP-high-level}. 
 Next, as long as the current matching $M$ is blocked by some pair, we further update the matching and agent's best and worst cases. 
 For this, we make a case distinction based on whether there is a blocking pair involving an unmatched agent $a$. 
 In this case we call the \textsc{Propagate}$(\cdot)$ function on $a$.
 In the case that all agents that are part of a blocking pair are matched, we pick an arbitrary blocking pair $\{a,b\}$, do some preprocessing steps and call the \textsc{Propagate}$(\cdot)$ function on the agent $x\in \{a,b\}$ with $\bc (x) = \bot$ and $\bc (y) \neq \bot$ with $y\in \{a,b\}\setminus \{x\}$ (we show in Section \ref{sec:defined} that such an agent exists in every blocking pair).
 Lastly, in \Cref{alg:updateapx}, we describe the \textsc{Propagate}$(\cdot)$ function called on some agent $a$.
 Here we make a case distinction based on whether $\bc(a)=\bot$ or $\bc(a)\neq \bot$.
 In both cases, we search for a stable pair $\{a,b\}$ that respects the current best and worst cases, initial guesses, and the constraints imposed by the blocking pair and that is worst ranked among these pairs.
 Subsequently, we update the best or worst cases of $a$ and $b$ and under certain constraints also of $M(a)$ and $M(b)$. 
 Finally, we update $M$ to include $\{a,b\}$. ¸
 
 	\begin{algorithm}[t!]
		\caption{Algorithm for \textsc{Incremental SR}}\label{alg:XP}
		\begin{algorithmic}[1]
			\Input{An \ISR instance $\mathcal{I}=(A,\mathcal{P}_1, \mathcal{P}_2,M_1,k)$ where $M_1$ is a perfect stable matching in $\mathcal{P}_1$ and there is a perfect stable matching in $\mathcal{P}_2$.}
			\Output{A stable matching~$M_2$ in $\mathcal{P}_2$ with $|M_1 \triangle M_2| \le k$ if one exists.}
			
			\State \textsc{Initialization}()
			\While{there exists a blocking pair 
			for $M$ in $\mathcal{P}_2$ \label{line:while}}
			\If{there is an agent~$a$ not matched by $M$ that is part of a 
			blocking pair for $M$
			}
              \textsc{Propagate}($a$). \Comment{see \Cref{alg:updateapx}}
            \Else
            
            \State{Let~$\{a, b\}$ be a blocking pair for $M$ in 
            $\mathcal{P}_2$.\label{line:selectbp}}
            
            \If{$\bc (a) \neq \bot $ and $\bc (b) \neq \bot$ 
            \label{line:reject-bp}}
              Reject this guess.
            \EndIf
            
            \State{Let $\{x, y\} =  \{a,b\}$ such that $\bc (x) = \bot$
            and $\bc (y) \neq \bot$.
            \label{line:choose-a}}
            \If{$\wc (x) = \bot$}
              Set $\wc (x): = M(x)$. \label{line:wcupdate}
            \EndIf
            
            \State \textsc{Propagate}($x$). \Comment{see \Cref{alg:updateapx}}
            \EndIf
            
            \If{there exists an agent~$a$ with $\wc(a) \succ^{\mathcal{P}_2}_a \bc (a)$}
              {Reject this guess.} \label{line:reject-bcwc}
            \EndIf
            
			\EndWhile
			
		\IfThenElse{$|M \triangle M_1| \le k$}{Return~$M$.}{Reject this 
		guess.\label{line:reject}}
			\end{algorithmic}
		\label{alg} 
	\end{algorithm}
	
 	\begin{algorithm}[t!]
		\caption{Function \textsc{Initialization} used in \Cref{alg}}\label{alg:init-XP}
		\begin{algorithmic}[1]
		\makeatletter
\setcounter{ALG@line}{11}
\makeatother
					
			\State Set $M := M_1$.  \label{line:init}
			\State Set $\bc (a) : = \bot$ and $\wc (a):=\bot$ for every 
agent~$a$.\label{line:initcs}
\bigskip
			\State Let $B$ be the set of agents with modified preferences and 
their partners in $M_1$. Guess for each agent~$a\in B$ how 
$a$ is matched in $M_2$, and update $M$ such that it contains the 
guessed pairs.\label{line:guess}

\State Guess a set $H\subseteq A\times A$ of up to  $| \mathcal{P}_1 \oplus 
			\mathcal{P}_2| $ pairs and update $M$ to include $H$. \label{line:guessH}
\State \Comment{pairs from $M_2$ where both endpoints prefer 
			$M_1$ to~$M_2$ in $\mathcal{P}_2$.}
			\State Let $X$ be the set of all agents from $B$ and their guessed 
partners and all agents that are part of a pair from $H$. 
			\State Set $\bc (a): = M(a)$ and $\wc (a) := M(a)$ for every 
agent~$a\in 
X$.\label{line:guesswcbc}
			\If{we guessed that two agents from $X$ have the same partner} 
Reject this guess. \EndIf
			\For{every $a\in X$ with $M_1 (a) \notin X$}
              \State Guess whether $M_1 (a)$ prefers~$M_2 (a)$ to $a$ or 
              prefers $a$ to $M_2 (a)$ in $\mathcal{P}_2$. \label{line:init-guessPart}
              
              \State In the former case, set $\wc (M_1(a)) := a$, while in the 
              latter case, set $\bc (M_1(a)) : = a$. \label{line:init-guess}
			
			\EndFor
			
			\bigskip
			
			\State Guess a set~$F$ of up to $| \mathcal{P}_1 \oplus 
			\mathcal{P}_2| $ pairs from~$M_1$.\label{line:guess-F} \State \Comment{pairs for which both endpoints prefer 
			$M_2$ to~$M_1$  in $\mathcal{P}_2$.}
			\If{there is a pair $\{a,M_1(a)\}\in F$ with $a\in X$ with~$ M_1 (a) \succ_a^{\mathcal{P}_2} \bc(a)$
} Reject this guess. \EndIf
			\For{every $e = \{a, b\} \in F$}
			 Set $\wc (a ) = b$, set $\wc (b) = a$ and delete $e$ from $M$. 
\label{line:endGuesses}
			\EndFor
			\end{algorithmic}
		\label{alg-init} 
	\end{algorithm}
	
	\begin{algorithm}[t!]
    
		\caption{Function \textsc{Propagate} used in \Cref{alg}}
		\label{alg:updateapx} 
    \begin{algorithmic}[1]
			\Input{An agent~$a$.}
\makeatletter
\setcounter{ALG@line}{27}
\makeatother
			\If{$\bc (a) \neq \bot$} \Comment{$a$ prefers $M_1(a)$ to 
			$M_2(a)$ in $\mathcal{P}_2$.}
              \State{Let~$e=\{a, b\}$ be the stable pair in $\mathcal P_2$ such 
              that 
              $ \bc (a) \succ^{\mathcal{P}_2}_a 
              b$, $a\succeq^{\mathcal{P}_2}_b \wc(b)$ (or $\wc(b)=\bot$), 
and $a\succeq^{\mathcal{P}_2}_b M_1(b)$ 
and $b$ is best-ranked 
              by $a$ 
              among 
              all such pairs.\label{line:def-b}}
          	  \If{no such pair exists or $a\in X$ or $b\in X$}
          	  {Reject this guess.}\label{line:rej}
          	  \Else
            \If{$\wc (M(a))=\bot$ and $\{a,M(a)\}\in M_1$}
              {Set $\wc (M(a)) : =  a$.} \label{line:resetwcbc}
            \EndIf
            \If{$\bc (M(b))=\bot$ and $\{b,M(b)\}\in M_1$}
              {Set  $\bc (M(b)) : = b$.} \label{line:resetwcbc2}
            \EndIf
          	  \State{Update~$M$ such that it contains~$e$, set $\bc (a) := b$ 
          	  	and $\wc (b ) : = a$. \label{line:updateM1}}
              \EndIf
            \Else \Comment{$a$ prefers $M_2(a)$ to 
            	$M_1(a)$ in $\mathcal{P}_2$.}
              \State{Let~$e=\{a, b\}$ be the stable pair in $\mathcal P_2$ such 
              that 
              $b \succ^{\mathcal{P}_2}_a \wc (a)$ and 
$\bc(b)\succeq^{\mathcal{P}_2}_b a$ (or $\bc(b)=\bot$) and $b $ is 
worst-ranked by 
              $a$ 
              among 
              all such pairs.\label{line:def-b-2}}
          \If{no such pair exists or $a\in X$ or $b\in X$}
          {Reject this guess.}\label{line:rej2}
          \Else 
          \If{$\bc (M(a))=\bot$ and $\{a,M(a)\}\in M_1$}
              {Set $\bc (M(a)) : =  a$.} \label{line:wresetwcbc}
            \EndIf
            \If{$\wc (M(b))=\bot$ and $\{b,M(b)\}\in M_1$}
              {Set  $\wc (M(b)) : = b$.} \label{line:wresetwcbc2}
            \EndIf
          \State{Update~$M$ such that it contains~$e$, set $\wc (a) := b$ and 
          	$\bc (b ) : = a$. \label{line:updateM2}}
\label{line:endUpdate}
          \EndIf
            \EndIf
            
			\end{algorithmic}
	\end{algorithm}
	
	\paragraph{Initial Assumptions on $M_1$ and $M_2$}
	\label{sec:perfect}
    Note that our algorithm assumes that in the given \ISR instance $(A,\mathcal{P}_1,\mathcal{P}_2,M_1,k)$ matching~$M_1$ is a perfect matching and there is a perfect stable matching in $\mathcal{P}_2$. 
    We now show that we may indeed assume this by showing in two steps that each \ISR instance can be reduced in linear time to an \ISR instance satisfying these properties. 
    
    We first establish that we can assume that there is a perfect stable matching in $\mathcal{P}_2$. 
	\begin{lemma}\label{le:wlog1}
	 An instance $\mathcal{I}=(A,\mathcal{P}_1,\mathcal{P}_2,M_1,k)$ of \ISR can be transformed in linear time into an equivalent instance $\mathcal{I}'=(A',\mathcal{P}'_1,\mathcal{P}'_2,M'_1,k')$ of \ISR with a perfect stable matching in $\mathcal{P}'_2$ and $|\mathcal{P}_1 \oplus \mathcal{P}_2| = |\mathcal{P}_1' \oplus \mathcal{P}_2'|$.
	\end{lemma}
	\begin{proof}
        Recall that in each stable matching in an \textsc{SR} instance the same set of agents is matched by the Rural Hospitals Theorem for \textsc{SR} \cite{DBLP:books/daglib/0066875}.
        Note further that we may assume without loss of generality that there is a stable matching in $\mathcal{P}_2$, as we can check this in linear time~\cite{DBLP:journals/jal/Irving85} and otherwise can map $\mathcal{I}$ to a trivial no-instance.
        Let $A_1$ be the set of agents that are unmatched in a stable matching in $\mathcal{P}_1$ and $A_2$ be the set of agents that are unmatched in a stable matching in $\mathcal{P}_2$. 
        We modify $\mathcal{I}$ to arrive at $\mathcal{I}'$ as follows. 
        For each agent $a\in A_2$, we add an agent $a'$ which in both preference profiles only finds $a$ acceptable and append $a'$ at the end of the preferences of~$a$ in both preference profiles.
        Moreover, for each $a\in A_1\cap A_2$, we add to matching $M_1$ the pair $\{a,a'\}$, i.e., $M_1' := M_1 \cup \{ \{a, a'\} : a \in A_1 \cap A_2\}$. 
        Lastly, we set $k'$ to $k+|A_2\setminus A_1|$.
        
        First note that $M'_1$ is a stable matching in $\mathcal{P}'_1$, as each agent $a\in A_2$ is either matched to~$a'$ or prefers its partner $M'_1(a)$ to $a'$. 
       Moreover, there exists a perfect stable matching in~$\mathcal{P}'_2$. 
       Let $M$ be a stable matching in $\mathcal{P}_2$ (we know that such a matching needs to exist). 
       Then $M\cup \{\{a,a'\}\mid a\in A_2\}$ is a perfect stable matching in $\mathcal{P}'_2$. 
        
        It remains to argue that $\mathcal{I}$ is a yes-instance if and only if $\mathcal{I}'$ is a yes-instance.
        Assume that $M_2$ is a solution for $\mathcal{I}$. 
        Then $M'_2=M_2\cup \{\{a,a'\}\mid a\in A_2\}$ is a stable matching in $\mathcal{P}'_2$. 
        Moreover we have that $M'_1\triangle M'_2=M_1\triangle M_2 \cup \{\{a,a'\}\mid a\in A_2\setminus A_1\}$ and thus $|M'_1\triangle M'_2|= k + |A_2\setminus A_1|$.
        
        For the other direction, assume that $M'_2$ is a solution for $\mathcal{I}'$. 
        First of all note that the matching $M_2$ resulting from $M'_2$ after removing all agent pairs including agents from $\{a' \mid a\in A_2\}$ is stable in $\mathcal{P}_2$, as all agents $a\in A$ prefer the same set of agents to $M_2(a)$ in $\mathcal{P}_2$ as they prefer to $M'_2(a)$ in $\mathcal{P}'_2$. 
        To bound the symmetric difference note that no agent $a\in A_2$ can be unmatched in $M'_2$, as otherwise it forms a blocking pair with agent $a'$. 
        As no agent from $A_2$ is matched in $M_2$ from this it follows that $\{\{a,a'\}\mid a\in A_2\}\subseteq M'_2$.
        Consequently, we have that $M'_1\triangle M'_2=M_1\triangle M_2 \cup \{\{a,a'\}\mid a\in A_2\setminus A_1\}$ and thus $|M_1\triangle M_2|=|M'_1\triangle M'_2|-|A_2\setminus A_1|\leq k´$.
	\end{proof}

	Using similar arguments as in the previous lemma, we now show that we can additionally assume that the given stable matching $M_1$ in $\mathcal{P}_1$ is perfect:
    \begin{lemma}\label{le:wlog2}
	 An instance $\mathcal{I}=(A,\mathcal{P}_1,\mathcal{P}_2,M_1,k)$ of \ISR with a perfect stable matching in $\mathcal{P}_2$ can be reduced transformed  linear time into an equivalent instance $\mathcal{I}'=(A',\mathcal{P}'_1,\mathcal{P}'_2,M'_1,k')$ of \ISR with $|\mathcal{P}_1 \oplus \mathcal{P}_2| = |\mathcal{P}_1' \oplus \mathcal{P}_2'|$ where $M'_1$ is a perfect stable matching in $\mathcal{P}'_1$ and there is a perfect stable matching in $\mathcal{P}'_2$.
	\end{lemma}
	\begin{proof}
	 Recall that in each stable matching in an \textsc{SR} instance the same set of agents is matched by the Rural Hospitals Theorem for \textsc{SR} \cite{DBLP:books/daglib/0066875}.
    Let $A_1=\{a_1,\dots,a_x\}$ be the set of agents that are unmatched in a stable matching in $\mathcal{P}_1$. 
    Note that $x$ is even, as $|A|$ needs to be even for a perfect matching to exist in $\mathcal{P}_2$. 
    
        We modify $\mathcal{I}$ to arrive at $\mathcal{I}'$ as follows. 
        For each $i\in [x]$, we add a dummy agent $a'_i$ that in both preference profiles ranks $a_i$ in the first position and ranks in the second position $a'_{i+1}$ if $i$ is odd and $a'_{i-1}$ if $i$ is even.  
        Moreover, we append $a'_i$ at the end of the preferences of $a_i$ in both preference profiles.
        Moreover, for each $i\in [x]$, we add to matching $M_1$ the pair $\{a_i,a'_i\}$. 
        Lastly, we set $k'$ to $k+\frac{3x}{2}$.
        
        First note that $M'_1$ is clearly perfect and also a stable matching in $\mathcal{P}'_1$, as each dummy agent is matched to its top-choice and blocking pairs for $M'_1$ where both agents are from $A$ would also block $M_1$ in $\mathcal{P}_1$. 
        Let $M$ be a perfect stable matching in $\mathcal{P}_2$. 
        Then, $M\cup \{\{a'_{2i-1}, a'_{2i}\}\mid i\in [\frac{x}{2}]\}$ is clearly perfect and also a stable matching in $\mathcal{P}'_2$, as each agent from $A$ is matched to an agent it prefers to the newly added dummy agents and each dummy agent is matched to the only other dummy agent it finds acceptable. 
        
        It remains to argue that $\mathcal{I}$ is a yes instance if and only if $\mathcal{I}'$ is a yes instance.
        Assume that $M_2$ is a solution for $\mathcal{I}$. 
        Let $M'_2=M_2\cup \{\{a'_{2i-1}, a'_{2i}\}\mid i\in [\frac{x}{2}]\}$.
        As argued above $M'_2$ is stable in~$\mathcal{P}'_2$. 
        Moreover note that as $M'_1=M_1\cup \{ \{a'_i, a_i\} \mid i \in [x] \}$ and $M'_2=M_2\cup \{\{a'_{2i-1}, a'_{2i}\}\mid i\in [\frac{x}{2}]\}$, we have that  $M'_1\triangle M'_2=M_1\triangle M_2 \cup \{ \{a'_i, a_i\} \mid i \in [x] \} \cup \{\{a'_{2i-1}, a'_{2i}\}\mid i\in [\frac{x}{2}]\}$ and thus $|M'_1\triangle M'_2|\leq k + \frac{3x}{2}$.
        
        For the other direction, assume that $M'_2$ is a solution for $\mathcal{I}'$. 
        First of all note that the matching $M_2$ resulting from $M'_2$ after removing all agent pairs including a dummy agent is stable in $\mathcal{P}_2$, as all agents $a\in A$ prefer the same set of agents to $M_2(a)$ in $\mathcal{P}_2$ as they prefer to $M'_2(a)$ in~$\mathcal{P}'_2$. 
        To bound the symmetric difference note that all agents from $A$ must be matched in $M_2$, as $M_2$ is a stable matching in $\mathcal{P}_2$ and there is a perfect stable matching in $\mathcal{P}_2$.
        Thus in $M'_2$ no agent from $A$ can be matched to a dummy agent and thus $\{\{a'_{2i-1}, a'_{2i}\}\mid i\in [\frac{x}{2}]\}\subseteq M'_2$.
        Consequently, as above, we get that $M'_1\triangle M'_2=M_1\triangle M_2 \cup \{ \{a'_i, a_i\} \mid i \in [x] \} \cup \{\{a'_{2i-1}, a'_{2i}\}\mid i\in [\frac{x}{2}]\}$ and thus $|M_1\triangle M_2|=|M'_1\triangle M'_2|-\frac{3x}{2}\leq k´$.
	\end{proof}
	
	\paragraph{Initial observations on $\mathbf{\bc}$ and $\mathbf{\wc}$} \label{sec:intial_obs}
	
	We continue by making some  observations on $\mathbf{\bc}$ and $\mathbf{\wc}$, how they relate to one another and to $M_1$, and when they are set to a non-trivial value. 
	These observations will be vital in the reminder of the proof of correctness. 
	
	We start by showing that as soon as an agent is matched differently in 
$M$ than in~$M_1$, the agent has a non-trivial best or worst case. 

	\begin{lemma} \label{le:match-diff}
	 After calling the \textsc{Initialization}() function, for each agent $c\in A$, it always holds that if $M(c)\neq M_1(c)$, then 
$\bc(c)\neq \bot$ or $\wc(c) \neq \bot$. 
	\end{lemma}
	\begin{proof}
	For some agent $c\in A$, we only change $M (c)$ in Lines~\ref{line:guess} and \ref{line:guessH}, and Line~\ref{line:endGuesses} of \Cref{alg-init} and Line~\ref{line:updateM1} and Line~\ref{line:updateM2} of \Cref{alg:updateapx}.
	In each of these cases we either set $\wc (c) \neq \bot $ or $\bc (c) \neq \bot$. 
	For Lines~\ref{line:guess} and \ref{line:guessH} this happens in Lines~\ref{line:guesswcbc} and \ref{line:init-guess} of \Cref{alg-init}. 
	For Line \ref{line:endGuesses}, this happens in Line~\ref{line:endGuesses} of \Cref{alg-init}. 
	Moreover for Line~\ref{line:updateM1} this happens in Lines \ref{line:resetwcbc} to~\ref{line:updateM1}, and for  Line~\ref{line:updateM2} this happens in Lines~\ref{line:wresetwcbc} to~\ref{line:updateM2} from \Cref{alg:updateapx}.
	\end{proof}
From this it directly follows that for every agent for which the 
propagation function is called 
	the agent has a non-trivial best case or worst case: 
	
	\begin{lemma}\label{le:welldefined}
		Whenever $\textsc{Propagate}$ is called for an agent~$c\in A$, it follows 
that it holds that $\bc(c)\neq \bot$ or $\wc(c)\neq \bot$. 
	\end{lemma}
\begin{proof}
	If $c$ is currently unmatched by $M$, then 
	\Cref{le:match-diff} (together with our assumption that $M_1$ is perfect) implies that $\bc(c)\neq \bot$ or $\wc(c)\neq 
	\bot$. If $c$ is 
currently matched, then by Line
	\ref{line:wcupdate} of \Cref{alg}, it follows that $\wc(c)\neq \bot$.
\end{proof}

	We next show that over the course of the algorithm from the perspective 
of an agent~$c\in A$, $\bc(c)$ becomes worse and worse, while $\wc(c)$ becomes 
better and better. 
	
	\begin{lemma} \label{le:gettingbetter}
	Consider an agent $c\in A$ and two timesteps $\mathfrak{A}$ and $\mathfrak{B}$ of the execution 
of \Cref{alg:XP} where $\mathfrak{A}$ appears before $\mathfrak{B}$. If $\bc(c)\neq \bot$ 
at point $\mathfrak{A}$, then either 
$\bc(c)$ is the same at points $\mathfrak{A}$ and $\mathfrak{B}$ or 
$c$ prefers $\bc(c)$ at point $\mathfrak{A}$ to~$\bc(c)$ at 
point $\mathfrak{B}$. If $\wc(c)\neq \bot$ at point $\mathfrak{A}$, then  either 
$\wc(c)$ is the same at points $\mathfrak{A}$ and $\mathfrak{B}$ or $c$ prefers 
$\wc(c)$ at point $\mathfrak{B}$ to $\wc(c)$ at point~$\mathfrak{A}$.
	\end{lemma}
	\begin{proof}
	 The only points during the algorithm where $\bc(c)\neq \bot$ and 
$\bc(c)$ is changed or $\wc(c) \neq \bot$ and $\wc(c)$ is changed are  
in Line \ref{line:updateM1} or \ref{line:updateM2} of \Cref{alg:updateapx}.
By the definition of the 
pair $\{a,b\}$ from Line \ref{line:def-b} or Line \ref{line:def-b-2} of \Cref{alg:updateapx}, the lemma 
follows.
	\end{proof}

We now continue by making several general observation concerning the 
values of $\bc(\cdot)$ and $\wc(\cdot)$ of the different agents. We start with the crucial 
observation that if an agent~$c\in A\setminus X$ has a non-trivial best case~$\bc (c)$,
then $c$ prefers $M_1(c)$ to~$\bc (c)$ or $M_1(c)=\bc(c)$ and if an agent $c\in A$ has a non-trivial worst 
case~$\wc (c)$, then $c$ prefers~$\wc (c)$ to $M_1(c)$ or $M_1(c)=\wc(c)$:
\begin{lemma} \label{le:wcbcM1}
	For each agent $c\in A\setminus X$ at the beginning of each execution of 
the while 
loop in Line \ref{line:while} 
it holds that if $\bc(c)\neq \bot$ 
then 
	$M_1(c)\succeq_c^{\mathcal{P}_2}\bc(c)$ and if $\wc(c) \neq \bot$ then 
	$\wc(c)\succeq_c^{\mathcal{P}_2}M_1(c)$.
\end{lemma}

\begin{proof}
	Fix some $c\in A\setminus X$. At the beginning of the algorithm, the 
	statement is trivially fulfilled, as $\bc(c)=\wc(c)=\bot$. 
	If 
$\wc(c)$ or $\bc(c)$ is changed in 
	Line \ref{line:init-guess} or Line~\ref{line:endGuesses} of \Cref{alg:init-XP}, then it is set to~$M_1(c)$ so again the statement is fulfilled. 

	It 
	remains to consider changes made within the while-loop from Line \ref{line:while}. We prove that the 
statement holds at the beginning of the $i$th execution of the while-loop by 
induction over $i$. For~$i=0$, we have already argued that the statement holds. So assume 
that the statement holds for an arbitrary but fixed~$i\geq 0$.
We will argue that 
from 
this we can conclude that it also holds for $i+1$. 	
We denote by $\bc^{\before}$ and $\wc^{\before}$ the best and worst case after the $i$th iteration, but before the $i+1$th iteration.

The only point outside of the \textsc{Propagate} function where the best or worst case of an agent might be changed in the $i+1$th iteration of the while loop is in Line \ref{line:wcupdate} of \Cref{alg:XP}. 
However, note that here it holds that $\wc^{\before}(c)=\bot$ and $\bc^{\before}(c)=\bot$ 
and thus by \Cref{le:match-diff} that $M(c)=M_1(c)$. As we set 
$\wc(c)=M(c)=M_1(c)$, the statement clearly holds after the update.
	
	It remains to consider a call of $\textsc{Propagate}$ and let $\{a,b\}$ be the 
	examined pair (i.e., the pair defined in Line~\ref{line:def-b} or Line~\ref{line:def-b-2} of \Cref{alg:updateapx}).

First assume that $\bc^{\before}(a)\neq\bot$.
If $c=a$, then, using the induction 
hypothesis,
$M_1(a)\succeq_a^{\mathcal{P}_2}\bc^{\before}(a)$ follows. 
\Cref{le:gettingbetter} then implies that the statement also holds after the update.
If $c=b$, then we know by the definition of $b$ that $a\succeq_b^{\mathcal{P}_2} M_1(b)$. As we set  
$\wc(b)$ to $a$, after the update we have $\wc(b)\succeq_b^{\mathcal{P}_2} M_1(b)$.
If $c=M(a)$ or $c=M(b)$, then $\bc(c)$ and $\wc(c)$ either remains unchanged 
or is set to $M_1(c)$.

Now we turn to the remaining case where $\bc^{\before}(a)=\bot$. 
\Cref{le:welldefined} implies that $\wc^{\before}(a)\neq \bot$.
If $c=a$, then by the induction hypothesis  ${\wc^{\before}(a)\succeq_a^{\mathcal{P}_2}M_1(a)}$ follows.
\Cref{le:gettingbetter} then implies that $\wc(a)\succeq_a^{\mathcal{P}_2}M_1(a)$ also holds after the update.
If ${c=b}$, then by the induction hypothesis for agent~$a$ and the definition of $b$  it follows that 
${b \succ_a^{\mathcal{P}_2} \wc^{\before} (a) \succeq_a^{\mathcal{P}_2} M_1(a)}$.
Note that the preferences of $a$ and $b$ are the 
same in $\mathcal{P}_1$ and~$\mathcal{P}_2$ (otherwise the guess would have 
been rejected in Line~\ref{line:rej2} of \Cref{alg:updateapx}).
Thus, from $b\succ_a^{\mathcal{P}_2} M_1(a)$,
it follows that $M_1(b)\succ_b^{\mathcal{P}_2} a$, as otherwise $\{a,b\}$ is a 
blocking pair 
for $M_1$.
As we set $\bc(b)$ to $a$ it follows that after the update of $\bc(b)$, agent~$b$ prefers $M_1(b)$ to $\bc(b)$.
If $c=M(a)$ or $c=M(b)$, then $\bc(c)$ and $\wc(c)$ either remain unchanged 
or are set to $M_1(c)$.
\end{proof}

From the previous lemma, we can easily conclude that if an agent has a 
non-trivial best case and worst case, they need to be identical:
\begin{lemma}  \label{le:one-nontrivial}
	For each agent $c\in A $ at the beginning of each execution of 
	the 
	while-loop in Line~\ref{line:while} it holds that $\bc(b)=\wc(b)$ or $\bc(b)= 
\bot$ or $\wc(c)= \bot$.
\end{lemma}
\begin{proof}
For every agent from~$X$, we set $\bc $ and $\wc $ accordingly in Line~\ref{line:guesswcbc} of \Cref{alg:XP} (and never change their best and worst cases afterwards again).

For every other agent~$c \in A \setminus X$, the statement clearly holds before the first execution of the while-loop.
For all other executions of the while-loop, \Cref{le:wcbcM1} implies that as soon as $\bc(c)\neq \wc(c)$ and 
$\bc(c)\neq  
\bot$ and $\wc(c)\neq  
\bot$, we have $\wc(c)\succ_c^{\mathcal{P}_2} \bc(c)$, which leads to a 
rejection in Line \ref{line:reject-bcwc} of \Cref{alg:XP}.
\end{proof}

Using this, we conclude that if one of $\bc(c)$ or $\wc(c)$ is set 
to some agent, then $c$ is matched to this agent in $M$ or $c$ is currently 
unmatched and this agent was $c$'s most recent partner:
\begin{lemma} \label{le:wcbc}
	For each agent~$c\in A$ at the beginning of each execution of 
	the 
	while-loop in Line~\ref{line:while}, the 
	following holds:
	\begin{enumerate}
		\item If $c$ is matched by $M$ and $\bc(c) \neq \bot$, then $M(c) 
		= \bc (c)$. \label{it:1}
		\item If $c$ is matched by $M$ and $\wc (c) \neq \bot$, then $M(c) 
		= \wc 
		(c)$. \label{it:2}
		\item If $c$ is unmatched by $M$ and $\bc(c) \neq \bot$, then $\bc 
		(c)$ is the last agent $c$ was matched to by~$M$. \label{it:3}
		\item If $c$ is unmatched by $M$ and $\wc (c) \neq \bot$, then 
		$\wc (c)$ is the last agent $c$ was matched to by~$M$. \label{it:4} 
	\end{enumerate}
\end{lemma}

\begin{proof}
	At the beginning of the algorithm, all four items
trivially hold, as all agents have trivial worst and best cases. 
We start by proving \Cref{it:1} and \Cref{it:2}. Observe that each time the 
partner of  $c$ is changed and $c$ is matched 
by $M$ after this change (this can happen in Lines~\ref{line:guess} or~\ref{line:guessH} of \Cref{alg-init} or Lines~\ref{line:updateM1} or~\ref{line:updateM2} of \Cref{alg:updateapx}), 
$\bc(c)$ or $\wc(c)$ is set to $M(c)$ 
after the change. 
The only other changes of~$\bc (c)$ or $\wc (c)$ appear in Lines~\ref{line:init-guess} and \ref{line:endGuesses} of \Cref{alg:init-XP} and in Lines~\ref{line:resetwcbc}, \ref{line:resetwcbc2}, \ref{line:wresetwcbc}, and \ref{line:wresetwcbc2} of \Cref{alg:updateapx}, and after these changes, $c$ is unmatched.
Applying \Cref{le:one-nontrivial}, 
\Cref{it:1} and \Cref{it:2}
follow.
We now turn to \Cref{it:3} and \Cref{it:4}.
Note that before the first execution of the while-loop, if an agent $c$ becomes unmatched, then its best or worst case is set to $M_1(c)$, which is the last agent $c$ was matched to. 
To see that \Cref{it:3} and \Cref{it:4} also hold before the $i$th iteration of the while-loop for $i>1$, we distinguish two cases.

First, observe that we never change the best and worst case of an agent that is unmatched before and after the change.
If agent $c$ becomes 
unmatched and $\bc(c)$ and $\wc(c)$ are not modified in this execution of the while-loop, then by \Cref{it:1} and 
\Cref{it:2}, we know that before $c$ becomes unmatched, agent~$c$ is matched to 
$\bc(c)$ if $\bc(c) \neq \bot$ and to $\wc(c)$ if $\wc(c) \neq 
\bot$. Thus, \Cref{it:3} and \Cref{it:4} follow. 

If $c$ becomes unmatched 
and $\bc(c)$ or $\wc(c)$ is modified (Lines 
\ref{line:resetwcbc} and \ref{line:resetwcbc2} and Lines~\ref{line:wresetwcbc} and~\ref{line:wresetwcbc2} of \Cref{alg:updateapx}), then $\bc(c)$ or 
$\wc(c)$ is set to the agent to which $c$ is matched by $M$ before the last 
modification of~$M$. Thus, \Cref{it:3} and \Cref{it:4} follow.
\end{proof}	

We conclude this part by proving two further useful lemmas:
	\begin{lemma}\label{lem:missing}
	 For each agent~$c \in A \setminus X$ at the beginning of each execution of the while-loop in Line \ref{line:while} with~$M(c) \neq M_1 (c)$ and $M( c) \neq \square$, we have $M(c) = \bc (c) $ or $M (c) = \wc (c)$.
	\end{lemma}

	\begin{proof}
	From $M(c)\neq M_1(c)$, by \Cref{le:welldefined}, it follows that $\bc(c)\neq \bot$ or $\wc(c)\neq \bot$. 
	As $c$ is matched by $M$, 
	 the lemma now follows from \Cref{it:1} of \Cref{le:wcbc} or \Cref{it:2} of \Cref{le:wcbc}.
	\end{proof}

	\begin{lemma} \label{le:abM1}
		If $\textsc{Propagate}(\cdot)$ is called for an agent
		$a\in A\setminus X$ and we have that the selected stable pair is $e=\{a,b\}$ with $b\in A\setminus X$, then $\{a,b\}\notin M_1$.  
	\end{lemma}
	\begin{proof}
		If $\bc(a)\neq \bot$, then from \Cref{le:wcbcM1}, we get that $M_1(a)\succeq_a^{\mathcal{P}_2}\bc(a)$. 
		As we require that $\bc(a)\succ_a^{\mathcal{P}_2} b$, we get that $M_1(a)\succ_a^{\mathcal{P}_2} b$. 
		
		Otherwise, by \Cref{le:welldefined}, we have  $\wc(a)\neq \bot$. 
		Then from \Cref{le:wcbcM1}, we get that $\wc(a)\succeq_a^{\mathcal{P}_2} M_1(a)$. 
		As we require that $b\succ_a^{\mathcal{P}_2} \wc(a)$, we get that $b \succ_a^{\mathcal{P}_2} M_1(a)$. 
	\end{proof}
	
	\paragraph{Well-Definedness} \label{sec:defined}
	We continue by showing that \Cref{alg} is well-defined. To do so, we need 
	to show that in Line~\ref{line:choose-a} of \Cref{alg:XP} in every blocking 
	pair one agent has a trivial best case and the other a non-trivial best 
	case. As otherwise the guess would have been rejected in Line 
	\ref{line:reject-bp} of \Cref{alg:XP}, we know that at least one agent from the blocking 
	pair has a trivial best case. Thus, it suffices to show the following: 
	
	\begin{lemma}\label{lem:well-def-1}
	  At any point of the execution of the algorithm, for every agent pair~$\{c, d\}$ currently blocking $M$ where both $c$ and $d$ are currently matched by $M$, 
	  we have $\bc (c) \neq \bot$ or $\bc (d) \neq \bot$.
	\end{lemma} 
	\begin{proof}
	  If the preferences of $c$ or $d$ differ in $\mathcal{P}_1$ and $\mathcal{P}_2$, then a 
	  non-trivial best case for~$c$ or~$d$ was set in Line~\ref{line:guesswcbc} of \Cref{alg:XP}. So assume that 
	  $c$ and $d$ 
	  have the same preferences in $\mathcal{P}_1$ and $\mathcal{P}_2$. Then, 
	  as $\{c,d\}$ does not block~$M_1$, agent~$c$ or $d$ is matched worse in 
	  $M$ than in $M_1$. Assume that $c$ is matched worse to an agent $c'$, i.e., we have $M (c) = c'$ with $M_1 (c) \succ_c c'$. By 
	  \Cref{lem:missing}, it follows that $\bc(c)=c'$ or $\wc(c)=c'$. By 
	  \Cref{le:wcbcM1}, and as $c$ prefers~$M_1 (c)$ to $c'$, it follows that 
	  $\bc(c)=c'$.
	\end{proof}
	
\paragraph{Propagation of best and worst cases} \label{sec:propagate}

	Now we show that we correctly propagate $\bc(\cdot)$ and $\wc(\cdot)$, i.e., 
	every stable matching obeying a given guess must for every agent $c\in A$ 
	obey 
	$\bc (c) $ and $\wc 
	(c) $ as set at each point during the execution of the algorithm: 
	For $c\in A$, we say that a 
	matching~$M$ \emph{obeys} $\bc (c)$ if $\bc (c)= \bot $ or $\bc (c) 
	\succeq^{\mathcal{P}_2}_c M(c) $, and we say that $M$ \emph{obeys} 
	$\wc (c)$ if $\wc (c)= \bot$ or $M(c) \succeq^{\mathcal{P}_2}_c \wc 
	(c)$.
	Further, we say that a matching obeys the best and worst cases if it 
	obeys $\bc (c)$ and $\wc (c)$ for every agent~$c\in A$.
	Lastly, we say that a matching $M$ \emph{obeys our guesses} if $M$ obeys all best and worst cases after the call of the \textsc{Initialization}($\cdot$) function and the guessed set $F$ contains exactly the set of agent pairs from $M_1$ where both endpoints prefer $M$ to $M_1$ in $\mathcal{P}_2$ and the guessed set~$H$ contains exactly the agent pairs from $M$ where both endpoints prefer $M_1$ to $M$ in $\mathcal{P}_2$. 

As the final piece before proving that we propagate the best and worst cases 
correctly, we show the following:
	
    \begin{lemma} \label{le:change-quo}
	 If $\textsc{Propagate}(\cdot)$ is called for an agent
$c\in A$, then it 
holds 
that 
    \begin{enumerate}
		\item If $\bc(c) \neq \bot$, then $c$ cannot be matched to 
$\bc(c)$ in any stable matching obeying the current guess and the current best 
and worst cases.
		\item If $\wc(c) \neq \bot$, then $c$ cannot be matched to 
$\wc(c)$ in any stable matching obeying the current guess and the current best 
and worst cases.
	\end{enumerate}
	\end{lemma}
	\begin{proof}
	 We distinguish two cases based on whether $c$ is matched in $M$ when $\textsc{Propagate}(\cdot)$ is called  or not. 
    If~$c$ is currently matched, then there exists a 
blocking pair $\{c,d\}$ for $M$ selected in Line~\ref{line:selectbp} of \Cref{alg:XP} where both $c$ and $d$ are matched. Moreover, by Lines~\ref{line:choose-a} and \ref{line:wcupdate}  of \Cref{alg:XP}, it 
needs to hold that $\bc(d)\neq \bot$ and $\wc(c)\neq \bot$. By 
\Cref{le:wcbc}, from this it follows that $M(d)=\bc(d)$ and $M(c)=\wc(c)$.
Note in particular that $c\succ_d^{\mathcal{P}_2} M(d)=\bc(d)$.
Thus, each matching that obeys the current best and 
worst cases needs to match $d$ 
to an agent to which $d$ prefers~$c$.
Using this and that 
$\{c,d\}$ blocks $M$, implying that $d\succ_c^{\mathcal{P}_2}M(c)=\wc(c)$, it follows that $c$ needs to be 
matched to an agent which it prefers to $d$ and thus that~$\{c,\wc(c)\}$ cannot 
be part of a stable matching.  

If $c$ is currently unmatched by $M$, then we again distinguish two cases.
First, if $c$ became unmatched before the first call of the while-loop, then 
it holds that $\bc(c)$ or $\wc(c)$ was set to~$M_1(c)$ and 
no 
matching respecting the current guess can match $c$ to $M_1(c)$. 

Second, assume that $c$ became unmatched in the while-loop and, more 
specifically, in a previous call of the $\textsc{Propagate}(\cdot)$ function. Let 
$d$ be the last 
agent $c$ was matched to by $M$ and $\{a,b\}$ the pair examined in the call of 
the $\textsc{Propagate}(\cdot)$ function during which $c$ became unmatched (note 
that it needs to hold that $d\in \{a,b\}$). Then by \Cref{le:wcbc}, it follows 
that currently it holds that $\bc(c)=d$ if $\bc(c)\neq \bot$ and 
$\wc(c)=d$ if $\wc(c)\neq 
\bot$. Thus, we need to exclude that $\{c,d\}$ is part of a stable matching. 
By \Cref{le:wcbc}, we have that before the call of 
the $\textsc{Propagate}(\cdot)$ function in which $c$ got unmatched it holds that 
$\bc(d)=c$ if $\bc(d)\neq \bot$ and $\wc(d)=c$ if $\wc(d)\neq \bot$. 
Using this, as in the call of $\textsc{Propagate}(\cdot)$ examining the pair $\{a,b\}$ the best or the worst cases of both $a$ and $b$ are changed,  after the call of the $\textsc{Propagate}(\cdot)$ function in which $c$ got unmatched, by 
\Cref{le:gettingbetter}, it holds that 
$c\succ_d \bc(d)$ if $\bc(d)\neq 
\bot$ or $\wc(d)\succ_d c$ if $\wc(d)\neq \bot$. 
Thus, a matching 
respecting the current worst and best cases cannot contain~$\{c, d\}$. 
	\end{proof}

	We are now ready to prove that we correctly propagate agent's best and worst cases, the cornerstone of our proof of correctness:
	\begin{lemma}\label{lem:bounds}
	  Every stable matching~$M^*$ in $\mathcal{P}_2$ obeying a guess obeys the best and worst cases at every point during the execution of the algorithm for this guess. 
	  
	  If a stable matching $M^*$ obeying a guess exists, then we do not reject  in  Line~\ref{line:rej} and Line~\ref{line:rej2} of \Cref{alg:updateapx} during the execution of the algorithm for this guess.
	\end{lemma}

	\begin{proof}
      Let $M^*$ be a stable matching in $\mathcal{P}_2$ obeying our guesses.
      We now show that $M^*$  obeys the  best and worst cases at every point during the execution of the algorithm for this guess.
      Before the first execution of the while-loop in Line \ref{line:while} of \Cref{alg:XP}, 
      the current best and worst cases are obeyed since~$M^*$ obeys the guess.
      We now argue that 
the matching $M^*$ continues to obey the best and worst cases. For this we examine  all points during the execution of the algorithm where best and worst cases may change.

\subparagraph{Changes outside of \textsc{Propagate}$\bm{(\cdot)}$.}
The only point where an update may happen inside of the while-loop but outside 
of the \textsc{Propagate}$(\cdot)$  function is in Line~\ref{line:wcupdate} of \Cref{alg:XP}. 
Let~$M^{\before}$, $\bc^{\before}$ and $\wc^{\before}$ be the matching $M$, best and worst case before the update, and $M^{\after}$, $\bc^{\after}$ and $\wc^{\after}$ the matching $M$, best and worst case after the update. 
We set 
$\wc^{\after} (x)$ to $M^{\before}(x)$ in Line~\ref{line:wcupdate} if $\{x, y\}$ forms a blocking pair for $M^{\before}$, both~$x$ and $y$ are matched in $M^{\before}$, $\bc^{\before}(y)\neq \bot$ and $\wc (x) = \bot = \bc (x)$. By  
\Cref{le:wcbc} and as $\bc^{\before}(y)\neq \bot$, it follows that $\bc^{\before}(y) = M^{\before}(y)$. Thus, as $\{x, y\}$ does not block~$M^*$ and $M^*$ obeys $\bc^{\before}(y)$ (i.e., $x\succ_y^{\mathcal{P}_2} M^{\before}(y) = \bc^{\before} (y) \succeq^{\mathcal{P}_2}_y M^* (y)$) it needs to hold that $M^*(x)\succeq^{\mathcal{P}_2}_x y$ and thus 
that $M^*(x)$ respects~$\wc^{\after}(x)=M^{\before}(x)\prec_x^{\mathcal{P}_2} y$, where $M^{\before}(x)\prec_x^{\mathcal{P}_2} y$ holds as $\{x,y\}$ blocks $M$.

\subparagraph{\textsc{Propagate}$\bm{(\cdot)}$.}
Let us now consider a call of the function 
\textsc{Propagate}$(\cdot)$ and let $a\in A$ be the agent on which the 
function was called and $\{a,b\}$ the considered stable pair from Line~\ref{line:def-b} or~\ref{line:def-b-2} of \Cref{alg:updateapx}.  
We prove that $M^*$ respects the best and worst cases by induction over the number of iterations of the the while loop. 
Assume that the statements holds after iteration $i\geq 0$. 
And let us examine the $i+1$th iteration of the while-loop. 
For this let $M^{\before}$, $\bc^{\before}$ and $\wc^{\before}$ be the matching $M$, best and worst case at the beginning of this call of \textsc{Propagate}$(\cdot)$, and $M^{\after}$, $\bc^{\after}$ and $\wc^{\after}$ the matching~$M$, best and worst case after this call of \textsc{Propagate}$(\cdot)$. 
By our induction hypothesis, we know that $M^*$ respects $\bc^{\before}$ and $\wc^{\before}$ (for $i=0$ this follows from the paragraph above). 

We first show that $a\notin X$. 
Suppose towards a contradiction that $a\in X$. 
As $M^*$ respects our guesses, we have guessed in the \textsc{Initialization}$(\cdot)$ function the partner of all agents from $X$ in $M^*$ and in particular set $M(a)=M^*(a)$ and $\bc(a)=\wc(a)=M^*(a)$.
Moreover, by \Cref{le:gettingbetter} and as $M^*$ still respects the current guesses, we  have $\bc^{\before}(a)=\wc^{\before}(a)=M^*(c)$ and by \Cref{le:wcbc} thus also $M^{\before}=M^*(a)$. 
However, from \Cref{le:change-quo} we get that if \textsc{Propagate}$(\cdot)$ is called on $a$, then $a$ cannot be matched to $\bc^{\before}(a)=M^*(a)$ in a stable matching obeying the current guess and $\bc^{\before}$ and $\wc^{\before}$, a contradiction to the existence of $M^*$.

We now argue chronologically one change of best and worst cases after each other why $M^*$ still obeys the best and worst cases $\bc^{\after}$ and $\wc^{\after}$.

\subparagraph{\textsc{Propagate}$\bm{(\cdot)}$---Changing best and worst cases of $\bm{a}$.}
First, we consider the case~$\bc^{\before} (a)\neq \bot$.
We now argue that $b\succeq^{\mathcal{P}_2}_a M^*(a)$ and thus that 
$M^*$  respects $\bc^{\after} (a)=b$.
First of all note that $\{a, M^* (a)\}$ clearly needs to be a 
stable pair in~$\mathcal{P}_2$. 
As $M^*$ respects $\bc^{\before}$ and $\wc^{\before}$, 
it needs to 
hold that $\bc(a)^{\before} \succ_a^{\mathcal{P}_2} M^*(a)$ (where we additionally apply \Cref{le:change-quo}) and that 
$a\succeq_{M^*(a)}^{\mathcal{P}_2} \wc^{\before} (M^*(a))$ or $\wc^{\before} (M^*(a))=\bot$. Moreover, by \Cref{le:wcbcM1} 
it holds that 
$M_1(a)\succeq_{a}^{\mathcal{P}_2} \bc^{\before}(a)\succ_a^{\mathcal{P}_2} M^*(a)$. We now make a case distinction based on the preferences of $M^*(a)$.

If $M^*(a)$  
prefers $M_1(M^*(a))$ to~$a$, then for the pair~$\{a,M^*(a)\}\in 
M^*$ both endpoints prefer~$M_1$ to $M^*$.
This means that $\{a,M^*(a)\}\in H$, 
as 
$M^*$ respects the current guess. 
Thus, we have $a \in X$, a contradiction to~$a \notin X$.
Consequently, we have
$a\succeq_{M^*(a)}^{\mathcal{P}_2} M_1(M^*(a))$.
Thus, we have proven that $\{a,M^*(a)\}$ is a stable pair, that 
$a\succeq_{M^*(a)}^{\mathcal{P}_2} \wc^{\before} (M^*(a))$ or $\wc^{\before} (M^*(a))=\bot$, and
$\bc(a)^{\before} \succ_a^{\mathcal{P}_2} M^*(a)$. 
These are exactly the constraints $b$ has to fulfill in Line \ref{line:def-b}. As $b$ is the agent best ranked by $a$ fulfilling these constraints, we clearly get 
$b\succeq^{\mathcal{P}_2}_a M^*(a)$.
Thus, $M^*$ also respects $\bc^{\after} (a)=b$. 

We now turn to the case that  $\bc^{\before}(a)=\bot$. Then, very similar to the previous case,  
\Cref{le:welldefined} implies
$\wc^{\before}(a)\neq \bot$.
First of all note that in~$M^*$, the pair containing $a$ clearly needs to be a 
stable pair in~$\mathcal{P}_2$. 
By  
\Cref{le:change-quo} and as $M^*$ respects the current best and worst cases, 
it needs to 
hold that $M^*(a) \succ_a^{\mathcal{P}_2} \wc^{\before}(a)$ and that 
$\bc^{\before} (M^*(a)) \succeq_{M^*(a)}^{\mathcal{P}_2}  a$ or $\bc^{\before } (M^* (a)) = \bot$. 
By the definition of $b$ it follows that $M^*(a)\succeq^{\mathcal{P}_2}_a b$. Thus, $M^*$ also respects $\wc^{\after}(a)=b$. 

\subparagraph{\textsc{Propagate}$\bm{(\cdot)}$---Changing best and worst cases of $\bm{b}$.}
      \textsc{Propagate}$(\cdot)$ also updates $\bc (b) $ and $\wc (b)$. First 
assume that we modified $\wc (b)$ (which happens if and only if ${\bc^{\before} (a) \neq \bot}$). 
Note that in this case we set $\bc^{\after} (a)$ to~$b$.
Since we have shown above that $M^*$ respects $\bc^{\after}(a)=b$, we have $b \succeq^{\mathcal{P}_2}_a 
M^*(a)$. As 
$\{a, b\}$ does not block~$M^*$ and $b \succeq^{\mathcal{P}_2}_a 
M^*(a)$, it follows that $M^* (b) \succeq_b^{\mathcal{P}_2} a$, and thus, $M^*$ respects $\wc^{\after}(b)=a$.
      
      Now assume that we modified $\bc (b)$ by setting $\bc^{\after}(b) = a$ (which happens if and only if $\bc^{\before} (a) = \bot$).
      In this case, we set $\wc^{\after} (a)$ to~$b$.
      As $M^*$ respects  
      $\wc^{\after}(a)=b$, it needs to hold that $M^*(a) \succeq^{\mathcal{P}_2}_a b$.
      As $\{a,b \}$ is a stable pair in $\mathcal{P}_2$ and $M^*(a) \succeq^{\mathcal{P}_2}_a b$ it follows from \Cref{lem:circular-prefs} that  $a\succeq_{b}^{\mathcal{P}_2} M^*(b)$. 
      This proves 
      that $M^*$ respects $\bc^{\after}(b)=a$.
      
Note that we have not shown (or assumed) so far that $b\notin X$. 
To establish that we do not reject this guess in Line \ref{line:rej}, assume for the sake of contradiction that $b\in X$. 
As $M^*$ respects our guesses, we have guessed in the \textsc{Initialization}$(\cdot)$ function the partner of all agents from $X$ in $M^*$ and in particular set $M(b)=M^*(b)$ and $\bc(b)=\wc(b)=M^*(b)$.
Moreover, by \Cref{le:gettingbetter} and as $M^*$ respects $\bc^{\before}(b)$ and $\wc^{\before}(b)$, we have $\bc^{\before}(b)=\wc^{\before}(b)=M^*(b)$ and by \Cref{le:wcbc} thus also $M^{\before}=M^*(b)$.
However, as clearly $\{a,b\}\notin M^{\before}$, we either have $\bc^{\before}(b)\succ_b^{\mathcal{P}_2}\bc^{\after}(b)$ or $\wc^{\after}(b)\succ_b^{\mathcal{P}_2}\wc^{\before}(b)$. 
In both cases, we afterwards have  $\wc^{\after}(b)\succ_b^{\mathcal{P}_2}\bc^{\after}(b)$. 
However, this cannot be the case as we have shown that $M^*$ respects both $\wc^{\after}(b)$ and $\bc^{\after}(b)$.

 \subparagraph{\textsc{Propagate}$\bm{(\cdot)}$---Changing best and worst cases of $\bm{M_1(a)}$.}     
       \textsc{Propagate}$(\cdot)$ may also update $\bc (M_1(a)) $ or $\wc 
(M_1(a))$ if $M(a)=M_1(a)$.
      First, we consider the case that the algorithm modified~$\wc (M_1(a))$.
      As $\bc^{\before}(a)\neq \bot$ in this case, by \Cref{le:wcbcM1}, we get that 
      $M_1(a)\succeq^{\mathcal{P}_2}_a \bc^{\before}(a)$.
      Moreover, as we have set $\bc^{\after}(a)=b$ with $\bc^{\before}(a)\succ^{\mathcal{P}_2}_a b$, it follows that $M_1(a)\succ^{\mathcal{P}_2}_a \bc^{\after}(a)$
      and as we have established above that $M^*$ respects $\bc^{\after}(a)$ we get 
      $M_1(a)\succ^{\mathcal{P}_2}_a 
M^*(a)$.
      As $\{a, M_1(a)\}$ does not block~$M^*$, it follows that
      $M^*(M_1(a))\succ^{\mathcal{P}_2}_{M_1(a)} a=\wc^{\after}(M_1(a))$ and thus that $M^*$ respects 
      $\wc^{\after} (M_1(a))$.
      
      Now we consider the case that the algorithm modified $\bc (M_1(a))$.
      In this case we have $\wc^{\before}(a)\neq \bot$ by \Cref{le:welldefined}. 
       By \Cref{le:wcbcM1}, we get that 
      $\wc^{\before}(a)\succeq^{\mathcal{P}_2}_a M_1(a)$.
      Moreover, as we have set $\wc^{\after}(a)=b$ with $b \succ^{\mathcal{P}_2}_a \wc^{\before}(a)$, it follows that $\wc^{\after}(a)\succ^{\mathcal{P}_2}_a M_1(a)$
      and as we have established above that $M^*$ respects $\wc^{\after}(a)$ we get 
      $M^*(a)\succ^{\mathcal{P}_2}_a M_1(a)$.
      Assume for the sake of contradiction that 
$M^*$ does not 
respect $\bc^{\after}(M_1(a))=a$, i.e., $M^*(M_1(a))\succ^{\mathcal{P}_2}_{M_1(a)} a$. However, as 
$\{a,M_1(a)\}\in M_1$ this 
implies that $\{a,M_1(a)\}$ is a pair from $M_1$ where both endpoints prefer $M^*$ to $M_1$. 
Thus, as $M^*$ respects the current guess it needs to hold that $\{a,M_1(a)\}\in F$ 
and that the pair~$\{a,M_1(a)\}$ was deleted from $M$ already before the first call 
of the $\textsc{Propagate}(\cdot)$ function (and was clearly never inserted again as we only insert the selected stable pairs $\{a,b\}$ in a call of the $\textsc{Propagate}(\cdot)$ function and by  \Cref{le:abM1}, $\{a,b\}\notin M_1$), a contradiction to $M(a)=M_1(a)$. 

\subparagraph{\textsc{Propagate}$\bm{(\cdot)}$---Changing best and worst cases of $\bm{M_1(b)}$.} 
Similar arguments as for the previous case apply here. 
If $\wc(M_1(b))$ was modified, then we have $\bc^{\after}(b)=a$ and as we have established above that $M^*$ respects $\bc^{\after}(b)$ that $a\succeq_b^{\mathcal{P}_2} M^*(b)$. 
From \Cref{le:wcbcM1} (applied to the next iteration of the while loop) it further follows that $M_1(b)\succeq_b^{\mathcal{P}_2} \bc^{\after}(b)=a$. 
It follows that $M_1(b)\succeq_b^{\mathcal{P}_2} M^*(b)$. 
Thus, for $\{b,M_1(b)\}$ not to form a blocking pair for $M^*$, it needs to hold that $M^*(M_1(b))\succeq_{M_1(b)}^{\mathcal{P}_2} b=\wc^{\after}(M_1(b))$. Consequently, $M^*$ respects $\wc^{\after}(M_1(b))$.

If $\bc(M_1(b))$ was modified, then we have by the definition of $b$ that $\wc^{\after}(b)=a\succeq_b^{\mathcal{P}_2} M_1(b)$.
Moreover, as we know that $\{a,b\}\notin M_1$ by \Cref{le:abM1}, it even holds that $\wc^{\after}(b)=a\succ_b^{\mathcal{P}_2} M_1(b)$.
As we have established above that $M^*$ respects $\wc^{\after}(b)$ it follows that $M^*(b) \succ_b^{\mathcal{P}_2} M_1(b)$.
Assume for the sake of contradiction that $M^*$ violates $\bc^{\after}(M_1(b))=b$, i.e., ${M^*(M_1(b))\succ_{M_1(b)}^{\mathcal{P}_2} b}$. 
However, as 
$\{b,M_1(b)\}\in M_1$ this 
implies that $\{b,M_1(b)\}$ is a pair from $M_1$ where both endpoints prefer $M^*$ to $M_1$. 
Thus, as $M^*$ respects the current guess it needs to hold that $\{b,M_1(b)\}\in F$ 
and that the pair $\{b,M_1(b)\}$ was deleted from $M$ already before the first call 
of the $\textsc{Propagate}(\cdot)$ function (and was clearly never inserted again as we only insert the selected stable pairs $\{a,b\}$ in a call of the $\textsc{Propagate}(\cdot)$ function and by  \Cref{le:abM1}, $\{a,b\}\notin M_1$), a contradiction to $M_1(b)=M(b)$.  

\medskip
We now turn to the second part of the lemma. 
First of all note that from the existence of~$M^*$ it follows that a stable pair as defined in Line \ref{line:def-b} or Line \ref{line:def-b-2} of \Cref{alg:updateapx} always exists, as we have shown that the partner of $a$ in $M^*$ fulfills all the required properties.
Moreover, we have argued above that we always have $a\notin X$ and $b\notin X$ implying that we do not reject in Line \ref{line:rej} or Line \ref{line:rej2} in case $M^*$ exists. 	\end{proof}
	
	\paragraph{Bounding the Number of Guesses} \label{sec:bounding}
	
	It remains to argue that our guesses are exhaustive and indeed cover all 
cases. In particular, we now show that the set~$F$ guessed in Line~\ref{line:guess-F} of 
	\Cref{alg-init} is 
	large 
	enough, i.e., for every stable matching~$M_2$ in $\mathcal{P}_2$ there are 
at 
	most $|\mathcal{P}_1 \oplus \mathcal{P}_2|$ pairs of~$M_1$ for which both 
	endpoints strictly prefer~$M_2$ to $M_1$, and that the set~$H$ guessed in Line~\ref{line:guessH} of \Cref{alg-init} is large enough, i.e., for every stable matching $M_2$ in $\mathcal{P}_2$ there are only $|\mathcal{P}_1 \oplus \mathcal{P}_2|$ pairs of~$M_2$ for which both 
	endpoints strictly prefer~$M_1$ to $M_2$ .
	
	We start by proving this bound for the guessed set $F$:
	\begin{lemma}\label{lem:F}
	 Let $M_2$ be a perfect stable matching in $\mathcal{P}_2$ and $M_1$ be a perfect stable 
	 matching in~$\mathcal{P}_1$.
	The number of pairs $e = \{b, c\} \in M_1$ such that $b$ prefers~$M_2(b)$ to $c$ and $c $ prefers~$M_2(c)$ to $b$ in $\mathcal{P}_2$ is bounded by the number 
	 of agents whose preferences differ in $\mathcal{P}_1$ and $\mathcal{P}_2$.
	\end{lemma}

	\begin{proof}
	 $M_1 \triangle M_2$ is a disjoint union of cycles of even length. 
	 For each such cycle~$C$, fix an orientation of~$C$.
	 Let $C= (v_1, \dots, v_{2r}, v_{2r+1}= v_1)$ be a cycle containing a
pair~$e = \{b, c\}\in M_1$ such that $b$ prefers~$M_2 (b)$ to $c$ and $c$ 
prefers $M_2 (c)$ to $b$ in $\mathcal{P}_2$.
	 We assume without loss of generality that $e = \{v_1, v_2\}$.
	 Let $i$ be the smallest index in $[1,2r]$ such that the preferences 
	 of~$v_i$ differ in 
	 $\mathcal{P}_1$ and $\mathcal{P}_2$ (and $i = 
	 \infty $ if no such agent exists).
	 We show by induction on $j$ that $v_j$ prefers $v_{j+1}$ to $v_{j-1}$ (in $\mathcal{P}_1$ and $\mathcal{P}_2$) for 
every $j \in \{2, 3, \dots, i-1\}$.
	 For $j= 2$, this follows by the definition of~$e$.
	 So consider~$j \ge 3$, and let $\{v_{j}, v_{j+1} \} \in M_p$ for 
some~$p\in 
	 \{1,2\}$.
	 By induction, $v_{j-1}$ prefers~$v_{j}$ to~$v_{j-2}$.
	 Since $\{v_{j-1}, v_{j}\}$ does not block $M_p$, it follows that $v_j$ 
	 prefers $v_{j+1} $ to~$v_{j-1}$.
	 This implies that $i < \infty$ because if $ i = \infty$, then $v_{2r + 1} = v_1$ prefers $v_{2} 
	 $ to $v_{2r}$ by the above proven claim. This leads to a contradiction, as $v_1$ prefers~$v_{2r}$ to $v_{2}$ by 
the 
	 definition of~$e$. Note that from the above claim it also follows that apart from $\{b,c\}$ there cannot be a second pair from $M_1$ in $\{v_1, \dots, v_i\}$ with both agents from this pair preferring $M_2$ to $M_1$, as
	 we have proven above that there cannot exist a $j\in [2,i]$ such that $\{v_j,v_{j+1}\}\in M_1$ and 
	 $v_{j-1}\succ_{v_{j}} v_{j+1}$. Thus, 
	 mapping pair $e$ to agent $v_i$ yields an injection from the 
	 set of pairs~$\{b, c\}\in M_1$ with both $b$ and $c$ preferring $M_2$ to 
	 $M_1$ to the set of agents 
with modified preferences.
	\end{proof}
    By swapping the roles of $M_2$ and $M_1$ and the roles of $\mathcal{P}_1$ 
and $\mathcal{P}_2$,  \Cref{lem:F} also shows that the set $H$ of pairs guessed 
in Line~\ref{line:guessH} of \Cref{alg-init} is large enough, i.e., for every stable matching 
$M_2$ in $\mathcal{P}_2$ there exist at most $|\mathcal{P}_1 \oplus 
\mathcal{P}_2|$ many pairs from $M_2$ where both endpoints prefer $M_1$ to 
$M_2$. 

\paragraph{Proof of Correctness: Putting the Pieces Together} \label{sec:putting}

	We are now ready to put the pieces together and prove the correctness of the algorithm.
 \thISRXPPP*
 \begin{proof}
  \Cref{lem:well-def-1}, and Line 
\ref{line:reject-bp} of \Cref{alg:XP} show that Line~\ref{line:choose-a} of \Cref{alg:XP} and thus the algorithm is 
  well-defined.
 
  Note that for each guess executing the algorithm takes $\mathcal{O}(n^3)$ time: 
  For each agent~$a$, the best case and worst case can only be changed $n$ times by \Cref{le:gettingbetter}.
   Since in every iteration of the while-loop 
$\textsc{Propagate}(\cdot)$ for some agent~$a$ is called and either 
$\bc(a)$ or $\wc(a)$ is modified,  the while-loop is executed at most~$n^2$ 
  times.
  Moreover, given the set of stable pairs in $\mathcal{P}_2$, which can be precomputed in  $\mathcal{O} (n^3)$ time \cite{DBLP:journals/algorithmica/Feder94}, each iteration of the while-loop takes~$\mathcal{O}(n)$~time. 
  Moreover, for our guesses made in the initialization phase in \Cref{alg-init} there exist overall $\mathcal{O}(2^{4|\mathcal{P}_1 \oplus \mathcal{P}_2| 
  }\cdot n^{5|\mathcal{P}_1 \oplus \mathcal{P}_2| })$ possibilities: 
  There are $n^{2|\mathcal{P}_1\oplus \mathcal{P}_2|}$ possibilities how the agents with  changed preferences and their partners in $M_1$  are matched (Line~\ref{line:guess}).
  There are $n^{2|\mathcal{P}_1\oplus \mathcal{P}_2|}$ possibilities for the set $H$ from Line \ref{line:guessH}. 
  Further, there are $n^{|\mathcal{P}_1\oplus \mathcal{P}_2|}$ possibilities for the set $F$ from Line \ref{line:guess-F} (as we only need to iterate over $|\mathcal{P}_1\oplus \mathcal{P}_2|$-subsets of $M_1$). 
  Lastly, as $X$ contains at most $4|\mathcal{P}_1\oplus \mathcal{P}_2|$ agents, there exist at most  $2^{4|\mathcal{P}_1\oplus \mathcal{P}_2|}$ possibilities for our guesses in Line \ref{line:init-guessPart}. 
  Thus, we get an 
  overall running time of $\mathcal{O}(2^{4|\mathcal{P}_1 \oplus \mathcal{P}_2| 
  }\cdot n^{5|\mathcal{P}_1 \oplus \mathcal{P}_2| }\cdot n^3)$.
  It remains to show the correctness of \Cref{alg:XP}.
 
  If \Cref{alg:XP} returns a matching, then this matching $M$ is stable, as all 
  blocking pairs get resolved in the while-loop and also satisfies $|M\triangle 
  M_1|\leq k$ by Line~\ref{line:reject} of \Cref{alg:XP}.
  
  It remains to prove that if there exists a stable matching~$M_2$ in $\mathcal{P}_2$ with~$|M_1 \triangle 
M_2| \leq k$, then  \Cref{alg:XP} returns a matching for some guess. 
Let $F'$ be the set of agent pairs from~$M_1$ where both endpoints prefer $M_2$ to $M_1$ in $\mathcal{P}_2$ and let $H'$ be the set of agent pairs from~$M_2$ where both endpoints prefer $M_1$ to $M_2$ in $\mathcal{P}_2$. 
Then, as proven in \Cref{lem:F}, we have $\max\{|F'|, |H'|\}\leq |\mathcal{P}_1\oplus \mathcal{P}_2|$.
So assume for the sake of contradiction that \Cref{alg:XP} rejected the guess with~$F= F'$, $H = H'$, and guesses in Lines \ref{line:guess} and \ref{line:init-guessPart} in \Cref{alg-init} made according to~$M_2$. 
  \Cref{lem:bounds} implies that $M_2$ obeys $\bc (c) $ and $\wc (c)$ for every 
agent~$c\in A$ at any point of the execution of the algorithm for this guess.
The guess clearly cannot be rejected during the initialization in \Cref{alg-init} because there exists a stable matching obeying the guess.

Assume for the sake of contradiction that the guess is rejected in Line~\ref{line:reject-bp} of \Cref{alg:XP} because there 
exists two agents~$\{a,b\}$ with $\bc(a)\neq \bot$ and $\bc(b)\neq 
\bot$ that form a blocking pair for~$M$. As \Cref{le:wcbc} implies
that $a$ is matched to $\bc(a)$ and $b$ to $\bc(b)$ in $M$ and as $M_2$ 
respects $\bc(a)$ and $\bc(b)$, it follows that $\{a,b\}$ also blocks $M_2$, a 
contradiction.  

Moreover, the existence of $M_2$ implies that the current 
guess cannot be rejected in Line~\ref{line:rej} or Line~\ref{line:rej2} of \Cref{alg:updateapx} by \Cref{lem:bounds}.
Lastly, as $M_2$ obeys the best and worst cases of all agents at any point during the execution of the algorithm, the guess cannot be rejected in Line~\ref{line:reject-bcwc} of \Cref{alg:XP}. 
 
  This means that the algorithm can only reject in Line~\ref{line:reject} of \Cref{alg:XP}.
  Let~$M$ be the matching in Line~\ref{line:reject} of \Cref{alg:XP}.
  Matching~$M$ contains every 
  pair~$\{a, b\} \in M_1$ with $\bc (a) = \bot \lor \bc (a)= M_1 (a)$, 
  $\wc (a) = \bot\lor M_1 (a)= \wc (a)$, $\bc (b) = \bot \lor \bc 
  (b)= M_1 (a)$, and $\wc (b) = \bot\lor M_1 (b)= \wc (b)$.
  As $M_2$ respects the best and worst cases, by \Cref{le:wcbcM1}, 
  if a pair from~$M_1$ does not satisfy the aforementioned criteria, then it cannot 
  be part of $M_2$. Thus, it holds that $|M_1 
  \triangle M| \le |M_1 \triangle M_2| \le 
  k$, a contradiction.
  
  Finally note that \Cref{alg:XP} (which requires that $M_1$ is a perfect matching and that there is a perfect stable matching in $\mathcal{P}_2$) also gives rise to an algorithm for general \ISR instances, as we have shown in \Cref{le:wlog1,le:wlog2} that each \ISR instance can be reduced to an \ISR instance fulfilling the above-described properties in linear time.
 \end{proof}
 
We remark that \ISR is NP-complete even if we know for each agent~$a$ whose preferences changed as well as~$M_1 (a)$ how they are matched in~$M_2$ and the set of pairs~$F\subseteq M_1$ for which both endpoints prefer~$M_2$ to~$M_1$.
This indicates that guessing the set~$H$ might be necessary for the XP-algorithm.
To prove this we only need to slightly alter the construction from \Cref{th:ISR-WP1P2}.
Specifically, for each $c\in [\ell]$, we add agents $x^c$, $y^c$, $\bar{x}^c$, and $\bar{y}^c$ and  replace pair~$\{s^c, t^c\}$ by pairs~$\{s^c, x^c\}$, $\{x^c, y^c\}$, $\{y^c, t^c\}$ and pair $\{\bar s^c, \bar t^c\}$ by pairs~$\{\bar s^c, \bar x^c\}$, $\{ \bar x^c, \bar y^c\}$, and $\{\bar y^c, \bar t^c\}$.
Here, $x^c$ prefers $y^c$ to $s^c$, and $y^c $ prefers $t^c$ to $x^c$ (and symmetrically, we have that $\bar x^c$ prefers $\bar y^c$ to~$\bar s^c$ and $\bar y^c$ prefers $\bar t^c $ to $\bar x^c$) in $\mathcal{P}_1$ and $\mathcal{P}_2$.
Matching $M_1$ contains edges $\{s^c,x^c\}$, $\{y^c,t^c\}$, $\{\bar{s}^c,\bar{x}^c\}$, and  $\{\bar{y}^c,\bar{t}^c\}$ instead of $\{s^c,t^c\}$ and $\{\bar{s}^c,\bar{t}^c\}$.
Note that every stable matching in~$\mathcal{P}_2$ then contains pairs~$\{u^c, t^c\}$, $\{x^c, y^c\}$, $\{\bar u^c, \bar t^c\}$, and $\{\bar x^c, \bar y^c\}$ (so there is nothing to guess how the agents with changed preferences and their partners in $M_1$ are matched).
Setting~$F:= \emptyset$ now results in an NP-complete problem by the same proof as the one of \Cref{th:ISR-WP1P2}.

\section{Incremental Stable Marriage with Ties Parameterized by the Number of 
Ties}\label{ISM}
 
Bredereck et al.~\cite{DBLP:conf/aaai/BredereckCKLN20} 
raised the question how the total number of ties influences the computational complexity of \ISMT.
Note that the number of ties in a preference relation is the number of maximal sets of pairwise tied agents containing more than one agent.
For instance the preference relation $a\sim b\sim c\succ d\sim e\succ f$ contains two ties, where the first tie ($a\sim b\sim c$) has size three and the second tie ($d \sim e$) has size two.
In this section, following a fundamentally different and significantly 
simpler path than Bredereck et al., we show that their W[1]-hardness result for \ISMT 
parameterized by $k$ 
for~$|\mathcal{P}_1 \oplus \mathcal{P}_2|=1$ still holds if we parameterize by $k $ \emph{plus} the number of ties.
To prove this, we introduce a natural extension of \ISMT called \textsc{Incremental Stable Marriage with Forced Edges and Ties} (\textsc{ISMFE-T}).
\textsc{ISMFE-T}  differs from \ISMT in that as part of the input we are additionally given a subset~$Q\subseteq M_1$ of the initial matching, and the question is whether there is a stable matching~$M_2$ for the changed preferences with $|M_1 
\triangle M_2| \le k$ containing all pairs from $Q$, i.e., $Q \subseteq M_2$. 

We first show that \textsc{ISMFE-T} is intractable even if $|Q|=1$ by reducing from a W[1]-hard local search problem related to finding a perfect stable matching with ties~\cite{DBLP:journals/algorithmica/MarxS10}:
\begin{restatable}{proposition}{ISMFEEE}
	\label{th:ISMFE}
	\textsc{ISMFE-T} parameterized by $k$ and the 
	summed number of ties in $\mathcal{P}_1 $ and $\mathcal{P}_2$ is W[1]-hard, even if~$|\mathcal{P}_1 \oplus \mathcal{P}_2| = 
	1$ and $|Q| = 1$ and only 
	women have ties in their preferences.
	
   \textsc{ISMFE-T} parameterized by $k$ is W[1]-hard, even if~$|\mathcal{P}_1 \oplus \mathcal{P}_2| =1$, $|Q| = 1$, only women have ties in their preferences, and each tie has size at most two. 
\end{restatable}
\begin{proof}
	We show both parts by reducing from the following problem related to 
	finding a perfect matching in an \textsc{Stable Marriage with Ties} 
	instance: Given a 
	\textsc{Stable Marriage with Ties} instance consisting of $n$ men and $n$ 
	women, an integer $\ell$, and a stable matching~$N$ 
	of 
	size~$n-1$, decide whether there exists a perfect stable 
	matching~$N^*$ with $|N \triangle N^*| \le \ell$. Marx and 
	Schlotter showed that this problem 
	is W[1]-hard parameterized by the number ties plus~$\ell$, even if only the 
	preferences of women contain ties \cite[Theorem 
	2]{DBLP:journals/algorithmica/MarxS10}  and W[1]-hard parameterized 
	by~$\ell$, 
	even if all ties have size two and are in the preferences of 
	women \cite[Theorem 
	3]{DBLP:journals/algorithmica/MarxS10}.\footnote{Notably, 
		Marx and 
		Schlotter \cite{DBLP:journals/algorithmica/MarxS10} use a different measure 
		for the difference between two matchings, i.e., the number of agents that 
		are matched differently. However, as we here know the number of pairs in $N$ 
		and $N^*$ the two distance measures can be directly converted to each other and differ by at most a factor of two.}
	
	We now establish a reduction from the above defined problem 
	to \textsc{Incremental Stable Marriage with Forced Edges and Ties}. As, in the reduction, the number of ties and the length of each tie 
	remain unchanged, we thereby establishing both statements at the same 
	time. The reduction works as follows.
	Let $(U \cupdot W,\mathcal{P})$ be an instance of \textsc{Stable Marriage} 
	with ties,
	let $N$ be a stable matching of size $n-1$, where $n = |U| = |W|$, and let~$\ell \in \mathbb{N}$. 
	Let $m_{\single}$ and $w_{\single}$ be the two agents unmatched by $N$.
	We add a man~$m^*$ and a woman~$w^*$ accepting each other and all agents from $U\cup W$ which prefer to be matched to any agent 
	from~$W $ respectively~$U$ to being matched together, and set $Q:= 
	\{\{m^*, w^*\}\}$ to be the set of forced pairs. Moreover, we add $m^*$ at 
	the 
	end of the preferences of all women and $w^*$ at the end of the preferences 
	of all men. 
	We now modify the instance such that $m_{\single}$ and $w_{\single}$ are 
	``bounded'' in $\mathcal{P}_1$ but become ``free'' and are thereby added to 
	the 
	relevant part of the instance in $\mathcal{P}_2$. As $M_2$ needs to 
	contain the pair $\{\{m^*, w^*\}\}$, we need to 
	find a matching that matches all agents from $U \cupdot W$ and has 
	a large 
	intersection with~$N$. 
	To realize these constraints on $m_{\single}$ and $w_{\single}$, we additionally add a man~$m_{\single}^*$ and 
	woman~$w_{\single}^*$ with preferences $m_{\single}^* : w_{\single} \succ 
	w_{\single}^*$ and $w_{\single}^* : m_{\single}^* \succ m_{\single}$. In 
	$\mathcal{P}_1$, we modify the preferences of man~$m_{\single}$ such that he 
	prefers~$w_{\single}^*$ to all other agents, and change the preferences of 
	$w_{\single}$ such that she prefers~$m_{\single}^*$ to all other agents.
	The preferences of~$\mathcal{P}_2$ now arise from $\mathcal{P}_1$ by 
	swapping the first two women in the 
	preferences 
	of~$m^*_{\single}$, i.e., the preferences of~$m^*_{\single}$ in $\mathcal{P}_2$ are $w_{\single}^* \succ w_{\single}$.
	Matching~$M_1$ is defined as $M_1 := N\cup \{\{m^*, w^*\}, \{m_{\single}, 
	w_{\single}^*\}, \{m_{\single}^*, w_{\single}\}\}$.
	The stability of~$M_1$ follows from the stability of $N$, as none of 
	$w_{\single}^*$, $m_{\single}^*$, $w^*$, and $m^*$ is part of a blocking 
	pair.
	We set $k: = \ell + 3$. We now prove that there exists a perfect matching~$N^*$ with $|N\triangle N^*| \leq \ell$ if and only if there exists a stable 
	matching $M_2$ in~$\mathcal{P}_2$ such that $\{m^*, w^*\} \in 
	M_2$ and $| M_1 \triangle M_2| \le k$.
	
	$(\Rightarrow):$
	Given a perfect stable matching $N^*$ with $|N \triangle N^*| \le \ell$, we 
	claim that $M_2 : = N^* \cup \{ \{ m_{\single}^*, w_{\single}^*\}, \{m^*, 
	w^*\}\}$ is a stable matching in $\mathcal{P}_2$.
	As $N$ is stable, every blocking pair must contain $w_{\single}^*, m_{\single}^*$, $w^*$, or $m^*$.
	Agents $w_{\single}^*$ and~$m_{\single}^*$ are not part of a blocking pair, as they are matched to their first choice.
	Since $N$ is a perfect matching, no agent prefers to be matched to $m^*$ or $w^*$.
	Therefore, $M_2$ is stable with respect to~$\mathcal{P}_2$.
	Since $|N \triangle N^*| \le \ell$, it follows that $|M_1 \triangle M_2| \le 
	\ell+ 3 = k$.
	Thus, $M_2$ is a solution to the constructed \ISMT instance.
	
	$(\Leftarrow):$   
	Vice versa, let $M_2$ be a stable matching with respect to $\mathcal{P}_2 $ 
	such that $\{m^*, w^*\} \in M_2$ and $| M_1 \triangle M_2| \le k$.
	As $w_{\single}^*$ and $m_{\single}^*$ are their mutual unique first choice, it follows that $M_2$ contains~$\{w_{\single}^*, m_{\single}^*\}$.
	As $M_2$ contains $\{m^*,w^*\}$, it follows that for each agent~$a\in U \cup W$, agent~$a$ has to be matched to an agent it prefers to $m^*$ and $w^*$, i.e., an agent from $U \cup W$.
	Therefore, $N := M_2 \setminus \{ \{m^*, w^*\}, \{w_{\single}^*, m_{\single}^*\}$ is a perfect matching on~$U \cup W$.
	It is also a stable one, as any blocking pair would also be a blocking pair for $M_2$.
	Furthermore, $|N \triangle N^*| =|M_1 \triangle M_2|  - 3 \le \ell$. 
\end{proof}

 Second, we reduce \textsc{ISMFE-T} to \ISMT. The general idea of this parameterized reduction is to replace a forced pair $\{m,w\}\in Q$ by a gadget consisting of~$6(k+1)$~agents.
 In~$M_1$, we match the agents from the gadget in a way such that if  ``$m$ and~$w$ are not matched to each other'' in $M_2$, then, compared to $M_1$, the matching in the whole gadget needs to be changed, thereby exceeding the given budget $k$.
 This reduction implies: 
 
\begin{restatable}{theorem}{ISMWties}
	\label{co:ISMWties}
	\ISMT parameterized by $k$ and the summed number of ties in $\mathcal{P}_1$ and $\mathcal{P}_2$  is W[1]-hard,
	even if $|\mathcal{P}_1 \oplus \mathcal{P}_2|=1$ and only 
	women have ties in their preferences.
	
 	Parameterized by $k$, \ISMT is W[1]-hard, even if $|\mathcal{P}_1 \oplus \mathcal{P}_2| = 1 $ and each tie has size at most two and is in the preferences of some women.
\end{restatable}
\begin{proof}
	We show the theorem by reducing \textsc{Incremental Stable Marriage with Forced Edges and Ties} to \ISMT 
	without increasing $k$, $|\mathcal{P}_1 \oplus \mathcal{P}_2|$, and the 
	number and length of ties (the theorem then follows from \Cref{th:ISMFE}).
	Let $\mathcal{I} = (A,\mathcal{P}_1, \mathcal{P}_2, M_1, k, Q)$ be an 
	instance 
	of \textsc{Incremental Stable Marriage with Forced Edges and Ties}. From $\mathcal{I}$ we construct an instance~$\mathcal{I}' = 
	(A',\mathcal{P}'_1, 
	\mathcal{P}'_2, M'_1, k')$ of \ISM:
	We create preference profiles~$\mathcal{P}_1'$ and $\mathcal{P}_2'$ by 
	replacing every forced pair~$e=\{v,w\}\in Q$ by a gadget inspired by a gadget 
	designed by 
	Cechl{\'{a}}rov{\'{a}} and Fleiner~\cite{DBLP:journals/talg/CechlarovaF05}  
	which allows to replace a pair by a cycle of six agents (see 
	\Cref{fig:eg}; note that the superscripts of each agent describes its 
	position in the gadget).
	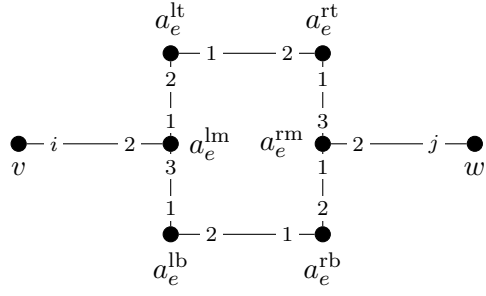
\begin{figure}[bt]
		\begin{center}
			\begin{tikzpicture}[xscale =2 , yscale = 1.2]
			\node[vertex, label=270:$v$] (m) at (-1, 0) {};
			\node[vertex, label=0:$a_{e}^{\lm}$] (wf) at (0, 0) {};
			\node[vertex, label=90:$a_{e}^{\lt}$] (mt) at (0, 1) {};
			\node[vertex, label=270:$a_{e}^{\lb}$] (mb) at (0, -1) {};
			\node[vertex, label=90:$a_{e}^{\rt}$] (wt) at (1, 1) {};
			\node[vertex, label=270:$a_{e}^{\rb}$] (wb) at (1, -1) {}; 
			\node[vertex,
			label=180:$a_{e}^{\rrm}$] (mf) at (1, 0) {};
			\node[vertex, label=270:$w$] (w) at (2, 0) {};
			\draw (m) edge node[pos=0.2, fill=white, inner sep=2pt] {\scriptsize
				$i$}  node[pos=0.76, fill=white, inner sep=2pt] {\scriptsize $2$}
			(wf);
			\draw (mf) edge node[pos=0.2, fill=white, inner sep=2pt] {\scriptsize
				${2}$}  node[pos=0.76, fill=white, inner sep=2pt]
			{\scriptsize $j$} (w);
			\draw (mf) edge node[pos=0.2, fill=white, inner sep=2pt] {\scriptsize
				${3}$}  node[pos=0.76, fill=white, inner sep=2pt]
			{\scriptsize $1$} (wt);
			\draw (mf) edge node[pos=0.2, fill=white, inner sep=2pt] {\scriptsize
				$1$}  node[pos=0.76, fill=white, inner sep=2pt] {\scriptsize $2$}
			(wb);
			\draw (wf) edge node[pos=0.2, fill=white, inner sep=2pt] {\scriptsize
				$3$}  node[pos=0.76, fill=white, inner sep=2pt] {\scriptsize $1$}
			(mb);
			\draw (wf) edge node[pos=0.2, fill=white, inner sep=2pt] {\scriptsize
				$1$}  node[pos=0.76, fill=white, inner sep=2pt] {\scriptsize $2$}
			(mt);
			\draw (wt) edge node[pos=0.2, fill=white, inner sep=2pt] {\scriptsize
				$2$}  node[pos=0.76, fill=white, inner sep=2pt] {\scriptsize $1$}
			(mt);
			\draw (wb) edge node[pos=0.2, fill=white, inner sep=2pt] {\scriptsize
				$1$}  node[pos=0.76, fill=white, inner sep=2pt] {\scriptsize $2$}
			(mb);
			\end{tikzpicture}
			
		\end{center}
		\caption{Gadget by 
			Cechl{\'{a}}rov{\'{a}} and Fleiner~\cite{DBLP:journals/talg/CechlarovaF05} for pair $e = \{v, w\}$, where $v$ ranks $w$ at the $i$-th position and $w$ ranks $v$ at the $j$-th position.}\label{fig:eg}
	\end{figure}
	Our gadget replaces every forced pair by concatenating $k+1$ copies of 
	the gadget by Cechl{\'{a}}rov{\'{a}} and Fleiner~\cite{DBLP:journals/talg/CechlarovaF05}  
	(see \Cref{fig:weigthed-edge}).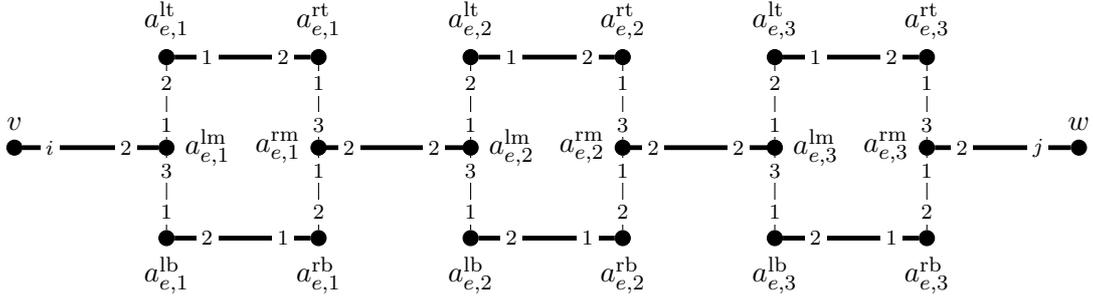
\begin{figure}[bt]
		\begin{center}
			\begin{tikzpicture}[xscale =2 , yscale = 1.2]
			\node[vertex, label=90:$w$] (w) at (6, 0) {};
			
			\draw[bedge] (1, 0) edge node[pos=0.2, fill=white, inner sep=2pt] 
			{\scriptsize
				${2}$}  node[pos=0.76, fill=white, inner sep=2pt]
			{\scriptsize $2$} (2,0);
			\draw[bedge] (3,0) edge node[pos=0.2, fill=white, inner sep=2pt] 
			{\scriptsize
				${2}$}  node[pos=0.76, fill=white, inner sep=2pt]
			{\scriptsize $2$} (4,0);
			
			\node[vertex, label=90:$v$] (m) at (-1, 0) {};
			\node[vertex, label=0:$a_{e,1}^{\lm}$] (wf) at (0, 0) {};
			\node[vertex, label=90:$a_{e,1}^{\lt}$] (mt) at (0, 1) {};
			\node[vertex, label=270:$a_{e,1}^{\lb}$] (mb) at (0, -1) {};
			\node[vertex, label=90:$a_{e,1}^{\rt}$] (wt) at (1, 1) {};
			\node[vertex, label=270:$a_{e,1}^{\rb}$] (wb) at (1, -1) {}; 
			\node[vertex,
			label=180:$a_{e,1}^{\rrm}$] (mf) at (1, 0) {};
			\draw[bedge] (m) edge node[pos=0.2, fill=white, inner sep=2pt] 
			{\scriptsize
				$i$}  node[pos=0.76, fill=white, inner sep=2pt] {\scriptsize $2$}
			(wf);
			\draw (mf) edge node[pos=0.2, fill=white, inner sep=2pt] {\scriptsize
				${3}$}  node[pos=0.76, fill=white, inner sep=2pt]
			{\scriptsize $1$} (wt);
			\draw (mf) edge node[pos=0.2, fill=white, inner sep=2pt] {\scriptsize
				$1$}  node[pos=0.76, fill=white, inner sep=2pt] {\scriptsize $2$}
			(wb);
			\draw (wf) edge node[pos=0.2, fill=white, inner sep=2pt] {\scriptsize
				$3$}  node[pos=0.76, fill=white, inner sep=2pt] {\scriptsize $1$}
			(mb);
			\draw (wf) edge node[pos=0.2, fill=white, inner sep=2pt] {\scriptsize
				$1$}  node[pos=0.76, fill=white, inner sep=2pt] {\scriptsize $2$}
			(mt);
			\draw[bedge] (wt) edge node[pos=0.2, fill=white, inner sep=2pt] 
			{\scriptsize
				$2$}  node[pos=0.76, fill=white, inner sep=2pt] {\scriptsize $1$}
			(mt);
			\draw[bedge] (wb) edge node[pos=0.2, fill=white, inner sep=2pt] 
			{\scriptsize
				$1$}  node[pos=0.76, fill=white, inner sep=2pt] {\scriptsize $2$}
			(mb);
			
			\begin{scope}[xshift = 2cm]
			\node[vertex, label=0:$a_{e,2}^{\lm}$] (wf) at (0, 0) {};
			\node[vertex, label=90:$a_{e,2}^{\lt}$] (mt) at (0, 1) {};
			\node[vertex, label=270:$a_{e,2}^{\lb}$] (mb) at (0, -1) {};
			\node[vertex, label=90:$a_{e,2}^{\rt}$] (wt) at (1, 1) {};
			\node[vertex, label=270:$a_{e,2}^{\rb}$] (wb) at (1, -1) {}; 
			\node[vertex,
			label=180:$a_{e,2}^{\rrm}$] (mf) at (1, 0) {};
			\draw (mf) edge node[pos=0.2, fill=white, inner sep=2pt] {\scriptsize
				${3}$}  node[pos=0.76, fill=white, inner sep=2pt]
			{\scriptsize $1$} (wt);
			\draw (mf) edge node[pos=0.2, fill=white, inner sep=2pt] {\scriptsize
				$1$}  node[pos=0.76, fill=white, inner sep=2pt] {\scriptsize $2$}
			(wb);
			\draw (wf) edge node[pos=0.2, fill=white, inner sep=2pt] {\scriptsize
				$3$}  node[pos=0.76, fill=white, inner sep=2pt] {\scriptsize $1$}
			(mb);
			\draw (wf) edge node[pos=0.2, fill=white, inner sep=2pt] {\scriptsize
				$1$}  node[pos=0.76, fill=white, inner sep=2pt] {\scriptsize $2$}
			(mt);
			\draw[bedge] (wt) edge node[pos=0.2, fill=white, inner sep=2pt] 
			{\scriptsize
				$2$}  node[pos=0.76, fill=white, inner sep=2pt] {\scriptsize $1$}
			(mt);
			\draw[bedge] (wb) edge node[pos=0.2, fill=white, inner sep=2pt] 
			{\scriptsize
				$1$}  node[pos=0.76, fill=white, inner sep=2pt] {\scriptsize $2$}
			(mb);
			
			\end{scope}
			
			\begin{scope}[xshift = 4cm]
			\node[vertex, label=0:$a_{e,3}^{\lm}$] (wf) at (0, 0) {};
			\node[vertex, label=90:$a_{e,3}^{\lt}$] (mt) at (0, 1) {};
			\node[vertex, label=270:$a_{e,3}^{\lb}$] (mb) at (0, -1) {};
			\node[vertex, label=90:$a_{e,3}^{\rt}$] (wt) at (1, 1) {};
			\node[vertex, label=270:$a_{e,3}^{\rb}$] (wb) at (1, -1) {}; 
			\node[vertex,
			label=180:$a_{e,3}^{\rrm}$] (mf) at (1, 0) {};
			\draw (mf) edge node[pos=0.2, fill=white, inner sep=2pt] {\scriptsize
				${3}$}  node[pos=0.76, fill=white, inner sep=2pt]
			{\scriptsize $1$} (wt);
			\draw (mf) edge node[pos=0.2, fill=white, inner sep=2pt] {\scriptsize
				$1$}  node[pos=0.76, fill=white, inner sep=2pt] {\scriptsize $2$}
			(wb);
			\draw (wf) edge node[pos=0.2, fill=white, inner sep=2pt] {\scriptsize
				$3$}  node[pos=0.76, fill=white, inner sep=2pt] {\scriptsize $1$}
			(mb);
			\draw (wf) edge node[pos=0.2, fill=white, inner sep=2pt] {\scriptsize
				$1$}  node[pos=0.76, fill=white, inner sep=2pt] {\scriptsize $2$}
			(mt);
			\draw[bedge] (wt) edge node[pos=0.2, fill=white, inner sep=2pt] 
			{\scriptsize
				$2$}  node[pos=0.76, fill=white, inner sep=2pt] {\scriptsize $1$}
			(mt);
			\draw[bedge] (wb) edge node[pos=0.2, fill=white, inner sep=2pt] 
			{\scriptsize
				$1$}  node[pos=0.76, fill=white, inner sep=2pt] {\scriptsize $2$}
			(mb);
			
			\draw[bedge] (mf) edge node[pos=0.2, fill=white, inner sep=2pt] 
			{\scriptsize
				${2}$}  node[pos=0.76, fill=white, inner sep=2pt]
			{\scriptsize $j$} (w);
			\end{scope}
			
			\end{tikzpicture}
			
		\end{center}
		\caption{The replacement for pair $e = \{v, w\}$, where $v$ ranks 
			$w$ at the $i$-th position and $w$ ranks $v$ at the $j$-th position, 
			for $k = 2$. The matching corresponding to matching $v$ and $w$ in the 
			original instance is marked in bold.}\label{fig:weigthed-edge}
	\end{figure}
	Formally, for each forced pair $e=\{v,w\}\in Q$, we add $6 (k+1)$ agents $a^{\lb}_{e, p}$, $a^{\lm}_{e,p}$, 
	$a^{\lt}_{e,p}$, $a^{\rb}_{e, p}$, $a^{\rrm}_{e, p}$, and  
	$a^{\rt}_{e, p}$ for $p\in [k+1]$. 
	Agent~$v$ replaces $w$ in its preferences by $a^{\lm}_{e, 1}$, and $w$ 
	replaces $v$ by~$a^{\rrm}_{e, k+1}$.
	The newly added agents have the following preferences:
	\begin{align*}
	a^{\lt}_{e, p} &: a^{\rt}_{e, p} \succ a^{\lm}_{e, p}, \qquad & p 
	\in [k+1];\\
	a^{\lm}_{e, 1} &: a^{\lt}_{e, 1} \succ v \succ a^{\lb}_{e, 1}; & 
	\\
	a^{\lm}_{e, p} &: a^{\lt}_{e, p} \succ a^{\rrm}_{e, p-1} \succ 
	a^{\lb}_{e, p}, \qquad & p \in \{2, 3, \dots, k+1\};\\
	a^{\lb}_{e, p} &: a^{\lm}_{e, p} \succ a^{\rb}_{e, p}, \qquad & p 
	\in [k+1];\\
	a^{\rt}_{e, p} &: a^{\rrm}_{e, p} \succ a^{\lt}_{e, p}, \qquad & p \in 
	[k+1];\\
	a^{\rrm}_{e, p} &: a^{\rb}_{e, p} \succ a^{\lm}_{e, p+1} \succ 
	a^{\rt}_{e, p}, \qquad & p \in [k];\\
	a^{\rrm}_{e, k+1} &: a^{\rb}_{e, k+1} \succ w \succ 
	a^{\rt}_{e, k+1}; &\\
	a^{\rb}_{e, p} &: a^{\lb}_{e, p} \succ a^{\rrm}_{e, p}, \qquad & p \in 
	[k+1].
	\end{align*}
	The preferences of the other agents remain unchanged. For the sake of 
	readability, for $e=\{v,w\}\in Q$, we define $M_e:=\{\{v, a^{\lm}_{e, 1}\}, 
	\{a^{\rrm}_{e, k+1}, w\}\}\cup  \{
	\{a^{\lt}_{e, p}, a^{\rt}_{e, p}\}, \{a^{\lb}_{e, p}, a^{\rb}_{e, 
		p}\}\mid p\in [k+1]\} \cup  \{ \{a^{\rrm}_{e, q}, a^{\lm}_{e, q+1}\}\mid 
	q\in 
	[k] \}$ (this matching is marked by bold edges in \Cref{fig:weigthed-edge}).
	
	To finish the construction of the instance~$\mathcal{I}'$ of 
	\ISM, we set $M_1' := \bigl ( M_1 
	\setminus Q\bigr) \cup \{M_e \mid e\in Q\}$ and $k' := k$.
	Note that by definition we have $Q\subseteq M_1$ and thus $M_1'$ is a matching and is stable in $\mathcal{P}'_1$. 
	
	It remains to show that $\mathcal{I}$ and $\mathcal{I}'$ are equivalent.
	Given solution~$M_2$ to $\mathcal{I}$, we get a solution~$M_2'$ 
	to~$\mathcal{I}'$ by replacing every forced pair~$e\in Q$ (note that $e$ needs to 
	be part of $M_2$) by~$M_e$.
	Clearly, it holds that $|M_1' \triangle M_2'| = |M_1 \triangle M_2| \le k = k'$.
	It is straightforward to verify that $M_2$ is indeed stable.
	
	Given a solution~$M_2'$ to~$\mathcal{I}'$, first observe that every stable
	matching~$M$ in $\mathcal{P}_2'$ which does not contain both~$\{v, a^{\lm}\}$ and~$\{a^{\rrm}, w\}$ for some forced pair~$e =\{v, w\} \in Q$ contains none of 
	the pairs~$\{v, a^{\lm}\}$, $\{a^{\rrm}, w\}$, and $\{a^{\rrm}_{e, p}, 
	a^{\lm}_{e, p+1}\}$ for $p \in [k]$.
	As $|M_1' \triangle M_2' | \le k$, it follows that
	$M_2'$ contains pairs~$\{v, a^{\lm}_{e, 1}\}$, $\{w, a^{r, 
		m}_{e, k+1}\}$, and $\{a^{\rrm}_{e, p}, a^{\lm}_{e, p+1}\}$ for every 
	$p \in [k]$ for every forced pair~$\{v, w\} \in Q$.
	Then $M_2 := \bigl ( M_2' \setminus \{M_e\mid 
	e\in Q\}\bigr) \cup Q$ is a stable 
	matching in $\mathcal{P}_2$.
	We have $|M_1 \triangle M_2| = |M_1 ' \triangle M_2'| \le k' =k$, and thus 
	$M_2$ is a solution to $\mathcal{I}$.
\end{proof}

On the algorithmic side, parameterized by 
the number of agents with at least one tie 
in their preferences in $\mathcal{P}_2$, \ISMT lies 
in XP. The idea of our algorithm is to first guess the partners of all agents in $M_2$ with a tie in their preferences in $\mathcal{P}_2$ and subsequently reduce the problem to an instance of \textsc{Weighted Stable Marriage}, which is polynomial-time solvable~\cite{DBLP:journals/jcss/Feder92}. 
Moreover, parameterizing by the summed size of all ties results in fixed-parameter tractability, as we can iterate over all possibilities of breaking the ties and subsequently apply the algorithm for \ISM. 

\begin{restatable}{proposition}{ISMNT}
	\label{pr:ISMNT}
	\ISMT parameterized by the number of agents with at least one tie 
	in their preferences in $\mathcal{P}_2$ lies in XP.
	\ISMT parameterized by the summed size of all ties in~$\mathcal{P}_2$ is fixed-parameter tractable.
\end{restatable}
\begin{proof}
  For the FPT-algorithm, we just enumerate all possibility of breaking the ties and subsequently apply the algorithm for \ISM.
  
  We now turn to the XP-algorithm.
	For each agent $a\in A$ with ties in their preferences in~$\mathcal{P}_2$, 
	we guess its partner in $M_2$. Let $B\subseteq 
	A$ be the set of agents who do not have an assigned partner in $M_2$ after 
	this. 
	To decide how these agents are matched, we now construct an instance of  
	\textsc{Weighted Stable Marriage} with agent set $B$.
	For every~$a\in A\setminus B$ and $b \in B$ with $b \succ_a M_2 (a)$, we modify the 
	preferences of $b$ by deleting each agent~$b' \in B$ with $a \succ_b b'$.
	The reason for this is that $ b$ needs to be matched better than $a$ in the stable matching $M_2$, as otherwise $\{a,b\}$ blocks $M_2$. 
	Finally, we delete all agents from $A\setminus B$ from the 
	preferences of 
	agents from 
	$B$, as we already know the partners of agents from $A\setminus B$ in $M_2$ and in particular that none of them is matched to an agent from $B$.
	Turning to the weights, all pairs from~$M_1$ have weight one, while all other pairs have weight zero.
	Thereby, we maximize the overlap between~$M_1$ and the computed matching.
	Let~$M'$ be the computed maximum-weight stable matching in this instance.
	We set~$M_2 (b):= M' (b)  $ for every~$a \in B$.
	If $M_2$ is stable and fulfills~$|M_1 \triangle M_2| \le k$, then we return~$M_2$;
	otherwise we reject the 
	current guess.
	
	It remains to argue that if there is a solution to the problem, then our algorithm finds one. 
	Assume that $M^*$ is a stable matching in $\mathcal{P}_2$ with $|M_1 \triangle M^*| \le k$.
	Consider the guess in which we guess for each agent~$a$ with ties in its preferences that it is matched to~$M^* (a)$.
	Let $B$ be the set of agents who are unmatched after this guess and let $\mathcal{I}'$ be the \textsc{Weighted Stable Marriage} instance constructed for this guess.
	Moreover, let $N^*$ be the stable matching $M^*$ restricted to the agents from $B$. 
	Then, $N^*$ is a stable matching in  $\mathcal{I}'$ (as argued above, none of the pairs deleted in the process of creating $\mathcal{I}'$ can be included in $M^*$). 
	Note that $N^*$ has weight $|N^*\cap M_1|$ in $\mathcal{I}'$.
	Moreover, let $M'$ be the computed maximum-weight stable matching in $\mathcal{I}'$ and $M_2$ the final matching constructed by the algorithm for this guess. 
	
    We continue by showing that that $M_2$ is stable and $|M_1\triangle M_2|\leq k$.
	We start with the latter. 
	As $N^*$ is of weight $|N^*\cap M_1|$, we get that $M'$  overlaps with $M_1$ in at least $|N^*\cap M_1|$ pairs. 
	As in $\mathcal{I}'$ all agents have strict preferences, by the Rural Hospitals Theorem, $N^*$ and $M'$ match the same set of agents and thus we have that $|N^*\triangle M_1|\geq |M'\triangle M_1|$ and thus $k\geq |M^* \triangle M_1| \geq |M_2 \triangle M_1|$. 
	
	To see that $M_2$ is stable, note first of all that there cannot be a blocking pair involving two agents from $A\setminus B$, as such a pair would also block $M^*$. 
	Next, we prove that there cannot be a blocking pair involving two agents from $B$. 
	Assume towards a contradiction that $\{u, w\}$ with $u, w\in B$ is a blocking pair for~$M_2$.
	We will show that both $u$ and $w$ find each other acceptable in $\mathcal{I}'$, directly implying that $\{u,w\}$ blocks $M'$ in $\mathcal{I}'$, a contradiction to the stability of~$M'$. 
	Note that when creating $\mathcal{I}'$ for each agent $a\in B$ which is matched in $M'$ it holds that if an agent~$b$ from $B$ is deleted from $a$'s preferences, then $a$ prefers~$M' (a)$ to~$b$. 
	Thus, if $u$ is matched in~$M'$, then as $u$ prefers $w$ to $M'(u)$, agent~$u$ still finds $w$ acceptable.
	So consider the case that $u$ is unmatched in $M'$.
	Then $u$ is also unmatched in $N^*$ and $M^*$ (due to the Rural Hospitals Theorem for $\mathcal{I}'$). 
	Recall that we only delete a woman~$b$ from $B$ from the preferences of $u$ if there exists a woman~$w'\in A\setminus B$ for which we have that $u\succ_{w'} M^*(w')$ (and $w' \succ_u b$);
	however if such a women $w'$ exists, then $\{u,w'\}$ blocks $M^*$, as $u$ is unmatched in $M^*$. 
	Thus, if $u$ is unmatched, then he finds the same women from $B$ acceptable in $\mathcal{I}$ and $\mathcal{I}'$. 
	Symmetric arguments apply for~$w$, proving that $u$ and $w$ still find each other acceptable in $\mathcal{I}'$ and thus block $M'$.
	
	It remains to consider pairs $\{u,w\}$ where one of them is contained in $A\setminus B$ and the other in~$B$. 
	Without loss of generality assume that $u\in A\setminus B$ and $w\in B$. 
	Then, as $\{u,w\}$ blocks~$M_2$, we have that $w\succ_u M_2(u)=M^*(u)$. 
	Consequently, all agents that come after $u$ in the preferences of $w$ are deleted from the preferences of $w$ in $\mathcal{I}'$.
	Thus, if $w$ is matched in $M'$, then $w$ prefers $M_2(w)=M'(w)$ to $u$, implying that $\{u,w\}$ does not block $M_2$. 
	If $w$ is unmatched in $M'$, then it is also unmatched in $N^*$ and $M^*$ (due to the Rural Hospitals Theorem for $\mathcal{I}'$), implying that $\{u,w\}$ blocks $M^*$, a contradiction.
\end{proof}
 
 \section{Master Lists} \label{se:ML}
After having shown in the previous section that \ISMT and \ISR
mostly remain  intractable even if we restrict several problem-specific parameters, in this section we analyze the influence of the structure of the preference profiles by considering what happens if the agents' preferences are similar to each other. 
The arguably most popular approach  in this direction is to assume that there exists a single central order (called master list) and that all agents derive their preferences from this order. This approach has already been applied to different stable matching problems in the quest for making them tractable \cite{DBLP:conf/wine/BredereckHKN20,DBLP:journals/cn/CuiJ13,DBLP:journals/dam/IrvingMS08,DBLP:conf/atal/Kamiyama19}. 
Specifically, we analyze in \Cref{sub:oml} the case where the preferences of all agents follow a single master list, in \Cref{sub:out} the case where all but few agents have the same preference list, and in \Cref{sub:fml} the case where each agent has one of few different preference lists (which generalizes the setting considered in \Cref{sub:out}).

\subsection{One Master List}\label{sub:oml}

In an instance of \textsc{Stable Marriage/Roommates} with agent set $A$, we say that the preferences 
of agent $a\in A$ can be \emph{derived} from some preference list
$\succsim^*$ over agents $A$ if the preferences of $a$ are 
$\succsim^*$ restricted to $\Ac(a)$. 
If the preferences 
of all agents in $\mathcal{P}_2$ can be derived from the same strict preference list (which is typically called \emph{master list}), then there is a unique 
stable matching in $\mathcal{P}_2$ which iteratively matches the so-far unmatched
top-ranked agent in the master list to the highest ranked agent it accepts: 
\begin{observation} \label{ob:ML}
 If all preferences in $\mathcal{P}_2$ can be derived from the same strict preference list, then \ISR can be solved in linear time.
\end{observation}

This raises the question what happens when the master list is not a strict but a 
weak order. If the preferences of the agents may be 
incomplete, then reducing from the NP-hard \textsc{Weakly Stable Pair} problem (the question is whether there is a stable matching in an \textsc{SM-T}/\textsc{SR-T} instance containing a given pair \cite[Lemma 3.4]{DBLP:journals/dam/IrvingMS08}), one can show that even assuming that all preferences are derived from the same weak 
master list is not sufficient to make \ISMT or \ISRT
polynomial-time solvable.
\begin{restatable}{proposition}{MLtiesincomplete}
	\label{pr:ML-tiesincomplete}
	\ISMT and \ISRT are NP-hard even if 
	all preferences in $\mathcal{P}_1$ and~$\mathcal{P}_2$ can be derived from the same weak preference list. 
\end{restatable}
\begin{proof}
	As \ISRT generalizes \ISMT it suffices to prove NP-hardness 
	of the latter problem. We do so by reducing from the \textsc{Weakly Stable 
		Pair} problem where we are given a set $U$ of men and a set $W$ of women 
	with preference profile $\mathcal{P}$ and a man-woman pair $\{u,w\}\in U\times 
	W$ and the question is to decide whether there exists a stable matching $M$ 
	with  $\{u,w\} \in M$. Irving et 
	al.~\cite[Lemma 3.4]{DBLP:journals/dam/IrvingMS08} proved that \textsc{Weakly 
		Stable 
		Pair} 
	 is NP-hard even if all preferences from $\mathcal{P}$ are 
	derived from the same weak master list. Given an instance of \textsc{Weakly Stable 
		Pair} $(U\cup W,\mathcal{P}, \{u,w\})$, we construct an instance of \ISM with 
	ties as follows. For each agent $a\in U\cup W$, we add two 
	agents~$c_a$ and $d_a$.
	Agent~$d_a$ has empty preferences in~$\mathcal{P}_1$ and only accepts $d_a$ in $\mathcal{P}_2$, thereby ensuring that~$c_a$ is matched in every stable matching in~$\mathcal{P}_2$.
	If~$a \in (U\cup W) \setminus \{u, w\}$, then agent~$c_a$ has empty preferences in $\mathcal{P}_1$ and the preferences of~$a$ in~$\mathcal{P}_2$, appended by~$d_a$, i.e., in $\mathcal{P}_2$ we have $\Ac (c_a) = \{c_{b} \mid b \in \Ac (a)\} \cup \{d_a\}$ and $c_b \succ_{c_a} c_{b'}$ if and only if $b \succ_a b'$, and $c_b \succ_{c_a} d_a$ for every $b \in \Ac (a)$.
	For each agent $a\in \{u,w\}$, we add an 
	agent which has the preferences of~$a $ (where each agent $b \in \Ac (a)$ is replaced by $c_b$) in~$\mathcal{P}_1$ and~$\mathcal{P}_2$.
	We set $M_1:=\{\{u,w\}\}$ and $k:= |U| + |W| -1$. Note that every stable matching in~$\mathcal{P}_2$ has size~${|U | + |W|}$. This implies that there is a 
	solution to the constructed \ISM with ties instance if and only if there 
	exists a stable matching in $\mathcal{P}_2$ containing the pair~$\{u,w\}$.
\end{proof}

In contrast to this,  if we assume that the preferences 
of agents in $\mathcal{P}_2$ are complete and derived from a weak master list, then 
we can solve \ISMT and \ISRT in polynomial time. 
While for \ISMT this follows from a characterization of stable matchings in such instances as the perfect matchings in a bipartite graph due to Irving et 
al.~\cite[Lemma 4.3]{DBLP:journals/dam/IrvingMS08}, for \ISRT this characterization does not directly carry over.
Thus, we need a new algorithm that we present below.
For this, we define an \emph{indifference class} of a master list to be a maximal set of tied agents (note that an indifference class may consist of only one agent).
Assume that the master list consists of $q$ indifference classes and let $A_i\subseteq 
A$ be the set of agents from the $i$th indifference class (where we order the indifference class according to the master list from most preferred to least preferred).
For instance, for the master list $a\succ b\sim c\sim d \succ e \sim f$, we have $p=3$, $A_1=\{a\}$, $A_2=\{b,c,d\}$, and $A_3=\{e,f\}$.
Distinguishing between several cases, we build the matching~$M_2$ by dealing for increasing $i\in [q]$ with each tie separately while greedily maximizing the overlap of the so-far constructed matching with $M_1$. Our algorithm exploits the observation that in a stable matching, for $i\in [q]$, all 
agents from $A_i$ are matched to agents from $A_i$ except if  
(i) $|\bigcup_{j\in[i-1]} 
A_j|$ is odd in which case one agent from~$A_i$ is matched to an agent from 
$A_{i-1}$, or
(ii) if  $|\bigcup_{j\in[i]} 
A_j|$ is odd in which case one agent from~$A_i$ is matched to an agent from 
$A_{i+1}$.
\begin{restatable}{proposition}{ISRTML}
	\label{pr:ISRTML}
	If the preferences of agents in $\mathcal{P}_2$ are complete and derived from a weak master 
	list, then \ISMT/\ISRT can be solved in polynomial time.
\end{restatable}
\begin{proof}
	We first give an algorithm for \ISMT.
	Consider an 
	instance of \ISMT consisting of agents $U\cup W$, preference 
	profiles $\mathcal{P}_1$ and $\mathcal{P}_2$, and a stable matching~$M_1$ for $\mathcal{P}_1$, where the preferences 
	of agents in $\mathcal{P}_2$ are complete and derived from a weak master list. 
	We can find a stable matching as close as possible to~$M_1$ in polynomial-time as follows: Irving et 
	al.~\cite[Lemma 4.3]{DBLP:journals/dam/IrvingMS08} proved that if 
	$\mathcal{P}_2$ satisfies these constraints, then it is possible to construct 
	in polynomial-time a bipartite graph $G$ on $U\cup W$ such that the 
	perfect matchings in $G$ one-to-one correspond to stable matching of agents 
	$U\cup W$ in $\mathcal{P}_2$. 
	By assigning each edge in $G$ that is contained in $M_1$ weight one and all 
	other edges weight zero and computing a perfect maximum-weight matching, we get 
	the stable matching in $\mathcal{P}_2$ with maximum overlap and thus minimum symmetric difference with $M_1$.
	
	For \ISRT, assume that the master list consists of $q$ indifference classes and let $A_i\subseteq 
	A$ be the set of agents from the $i$th indifference class for $i\in [q]$. Given a matching $M$ 
	of agents $A$, for $A'\subseteq A$, let $M|_{A'}$ be the matching $M$ 
	restricted to agents $A'$.
	
	First of all we observe that, in a stable matching, for $i\in [q]$ all 
	agents from~$A_i$ are matched to agents from $A_i$ except if  
	\begin{enumerate*}[label=(\roman*)]
		\item $|\cup_{j\in[i-1]} 
		A_j|$ is odd in which case one agent from~$A_i$ is matched to an agent from 
		$A_{i-1}$ and/or,\label{item:C1}
		\item if  $|\cup_{j\in[i]} 
		A_j|$ is odd in which one agent from $A_i$ is matched to an agent from 
		$A_{i+1}$.\label{item:C2}
	\end{enumerate*}
	It now remains to find a matching fulfilling this constraint maximizing the 
	intersection with the given matching $M_1$. 
	
	To achieve this, we build the matching $M_2$ by iterating over the master list 
	and dealing with each indifference class~$A_i$ separately. To deal with situations where 
	agents are matched outside their indifference class, we introduce a variable $a$. After the 
	processing of the $i$th indifference classs, this variable is set to~$\{b\}$ for an agent~$b\in A_i$ if 
	$|\cup_{j\in[i]} 
	A_j|$ is odd, which implies that $b$ needs to be matched to an agent from~$A_{i+1}$, and to $\emptyset$ otherwise.
	For each indifference class~$i\in [q]$, we distinguish several
	cases. First of all, we distinguish based on whether $|\cup_{j\in[i]} 
	A_j|$ is even (Case~1) or odd (Case 2).
	In Case 1, all agents from $A_i\cup a$ 
	need to be matched among themselves and we simply match them in a way
	maximizing the overlap with $M_1$.
	In Case 2, in which one agent $a$ from~$A_i$  
	needs to be selected to be matched to 
	an agent from $A_{i+1}$, we again distinguish different cases: If there are 
	agents $b''\in A_{i+1}$ and $b'\in A_i$ with $\{b',b''\}\in M_1$ (Case 2a), we 
	match agents from~$A_i\setminus \{b'\}$ as to maximize the overlap with~$M_1$ 
	and set $a$ to~$\{b'\}$. Otherwise, we match the agents from $A_i\cup a$ in a 
	way to maximize the overlap with $M_1$, while leaving one agent from $A_i$ 
	unmatched.
	This reasoning gives rise to \Cref{algML}. 
	\begin{algorithm}[t]
		\caption{Algorithm for \ISRT with complete preferences derived from a weak master 
			list}\label{algML} 
		\begin{algorithmic}[1]
			\Input{A matching $M_1$ and sets of agents $A_1, \dots , 
				A_q$ constituting the indifference classes of the master list (with all agents from~$A_i$ being preferred to all agents from~$A_{i+1}$).}
			\Output{A stable matching~$M_2$ with $|M_1 \triangle M_2| \le k$ 
				if one 
				exists.}
			\State $M_2:=\{\}$; $a:=\emptyset$
			\For{$i = 1 $ to $q$}
			\If{$|A_i \cup a|$ is even} 
			\Comment{Case 1}
			\State Add all pairs $M_1|_{A_i\cup a}$ to 
			$M_2$ and match remaining agents from $A_i\cup a$ arbitrarily
			\State $a:=\emptyset$
			\Else 
			\If{ there are agents $b''\in A_{i+1}$ and $b'\in 
				A_i$ with $\{b',b''\}\in M_1$} \Comment{Case 2a}
			\State Add all pairs $M_1|_{(A_i\setminus\{b'\})\cup a}$ to $M_2$ and 
			match 
			remaining agents from $(A_i\setminus\{b'\})\cup a$ arbitrarily
			\State $a:=\{b'\}$
			\Else
			\If{ $M_1|_{A_i}$ is a perfect matching for $A_i$} \Comment{Case 2b}
			\State Pick an arbitrary pair $\{a',a''\}\in M_1|_{A_i}$
			\State Add $M_1|_{A_i\setminus \{a',a''\}}$ and $\{a,a'\}$ to $M_2$
			\State $a:=\{a''\}$
			\Else \Comment{Case 2c}
			\State Add $M_1|_{A_i\cup a}$ to $M_2$ and match 
			remaining agents from $A_i\cup a$ arbitrarily such that $a$ gets 
			assigned a partner
			\State Set $a:=\{b\}$ with $b$ being the agent from $A_i$ that is not 
			matched by $M_2$
			\EndIf
			\EndIf
			\EndIf
			\EndFor
			\If{$|M_1 \triangle M_2| \le k$ } Return $M_1$
			\Else { Return NO}
			\EndIf
		\end{algorithmic}
	\end{algorithm}
	
	It is easy to verify that the matching returned by \Cref{algML} satisfies 
	the above stated conditions for being a stable matching so it remains to 
	show that it maximizes the intersection between $M_1$ and a stable matching in 
	$\mathcal{P}_2$.
	Note that maximizing the intersection is equivalent to minimizing the symmetric difference because all stable matchings for~$\mathcal{P}_2$ have the same size as the preferences are complete and thus every stable matching matches every or all but one agents (depending on whether the number of agents is even).
	Let $M_2$ be the returned matching.
	We claim that 
	for all~$i\in [q]$, matching~$M_2$ simultaneously maximizes $|M_1\cap M_2 \cap \{\{a, b\}: a \in A_i, 
	b\in A_i\}|$ and $|M_1\cap M_2 \cap \{\{a, b\}: a \in A_i, b\in A_{i+1}\}|$ 
	among all stable matchings in 
	$\mathcal{P}_2$. As for all pairs~$\{a,b\}$ in a stable matching it 
	either holds that $a,b\in A_i$ or $a\in A_i$ and $b\in A_{i+1}$, from this the 
	maximality of $|M_1\cap M_2|$ follows.
	
	The maximality of~$|M_1\cap M_2 \cap \{\{a, b\}: a \in A_i, b\in A_i\}|$ is clearly ensured in Cases 1, 2a, and 2c, as in each case we add to $M_2$ the matching $M_1|_{A_i}$. 
	For Case 2b, we add all but one pair from $M_1|_{A_i}$ to $M_2$. 
	Note however that no stable matching can add all pairs from $M_1|_{A_i}$ to $M_2$, as in this case one agent from $\bigcup_{j=1}^{i-1} A_{j}$ is matched to an agent from $\bigcup_{j=i + 1}^{q}A_{j}$, resulting in a blocking pair. 

	The maximality of~$|M_1\cap M_2 \cap \{\{a, b\}: a \in A_i, b\in A_{i+1}\}|$ is ensured by always adding a pair $\{a,b\}\in M_1$ with $a\in A_i$ and $b\in A_{i+1}$ to $M_2$ in Case 2a if this is possible in a matching fulfilling the necessary criteria for stability.
\end{proof}

 \subsection{Few Outliers} \label{sub:out}

Next, we consider the case that almost all agents derive their complete 
preferences from a single strict preference list (we will call these agents 
\emph{followers}), while the remaining agents (we will call those agents 
\emph{outliers}) have arbitrary preferences.
We will show that \ISR is fixed-parameter tractable with respect to the number of outliers by showing that all stable matchings in a \textsc{Stable Roommates} instance can be enumerated in FPT time with respect to this parameter:
 
\begin{restatable}{theorem}{ISRFPTO}
	\label{pr:ISRFPTO}
	Given a \textsc{Stable Roommates} instance $(A,\mathcal{P})$ and a 
	partitioning~$F\cupdot S$ of the agents $A$ such that all agents from $F$ have 
	complete preferences that can be derived from the same strict preference list in $\mathcal{P}_2$, one can 
	enumerate all stable matchings in $(A,\mathcal{P})$ in~$\mathcal{O}(n^2 \cdot 
	|S|^{|S|+1})$~time. Consequently, \ISR is solvable in~$\mathcal{O}(n^2 \cdot 
	|S|^{|S|+1})$~time.
\end{restatable}
\begin{proof}
	We start by guessing the set~$S^*\subseteq S$ of outliers which are matched to 
	another outlier.
	For every~$a \in S^*$, we additionally guess to which agent from~$S^*$ it is matched.
	We denote this guess by~$M^* (a)$.
	
	We say that a stable matching $M$ \emph{respects} a guess $(S^*,M^*)$ if $S^*=\{s\in S \mid M(s)\in S\}$ and $M^*(a)=M(a)$ for each $a\in S^*$.
	The basic idea of our algorithm now is that we can pass through the master list and greedily find pairs contained in the stable matching.
	
	We claim that there is at most one stable matching 
	respecting the guess~$(S^*,M^*)$, and that the respective stable matching $M$ can be 
	found, if it exists, as follows.
	We start with $M := \{ \{a, M^* (a)\} : a \in S^*\}$ and delete all agents from $S^*$ from the instance including the master list. 
	Subsequently, we construct the matching for the agents from~$A\setminus S^*$ iteratively: As long as there are at least two agents unmatched by~$M$, we 
	consider 
	the first unmatched agent~$a$ from the master list. If $a\in S$, then we match $a$ to the unmatched (by~$M$) agent~$a'$ from~$A \setminus S$ which $a $ likes most (i.e., we add $\{a, a'\}$ to $M$).
	If~$a\in F$, then let $a'$ be the next 
	follower in the master list and let $b_1,\dots ,b_j$ be the outliers which appear (in that order)
	between~$a$ and~$a'$ in the master list.
	Intuitively, as $a$ is 
	the so-far unmatched agent best-ranked in the master list, $a$ cannot be 
	matched worse than $a'$ in any stable matching.
	We now compute a temporary matching~$M^{\text{temp}}$ by starting with~$M^{\text{temp}}:= M$ and for~$i=1$ to $j$, if $b_i$ is unmatched in~$M^{\text{temp}}$, we iteratively add the pair~$\{b_i, b_i^*\} $ to $M^{\text{temp}}$, where $b_i^*$ is the follower which $b_i$ likes most among all agents which are currently unmatched in~$M^{\text{temp}}$.
	If after this for-loop agent $a$ is matched in~$M^{\text{temp}}$, then we set~$M := M^{\text{temp}}$.
	Otherwise, we discard $M^{\text{temp}}$ and add~$\{a, a'\}$ to~$M$.
	When all but at most one agent are matched, we check whether the resulting 
	matching is stable.
	If this is the case, then we output the matching and afterwards proceed with the next guess.
	
    We now prove the correctness of the described algorithm.
	Assume for the sake of contradiction that there exists
	a stable matching~$M'$ respecting a guess~$(S^*, M^*)$ that differs from the  
	matching $M$ constructed for this guess. As $M$ and $M'$ both respect the same 
	guess~$(S^*, M^*)$, we have $M (a)=M' (a)$ for each~$a \in S^*$. 
	
	\begin{claim}
			Let $\hat a\in A\setminus S^*$ be the first agent in the master list that is
			matched differently in $M$ and $M'$ and let~$ b:= M (\hat a)$.
			Agent $b$ prefers~$\hat a$ to~$M'  (b)$.
	\end{claim}
	\begin{claimproof}
		First we consider the case that $b$ is a follower.
		As $b$ is not matched to 
		$\hat a$ in $M'$ and all agents appearing before $\hat a$ in the master list
		are matched the same in $M$ and $M'$, agent~$b$ prefers~$\hat a$ to $M' (b)$.
		
		It remains to consider the case that $b$ is an outlier.
		Because~$\hat a \notin S^*$, it follows that~$\hat a$ is a follower.
		Therefore, by the construction of $M$, between $\hat a$ and $b$, the master list 
		contains only followers which are matched to outliers appearing before~$\hat a$ in the master list 
		(and none of these agents can be~$M' (b)$ by the definition of~$\hat a$) and some outliers~$b_1, \dots, b_k$.
		We assume that for every~$i \in [k-1]$, outlier~$b_i$ appears before~$b_{i+1}$ in the master list.
		Note that the pair $\{\hat a,b\}$ must have been added to~$M$ by setting~$M$ to a matching $M^{\temp}$.
		Let~$X$ be the set of agents which are matched in~$M$ before the algorithm set $M$ to $M^{\temp}$ (i.e., all agents appearing before $\hat{a}$ in the master list and their partners in $M$).
		
		Note that all agents from $X$ are matched the same in $M$ and $M'$ by the definition of~$\hat{a}$. 
		Moreover, we later show that $M(b_i)=M'(b_i)$ for all $i\in [k]$. 
		From this, the statement of the claim follows as follows:
		Further note that by the construction of $M$, $\hat{a}$ is the most preferred follower by $b$ among the agents $F\setminus (X\cup \{M (b_1),\dots, M (b_k)\})$.
		As $b$ is matched to a follower in~$M'$ (due to $b\notin S^*$ and $M'$ respecting our guess)
		and all agents from $(X\cup \{M ({b_1}),\dots, M ({b_k})\})$ are matched the same in $M$ and $M'$, it follows that $b$ prefers $\hat{a}$ to $M'(b)$. 
		
		Thus, it suffices to show that $M(b_i) = M' (b_i)$ for every $i \in [k]$.
		Assume towards a contradiction that $M(b_i) \neq M' (b_i)$ for some~$i \in [k]$, and let $i$ be the minimal index fulfilling~$M(b_i) \neq M'(b_i)$.
		By the definition of~$i$, all agents from $X\cup \{M(b_1), \dots, M(b_{i-1})\}$ are matched the same in~$M$ and $M'$. 
		Thus, as $M$ matches~$b_i$ to its favorite follower from~$F\setminus (X\cup \{M(b_1), \dots, M(b_{i-1})\})$ and $b_i$ is also matched to a follower in $M'$ (as $b_i\notin S^*$), it follows that $b_i$ prefers $M (b_i)$ to $M' (b_i)$.
		Because $M (b_i)$ is a follower and $M'$ can match $M (b_i)$ only to~$\hat a$ or an agent after $b_i$ in the master list, we have that either~$\{\hat a, M (b_i)\} \in M'$ or $M(b_i)$ prefers $b_i$ to~$M'(M(b _i))$.
		In the latter case, we have that $\{b_i, M(b_i) \}$ blocks~$M'$, a contradiction to the stability of~$M'$.
		Thus, we focus on the former case.
		We claim that in this case $M (b_j) \in \{ M' (b_i), M' (b_{i+1}), \dots, M' (b_j)\}$ for every~$i < j \le k$.
		To see this, assume towards a contradiction that this is not the case and let~$j$ be the minimal such index.
		Then $M(b_j)$ prefers~$b_j$ to $M' (M (b_j))$ because $M(b_j) $ is a follower, matched neither to~$\hat{a}$ nor to one of $b_i,\dots , b_{j}$ in~$M'$, and all agents from $X\cup \{b_1, \dots, b_{i-1}\}$ are matched the same in $M$ and $M'$ (and thus in particular not to $M(b_j)$).
		Moreover, outlier~$b_j$ prefers~$M(b_j)$ to~$M' (b_j)$ because all agents from $X\cup \{M(b_1), \dots, M(b_{i-1})\}$ are matched the same in $M$ and $M'$, and $b_j$ is not matched to~$M(b_\ell)$ for $\ell < j$, and $M(b_j)$ is its first choice among the remaining agents.
		Consequently, $\{b_j, M(b_j)\}$ blocks~$M'$, a contradiction to the stability of~$M'$.
		Thus, $b$ prefers~$\hat a$ to~$M'(b)$.
	\end{claimproof}
	
	Having shown that $b$ prefers~$\hat a$ to~$M' (b)$, it is enough to show that $\hat a$ prefers~$b$ 
	to~$M' (\hat a)$, which implies that $\{\hat a, b\}$ blocks~$M'$, contradicting the stability of~$M'$: 
	\begin{claim}
		Let $\hat a\in A\setminus S^*$ be the first agent appearing in the master list that is
		matched differently in $M$ and $M'$ and  let~$ b:= M (\hat a)$.
		Agent $\hat a$ prefers~$b$ to~$M'  (\hat a)$.
	\end{claim}
	\begin{claimproof}
	If $\hat a$ is an outlier, then $\hat a$ prefers~$b$ to~$M' (\hat a)$, 
	as the algorithm matches $\hat a$ to the most preferred follower that is not matched to 
	an 
	agent which comes before $\hat a$ in the master list. Since all agents before $\hat a$ in the master list are 
	matched the same in $M$ and $M'$, agent~$\hat a$ cannot be matched better than $M(\hat a)$. 
	
	Otherwise, $\hat a$ is a follower.
	Let $X$ contain the agents appearing before $\hat{a}$ in the master list and their partners in $M$.
	For the sake of contradiction, assume that~$\hat a$ prefers~$M' (\hat a)$ to~$b$ (we will later distinguish two cases and in each cases establish a contradiction, thereby ultimately proving that $\hat a$ prefers $b$ to $M'(\hat a)$).
	Let $a'$ be the next follower in the master 
	list that is not matched to an agent appearing before $\hat a$ in the master list 
	in $M$ (and thereby also in $M'$). Further, let~$b_1,\dots ,b_k$ be the
	outliers that appear between $\hat a$ and $a'$ in the 
	master list.
	We assume that $b_i$ is before $b_{i+1}$ in the master list for $i\in [k-1]$.
	By the construction of 
	$M$, agent~$\hat a$ is matched to $b_1,\dots ,b_k$ or~$a'$ in~$M$. Thus, for $\hat a$ to 
	be matched better in $M'$ than in $M$, it needs to hold that $M'(\hat a)=b_i$ for 
	some $i\in [k]$. 
	
	We make a case distinction based on whether $b=a'$. If $b \neq a'$, then
	let $i'\in [i]$ be the smallest~$i'$ such that $b_{i'}$ is 
	matched differently in $M$ and $M'$.
	Note that $b \neq b_{i'}$, as we have assumed that $\hat a$ prefers~$M' (\hat a) $ to~$b = M(\hat a)$.
	Because $b_{i'}$ is matched to a follower in both~$M$ and $M'$ (as~$b_{i'}\notin S^*$) and all agents that appear before~$b_{i'}$ in the 
	master list expect of $\hat a$ are matched the same in $M$ and~$M'$, it follows that $M (b_{i'})$ prefers~$b_{i'}$ to~$M' (M (b_{i'}))$.
	Because $M(b_{i'})$ is the most preferred follower of~$b_{i'}$ among $F\setminus (X\cup \{M(b_1), \cdot, M(b_{i'-1})\})$ and all agents from $X\cup \{M(b_1), \cdot, M(b_{i'-1})\}$ are matched the same in $M$ and $M'$, it follows that $b_{i'}$ prefers~$M (b_{i'})$ to~$M' (b_{i'})$.
	Consequently, $\{b_{i'}, M(b_{i'})\}$ blocks~$M'$, a contradiction to $M'$ being a stable matching.
	
	Otherwise, we have $b =a'$.
	As all agents from $X$
	are matched the same in $M$ and $M'$, matching~$M'$ needs to match~$b_1$ to its 
	most preferred follower~$a_1$ from $F\setminus X$, as otherwise $b_1$ would form a blocking pair together with~$a_1$ (note that $M' (\hat a)$ is an outlier and thus~$M' (a_1) \neq \hat a$).
	By induction, one easily sees that $M'$ matches agent $b_j$ for $j\in [k]$ to its most preferred follower that is not matched to an agent 
	appearing before $b_k$ in the master list. As all agents from $X$ are matched the same in $M$ and $M'$ and we know that $M'$ matches~$\hat a$ to~$b_i$
	for some~$i \in [k]$, it follows that by construction matching~$M$ also contains~$\{\hat a, b_i\}$, a contradiction to $\hat a$ being matched differently in~$M$ and~$M'$.
\end{claimproof}
	
	Finally, we analyze the running time of the algorithm.
	First observe that the number of possible sets~$S^*$ and matchings on~$S^*$ can be upper-bounded by~$|S|^{|S|+1}$.
	Passing through the master list can be done in~$\mathcal{O}(n)$ time.
	For each agent~$\hat a$ which is the first agent of the master list which is still unmatched, we may construct a matching~$M^{\temp}$, which can be done in~$\mathcal{O}(n)$ (the remaining computations when $\hat a$ is the first agent in the master list can be done in constant time).
	Checking whether a matching is stable trivially runs in $\mathcal O(n^2)$ time.
	Thus, the running time of~$\mathcal{O}(n^2 |S|^{|S|+1})$ follows.
\end{proof}

If the master list may contains ties, then enumerating stable matchings becomes a lot more 
complicated, as here we have much more flexibility on how the agents are matched.
Specifically, even if there are no outliers, there may be exponentially many stable matchings (if the master list ties every agent, then every (near-)perfect matching is stable).
We leave it open whether there exists a 
similar fixed-parameter tractability result for a weak master list (both in the 
roommates and marriage setting).

\subsection{Few Master Lists} \label{sub:fml}
Motivated by the positive result from \Cref{sub:out}, in this section we consider the smaller parameter ``number of different preference lists''. 
Recall that \Cref{ob:ML} states that if the preference lists of all 
agents are derived from a strict master list in a \textsc{Stable Roommates} instance, then there exists only one 
stable matching (even if the preferences of the agents may be incomplete). This raises the question what happens if there exist ``few'' 
master lists and each agent derives its preferences from one of the lists.
To the best of our knowledge, the parameter ``number of master lists'' has not been considered before. However, it nicely complements (and lower-bounds) the parameter ``number of agent types'' as studied by Meeks and Rastegari~\cite{DBLP:journals/tcs/MeeksR20}. Two agents are of the same type if they have the same preferences and all other agents are indifferent between them. 
Notably, Boehmer et al.~\cite[Proposition~5]{uschanged} proved that \ISMT is fixed-parameter tractable with respect to the number of agent types. Their algorithm also works for \ISRT.
	
If the preferences of agents are incomplete, then as proven in \Cref{pr:ML-tiesincomplete}, \ISMT is already NP-hard for just one weak master list. Moreover, note that a reduction of Cseh and Manlove~\cite[Theorem 4.2]{DBLP:journals/disopt/CsehM16} implies that \ISR\ with incomplete preferences is NP-hard even if the preferences of each agent are derived from one of two strict preference lists.
The preferences of all agents in this reduction can be derived from the preference list $[\{\bar p_i : i \in [n]\}] \succ [\{p_i : i\in [n]\}] \succ [\{\bar q_i: i\in [n]\}] \succ [\{q_i: i \in [n]\}]$ or the preference list $[\{q_i: i\in [n]\}] \succ [\{\bar q_i: i\in [n]\}] \succ [\{p_i : i \in [n]\}] \succ [\{\bar q_i : i \in [n]\}]$ where for a set $A'$ of agents, $[A']$ denotes an arbitrariy but fixed strict order of $A'$.
Consequently, in this subsection we focus on the case with complete preferences.

In the following, we show that \ISR (\Cref{sec:ISR-master-lists}) and \ISMT (\Cref{sec:ISMT-masterLists}) are W[1]-hard parameterized by the number of master lists even if agent's preferences are complete. 
For the sake of readability, we sometimes only 
specify parts of the agents' preference relation and end the preferences with 
``$\pend$'', which means that all remaining agents appear afterwards in some arbitrary strict ordering. 

\subsubsection{Incremental Stable Roommates}
\label{sec:ISR-master-lists}

In contrast to the two fixed-parameter tractability results for the number of outliers (\Cref{pr:ISRFPTO}) and the number of agent types \cite{uschanged}, we show that parameterized by the number~$p$ of master lists, \ISR is W[1]-hard even if the preferences of agents are complete:

\begin{restatable}{theorem}{isrmasterlists}
\label{thm:isr-master-lists}
\ISR is W[1]-hard parameterized by the minimum number~$p$ such that in~$\mathcal{P}_2$ the preferences of each agent can be derived from one of $p$~strict preference lists, even if in $\mathcal{P}_1$ as well as in~$\mathcal{P}_2$ all agents have complete preferences.
\end{restatable}
\begin{proof}
  The basic structure of the reduction is similar to the one in \Cref{th:ISR-WP1P2}:
  We again reduce from \textsc{Multicolored Clique}, and will construct a vertex-selection gadget for every color but this time only one edge gadget for all edges together.
  Each  gadget will have a constant number of master list in $\mathcal{P}_2$, resulting in $\mathcal{O} (\ell)$ many master lists.
  Let $(G  = (V^1\cup \dots \cup V^\ell, E), \ell)$ be an instance of \textsc{Multicolored Clique}.
  Again, we will refer to elements from $[\ell]$ as colors or color classes.
  We assume that $\ell $ is even, $G$ is $r$-regular 
  and that every color class contains $\nu$ vertices.
  Let $m := \frac{\ell \cdot r \cdot \nu}{2}$ be the number of edges in~$G$.
  \paragraph*{Vertex-selection gadget}
  Note that later in the edge gadget, we will add agents $a_{e,1},a_{e,2},a_{e,3}$, and $a_{e,4}$ for each edge $e\in E$.
  For every color class, we add a vertex-selection gadget (where the preferences of each agent are derived from one of four master lists, depending on the color).
  Fix some color $c\in [\ell]$ and an arbitrary order $v_1^c, \dots, v_\nu^c$ of the vertices of color class~$c$.
  The gadget contains agents $a^c_{i,j}$ for every $i \in \{1, \dots, \largeConst\} $ and every~$j\in \{1,2,3,4\}$.
  Let $e_{i_1}^c, \dots, e_{i_{r\nu}}^c$ be the edges with one endpoint of color $c$.
  For $t\in [r\nu]$,  let~$v_{r_t}^c$ be the endpoint of $e_{i_t}^c$ of color~$c$.
  For every $s \in [m]$, we define $\mathcal A_{s, 1}^c := a^c_{(s-1)\nu + 1, 2} 
  \succ a^c_{(s-1)\nu + 2, 2} \succ \dots \succ a^c_{(s-1) \nu + r_t, 2} \succ 
  a_{e_{s},1} \succ a^c_{(s-1) \nu + r_t +1, 2} \succ a^c_{(s-1) \nu + r_t +2, 2} 
  \succ \dots \succ a_{s\nu, 2}^c$ if $s = {i_t}$ for some~$t\in [r\nu]$ and $\mathcal A_{s, 1}^c := a^c_{(s-1)\nu + 1, 2} 
  \succ a^c_{(s-1)\nu + 2, 2} \succ \dots \succ a_{s\nu, 2}^c$ otherwise (this incomplete preference list will be part 
  of the preferences of the master list for agents~$a_{i, 1}^c$). 
  Similarly, for $s\in [m]$, let
  $\mathcal{A}_{s, 2}^c :=  a^c_{(s-1)\nu + 1, 1} \succ a^c_{(s-1)\nu + 2, 1} \succ 
  \dots \succ a^c_{s \nu + 1 - r_t, 1} \succ a_{e_{s},1} \succ a^c_{s \nu + 2 - 
  r_t, 1} \succ a^c_{s \nu +3 - r_t, 1} \succ \dots \succ a_{s\nu, 
  1}^c$ if $s = i_t$ for some~$t\in [r\nu]$ and $\mathcal A_{s, 2}^c := a^c_{(s-1)\nu + 1, 1} 
  \succ a^c_{(s-1)\nu + 2, 1} \succ \dots \succ a_{s\nu, 1}^c$ otherwise (this 
  incomplete preference list 
  will be part of the preferences of the master list for agents $a_{i, 2}^c$).
  In $\mathcal{P}_2$, the preferences of the agents in the vertex-selection gadget for color class $c$ look as follows.
  \begin{align*}
    a_{i, 1}^c & : \mathcal{A}_{1,1}^c \succ \mathcal{A}_{2,1}^c \succ \mathcal{A}_{3, 1}^c \succ \dots \succ \mathcal{A}_{\largeConst, 1}^c \succ a_{1, 4}^c \succ a_{2, 4}^c \succ \dots \succ a_{\largeConst, 4}^c \pend \\
    a_{i, 2}^c & : a_{1, 3}^c \succ \dots \succ a_{\largeConst, 3}^c \succ \mathcal{A}_{1,2}^c \succ \mathcal{A}_{2,2}^c \succ \mathcal{A}_{3, 2}^c \succ \dots \succ \mathcal{A}_{\largeConst, 2}^c \pend \\
    a_{i, 3}^c &: a_{1, 4}^c \succ \dots \succ a_{\largeConst, 4}^c \succ a_{1, 2}^c \succ \dots \succ a_{\largeConst, 2}^c\pend\\
    a_{i, 4}^c &: a_{1, 1}^c \succ \dots \succ a_{\largeConst, 1}^c \succ a_{1, 3}^c \succ \dots \succ a_{\largeConst, 3}^c\pend
  \end{align*}
  We say that a vertex-selection gadget selects vertex $v_i^c$ with $i\in [\nu]$ 
  if $M_2$ contains pairs~$\{a^c_{j, 1} , a^c_{j+ (i-1) , 2}\}$ and $\{a^c_{j, 3}, 
  a^c_{j + (i -1), 4}\}$ for $j \le \largeConst - (i - 1)$ together with $\{a^c_{j, 1}, a^c_{j- 
  (\largeConst -i + 1), 4}\}$ and $\{a^c_{j, 3}, a^c_{j -(\largeConst -i + 1), 2}\}$ 
  for $j > \largeConst - i+ 1$.

  \paragraph*{Edge gadget}
  Fix an arbitrary order $e_1, \dots, e_m$ of the edges, where $e_t = \{v, w\}\in E$ with $v\in V^{c_t}$ and $w\in V^{d_t}$.
  For every edge~$e\in E$, the gadget contains agents $a_{e, 1}$, $a_{e, 2}$, $a_{e, 3}$, and $a_{e,4}$.
  The preferences of these agents are as follows.
  \begin{align*}
    a_{e, 1} & : a_{e_1, 2} \succ a_{1, 1}^{c_1} \succ a_{\nu, 2}^{c_1} \succ a_{1, 1}^{d_1}\succ a_{\nu, 2}^{d_1} \succ a_{e_1, 4} \succ a_{e_2, 2} \succ a_{\nu + 1, 1}^{c_2} \\
    & \succ a_{2\nu, 2}^{c_2}\succ a_{\nu + 1, 1}^{d_2}\succ a_{2\nu, 2}^{d_2} \succ a_{e_2, 4} \succ \dots \succ  a_{e_m, 2} \succ a_{\nu(m-1) + 1,1}^{c_m} \\
    & \succ a_{\nu m, 2}^{c_m}\succ a_{\nu (m-1) + 1,1}^{d_m} \succ a_{\nu m, 2}^{d_m}\succ a_{e_m, 4} \pend\\
    a_{e, 2} & : a_{e_1, 3} \succ a_{e_1, 1} \succ a_{e_2, 3} \succ a_{e_2, 1} \succ \dots \succ a_{e_m, 3} \succ a_{e_m, 1} \pend\\
    a_{e, 3} & : a_{e_1, 4} \succ a_{e_1, 2} \succ a_{e_2, 4} \succ a_{e_2, 2} \succ \dots \succ a_{e_m,4} \succ a_{e_m, 2} \pend\\
    a_{e, 4} & : a_{e_1, 1} \succ a_{e_1, 3} \succ a_{e_2, 1} \succ a_{e_2, 3} \succ \dots \succ a_{e_m, 1} \succ a_{e_m, 3} \pend
  \end{align*}
  The basic idea of this gadget is that for every edge~$e = \{v, w\}\in E$ with $v\in V^c$ and $w\in V^d$ for some $c,d\in [\ell]$, we may pick either pairs~$\{a_{e, 1}, a_{e, 2}\}$ and $\{a_{e, 3}, a_{e, 4}\}$ or pairs~$\{a_{e, 1}, a_{e, 4}\} $ and $\{a_{e, 3} , a_{e, 2}\}$ to be part of a stable matching.
  The first possibility will not have any intersection with~$M_1$ (thus increasing the symmetric difference by~4), while the second possibility will intersect with~$M_1$. However, it is only possible to include these pairs if the vertex-selection gadgets for the colors~$c$ and~$d$ select the endpoints~$v$ and $w$ of $e$.

  Set $M_1 :=\{ \{a_{i, j}^{2c-1}, a_{i, j}^{2c}\} : i \in [\largeConst], j \in [4], c \in [\frac{\ell}{2}]\} \cup  \{ \{a_{e, 1}, a_{e, 4}\}, \{a_{e, 2}, a_{e, 3}\} : e\in E\}$.
  As the distance between $\mathcal{P}_1$ and $\mathcal{P}_2$ is allowed to be unbounded, it is easy to construct a preference profile~$\mathcal{P}_1$ such that $M_1$ is a stable matching for~$\mathcal{P}_1$.
  Note that no stable matching for~$\mathcal{P}_2$ contains any edge from $\{ \{a_{i, j}^{2c-1}, a_{i, j}^{2c}\} : i \in [\largeConst], j \in [4], c \in [\frac{\ell}{2}]\} $. 

  Furthermore, we set $k := 2|M_1| - 4 \binom{\ell}{2}$, as $|M_1|=|M_2|$ this enforces that
  the intersection of $M_1$ and $M_2 $ contains at least $2\binom{\ell}{2}$ pairs.
  Since the constructed instance uses $4\ell + 4$ master lists and the reduction clearly runs in polynomial time, it remains to show the correctness of the reduction.
  \paragraph*{Forward Direction}
  Let $X$ be a multicolored clique.
  We construct a stable matching $M_2 $ as follows.
  For every edge~$e$ not corresponding to an edge inside the clique, we add 
  pairs~$\{a_{e, 1}, a_{e,2}\}$ and $\{a_{e, 3}, a_{e, 4}\}$ to $M_2$, while 
  for every other edge~$ e \subseteq X$, we add pairs $\{a_{e, 1}, a_{e, 4} \} $ 
  and $\{a_{e, 2}, a_{e, 3}\}$.
  For every vertex-selection gadget, we add the matching corresponding to selecting the vertex of this color class that is part of $X$. 
  Clearly, $M_1 \cap M_2 = \{\{a_{e, 1}, a_{e, 4}\}, \{a_{e, 2}, a_{e, 3}\} : e \subseteq X\}$, and this set has cardinality $2 \binom{\ell}{2}$.
  Thus, $|M_1 \triangle M_2 |= |M_1| + |M_2| - 2\cdot 2 \binom{\ell}{2} = k$ (note that $|M_1 | = |M_2|$ as the preferences in $\mathcal{P}_1$ as well as $\mathcal{P}_2$ are complete).
  It remains to show that $M_2$ is stable.
  
  First, we show by induction on $s$ that none of the agents $a_{e_s, 1}$, $a_{e_s, 2}$, $a_{e_s, 3}$, and $a_{e_s, 4}$ is contained in a blocking pair.
  For $s= 0$, there is nothing to show. 
  So assume towards a contradiction that $a_{e_s,j}$ is contained in a blocking pair for some~$s > 0$.
  Using the induction hypothesis that no agent $a_{e_s',j'}$ for some $s'<s$ and $j'\in [4]$ is part of a blocking pair,  it is easy to see that $j=1$ and that the other of agent the blocking pair needs to be from a 
  vertex gadget.
  Let $a_{p, q}^c$ for some $c\in [\ell]$, $p\in [\nu m]$, and $q\in [4]$, be the other agent of the blocking pair, and let~$v_{t}^c$ be the endpoint of~$e_s$ of color~$c$.
  Further, let~$v_i^c$ be the vertex from the multicolored clique~$X$ of color~$c$.
  We distinguish whether $q=1$ or $q=2$.
  
  First assume that $q = 1$.
  Note that agent~$a_{p, 1}^c$ prefers $a_{e_s, 1}$ to $M_2(a_{p, 1}^c)$ only if $M_2(a_{p,1}^c)=a_{p',4}^c$ for some $p'\in [\nu m]$ or $M_2 (a_{p, 1}^c) = a_{p', 2}^c$ for some~$p' > (s-1) \nu + t$. 
  These two conditions are by the construction of~$M_2$ equivalent to~$p > (s - 1) \nu + t - i + 1$.
  If $\{a_{e_s, 1}, a_{e_s, 2}\} \in M_2$, then $a_{e_s, 1}$ prefers~$a_{p, 1}^c$ to $M_2(a_{e_s,1})$ only if $p \le (s-2) \nu + 1$.
  This contradicts $p > (s-1) \nu + t -i +1$ as $i \le \nu$.
  Otherwise we have~$\{a_{e_s, 1}, a_{e_s, 4}\} \in M_2$.
  By the construction of~$M_2$, it follows that $i = t$.
  In this case agent~$a_{e_s, 1} $ prefers~$a_{p, 1}^c$ only if $p \le (s-1) \nu + 1$.
  However, this contradicts $p > (s-1) \nu + 1=(s-1)\nu +t-i+1$.
  
  Next assume~$q = 2$.
  Agent~$a_{p, 2}^c$ prefers $a_{e_s, 1}$ to $M(a_{p, 2}^c)$ only if $M_2 (a_{p, 2}^c) = a_{p', 1}^c$ for some~$p' > s \nu + 1 - t$ which  by the construction of~$M_2$ implies that~$p > s \nu + i - t$.
  If $\{a_{e_s, 1}, a_{e_s, 2}\} \in M_2$, then $a_{e_s, 1}$ prefers~$a_{p, 2}^c$ to $M(a_{e_s, 1)}$  only if $p \le (s-1) \nu$.
  This contradicts $p > s \nu + i -t $ as $t \le \nu$.
  Otherwise we have~$\{a_{e_s, 1}, a_{e_s, 4}\} \in M_2$.
  By the construction of~$M_2$, it follows that $i = t$.
  Agent~$a_{e_s, 1} $ prefers~$a_{p, 2}^c$ to $M(a_{e_s, 1})$ only if $p \le s \nu$.
  This contradicts $p > s \nu=s\nu+i-t$.
  It follows that there cannot be a blocking pair involving an agent from the edge gadget.
  
  Finally, we show that there is no blocking pair involving an agent from a vertex-selection gadget.
  So assume for a contradiction that there is a blocking pair $\{a_{p, q}^c, b\}$ for some $p\in [\nu m]$, $q\in [4]$, and $c\in [\ell]$.
  As we have preciously shown that $b$ cannot be contained in the edge gadget, $b$ has to be contained in the 
  same vertex-selection gadget, i.e., $b = a_{p', q'}^c$ for some $p' \in [\largeConst]$ 
  and $q'\in [4]$.
  We may assume that $q\in \{1, 3\}$ and $q' \in \{2, 4\}$.
  We assume $q=1$; the case $q=3$ is symmetric.
  Let $v_i^c$ be the vertex of color $c$ which is contained in the multicolored clique $X$.
  We make a case distinction based on the value of $p$.
 
  If $p \le \largeConst - (i -1)$,
  then the only agents from the vertex selection gadget 
  which $a_{p,1}^c$ prefers to $M_2(a^c_{p,1})=a_{p+(i-1),2}$ are agents~$a_{x, 2}^c$ with $ x < p + i 
  -1$.
  However, every such agent~$a_{x, 2}^c$ is matched to an agent~$a_{y, 3}^c$ 
  for some~$y\in [\largeConst]$ or an agent $a_{z, 1}^c$ 
  with $z < p$.
  Agent $a_{x, 2}^c$ prefers these agents to $a_{p, 1}^c$, a 
  contradiction to $\{a_{p,1}^c, a_{x, 2}^c\}$ being blocking.
  
  If $p > \largeConst -i + 1$, then the only agents which $a_{p, 1}^c$ prefers to 
  $M_2(a^c_{p,1})=a_{p-(\nu m -i+1),4}$ are 
  agents $a_{x, 2}^c$ for $x\in[\largeConst]$
  or agents 
  $a^c_{y, 4} $ with $y < p - (\largeConst - i + 1)$.
  Every agent $a_{x, 2}^c$ is matched to an agent $a^c_{z, 1} $ with $z \le \largeConst - (i - 1)
  < p$ or $a_{z, 3}^c$ for some~$z\in [\largeConst]$ and thus does not prefer 
  $a_{p, 1}^c$ to $M_2(a_{x, 2}^c)$.
  Agent $a_{y, 4}^c$ with $y < p - (\largeConst -i +1)$ is 
  matched to an 
  agent $a_{z, 1}^c$ with $z < p$ and thus prefers $M_2(a_{y, 4}^c)$ to $a_{p, 1}^c$.
  Therefore, $a_{p, 1}^c$ is not part of a blocking pair, a contradiction.
  Thus, $M_2$ is stable.
  
  \paragraph*{Backward Direction}
  Let $M_2$ be a stable matching with $|M_1 \triangle M_2| \le 2|M_1| - 4\binom{\ell}{2}$.
  Note that $M^* := \{ \{a_{i, 1}^{c}, a_{i, 2}^{c}\}, \{a_{i, 3}^c, a_{i, 4}^c\} : i \in [\largeConst], c \in [\ell]\} \cup  \{ \{a_{e, 1}, a_{e, 2}\}, \{a_{e, 3}, a_{e, 4}\} : e\in E\}$ is a stable matching in~$\mathcal{P}_2$.
  By \Cref{lem:circular-prefs}, no stable matching in $\mathcal{P}_2$ can contain a pair~$\{a, b\}$ with both $a$ preferring~$M^* (a)$ to~$b$ and $ b$ preferring~$M^* (b)$ to~$a$.
  Thus, every stable matching~$M_2$ in $\mathcal{P}_2$ is also stable in the instance~$\mathcal{I}^{\pend}$ which arises from~$\mathcal{P}_2$ by deleting all edges~$\{a, b\}$ with $b$ appearing in the~$\pend $ part of the preferences of~$a$ and $a$ appearing in the~$\pend $ part of the preferences of~$b$.
  Further, as $M^*$ is perfect in $\mathcal{I}^{\pend}$, the Rural Hospitals Theorem for \textsc{SR} \cite{DBLP:books/daglib/0066875} implies that $M_2$ is perfect in $\mathcal{I}^{\pend}$.
  
  First, we show by induction on $s$ that for every edge~$e_s$, matching $M_2$ contains pairs~$\{a_{e_s, 1}, a_{e_s, 2}\}$ and $\{a_{e_s, 3}, a_{e_s, 4}\}$ or pairs $\{a_{e_s, 1}, a_{e_s, 4}\}$ and $\{a_{e_s, 2}, a_{e_s, 3}\}$.
  In $\mathcal{I}^{\pend}$, agents~$\{a_{e, 2}, a_{e, 4} : e\in E\}$ are only incident to~$\{a_{e,1}, a_{e, 3}: e \in E\}$.
  As $M_2$ is perfect in $\mathcal{I}^{\pend}$, each agent from~$\{a_{e, 2}, a_{e, 4} : e\in E\}$ is matched to an agent from~$\{a_{e,1}, a_{e, 3}: e \in E\}$.  We now turn to the induction.
  For $s = 0$ there is nothing to show.
  Fix~$s > 0$.
  If $M_2$ contains neither~$\{a_{e_s, 1}, a_{e_s, 2}\}$ nor $\{a_{e_s, 1}, a_{e_s, 4}\}$,
  then $\{a_{e_s, 1}, a_{e_s, 4}\}$ blocks~$M_2$ (note that $a_{e_s,1} $ cannot be matched to an agent it prefers to~$a_{e_s, 4}$ by the induction hypothesis and as~$a_{e_s,1}$ must be matched to an agent~$a_{e_j, q}$ for some~$j \in [m]$ and $q\in \{2,4\}$; the same holds for $a_{e_s,4}$), a contradiction to~$M_2$ being stable.
  If $M_2$ contains neither~$\{a_{e_s, 3}, a_{e_s, 2}\}$ nor $\{a_{e_s, 3}, a_{e_s, 4}\}$,
  then $\{a_{e_s, 3}, a_{e_s, 2}\}$ blocks~$M_2$, a contradiction to~$M_2$ being stable.
  
  Let us fix some color $c\in [\ell]$. Since $M_2$ matches no agent from a vertex-selection gadget to an agent from the edge gadget, it follows that $M_2$ matches every agent $a_{p, 1}^c$ or $a_{p, 3}^c$ to an 
  agent~$a_{p', 2}^c $ or $a_{p', 4}^c$ (if there existed an agent $a_{p^*, 
  q^*}^c$ with $q^* \in \{1,3\}$ for which this does not hold, then there also 
  exists an agent $a_{p', q'}$ with $q'\in \{2, 4\}$ which is not matched to an 
  agent $a_{p'', q''}^c$ with $q'' \in \{1,3\}$, and then $\{a_{p^*, q^*}^c, 
  a_{p', q'}^c\}$ forms a blocking pair).
  If $M_2 (a_{1,1}^c) = a_{i, 4}^c$ for some $i \in [\largeConst]$,
  then $\{a_{1, 1}^c, a_{e, 1}\}$ where~$e\in E\setminus \{e_1\}$ is incident to a vertex of color~$c$  blocks~$M_2$,
  contradicting the 
  stability of $M_2$.
  Thus we have $M_2 (a_{1, 1}^c) = a_{i, 2}^c$ for some $i \in [\largeConst]$.
  We now continue by proving a more detailed statement about $M_2$.
  \begin{claim} \label{claim:3}
  For each $c\in [\ell]$, we have $M_2 (a_{1, 1}^c) = a_{i, 2}^c$ for some $i \in [\largeConst]$.
   Moreover, we have $\{ \{a_{j, 1}^c, a_{j + (i - 1), 2}^c\}, \{a_{j, 3}^c, a_{j + (i -1), 4}\} 
   : j\in [\largeConst  + 1 - i]\} \cup \{\{a_{j, 1}^c, a_{j- (\largeConst -i + 1), 4}\},\allowbreak \{a_{j, 
   3}^c, a_{j-(\largeConst -i+1),2}^c \}: j \in [\largeConst -i +2, \largeConst]\} \subseteq M_2$.
  \end{claim}

  \begin{proof}
   We have already observed above that the first part holds and now argue that the second part follows from the first part. 
   
   Note that $M_2$ cannot contain pairs $\{a_{j, p}^c, a_{s, q}^c\}$ and $\{a_{j', p}^c, a^c_{s', q}\}$ with~$p \in \{1, 3\}$, $q \in \{2, 4\}$, and both $j >j' $ and $s < s'$, since otherwise $\{a_{j', p}^c, a_{s, q}^c\}$ would be blocking.
   Further note that for every $j> i$, agent $a_{j, 2}^c$ needs to be matched to an agent $a_{j', 1}^c$ for some $j'\in [\nu m]$, as otherwise it is matched to some $a_{\cdot,3}^c$ and $\{M_2 (a_{j,2}^c), a_{i,2}^c\}$ would be blocking.
   Lastly, note that, for every~$j < i$, agent $a_{j, 2}^c $ needs to be matched to an agent~$a_{j', 3}^c$ for some $j'\in [\nu m]$, as otherwise $\{a_{1,1}^c, a_{j, 2}^c\}$ would be blocking.
   From these three observations it follows that $M_2$ contains pair $\{a_{j, 1}^c, a_{j + (i-1), 2}^c\}$ for every $j \in [\largeConst +1 -i]$. 
   
   It follows that only $i-1$ agents $a_{\cdot ,2}$ can be matched to some $a_{\cdot,3}$.
   Thus, in case that $j < i$ where agent $a_{j, 2}^c $ is matched to an agent~$a_{j', 3}^c$ for some $j'\in [\nu m]$ we can conclude that $j'  > \largeConst + 1 -i$:
   Otherwise there exists some $j^*> j' $ such that $a_{j^*, 3}^c$ is not matched to some $a_{\cdot, 2}$ and thereby $\{a_{j^*, 3}^c, a_{t, 4}^c\} \in M_2$ for some $t \in [\largeConst]$. 
   In this case, $\{a_{j', 3}^c, a_{t,4}^c\}$ would be a blocking pair.
   Consequently, $M_2$ contains pair $\{a_{j, 3}^c, a_{j - (\largeConst - i + 1), 2}^c\}$ 
   for every $j \in [\largeConst -i + 2, \largeConst]$. 
   
   The remaining $i-1$ agents $a_{\cdot, 1}^c$ as well as the remaining $\largeConst - i +1$ agents~$a_{\cdot,3}^c$ need to be matched to agents $a_{\cdot, 4}^c$ by~$M_2$.
   It is straightforward to see that $M_2$ contains pairs $\{a_{j, 3}^c, a^c_{j + (i -1), 4}\}$ for $j \in [\largeConst + 1 -i]$ and $\{a_{j, 1}^c, a^c_{j- (\largeConst -i + 1), 4}\}$ for $ j \in [\largeConst -i +2, \largeConst]$.
  \end{proof}

  Let $M_2$ contain pair~$\{a_{e_{x}, 1}, a_{e_{x}, 4}\}$ for an edge $e_x = \{v_{y}^c, v_z^{d}\}$ with $y,z\in [\nu]$ and $v_y^c\in V^c$ and $v_z^d\in V^d$.
  Then both $a_{(x-1) \nu + 1, 1}^{c}$ and $a_{x\nu, 2}^c$ prefer their partner in $M_2$ to  
  $a_{e_x, 1}$, as otherwise they form a blocking pair for $M_2$ together with  $a_{e_x, 1}$.
  From $a_{(x-1) \nu + 1, 1}^c$ preferring its partner in~$M_2$ to $a_{e_x, 1}$ it follows that $a_{(x-1) \nu + 1, 1}^c$ is matched to some agent $a_{\cdot, 2}$. 
  By \Cref{claim:3} we get that  
  $M_2 (a_{(x-1) \nu + 1, 1}^c) = a_{(x-1) \nu + 1 + (i-1), 2}^c$ for some $i\in [\nu m]$ with $i\leq (m-x+1)\nu$. 
  As $a_{(x-1) \nu + 1, 1}^c$ prefers  $a_{(x-1) \nu + 1 + (i-1), 2}^c$ to $a_{e_x, 1}$  it needs to hold that $(x-1) \nu + 1 + (i-1) \le (x-1) \nu + y$, which implies
  $i \le y$. 
  Moreover, by \Cref{claim:3}, we get that $M(a^c_{x\nu,2})=a_{x\nu-i+1, 1}$ (as clearly $x\nu-i+1\leq m\nu+1-i$).
  As $a^c_{x\nu,2}$ prefers $a^c_{x\nu-i+1, 1}$ to $a_{e_x, 1}$ which is directly after $a^c_{x\nu - y+ 1,1}$ in the preferences of~$a^c_{x\nu , 2}$ it needs to hold that $x \nu - i 
  + 1\le x\nu - y+1$ which is equivalent to $y \le i$.
  Thus, we have~$i = y$, i.e., 
  the corresponding vertex gadget needs to select the vertex~$v^c_{y}$.
  Using symmetric arguments for $a_{(x-1) \nu + 1, 1}^d$ and $a_{x \nu, 2}^d$, we have that for every edge~$e_x$ 
  with $\{e_{x, 1}, e_{x, 4} \} \in M_2$, both endpoints have to be selected by 
  vertex-selection gadgets.
  
  The only edges which may contribute to~$M_1 \cap M_2$ are~$\{a_{e, 1}, a_{e,4}\} $ and $\{a_{e,2}, a_{e, 3}\}$ for some~${e \in E}$.
  By the definition of~$k$, there must be at least $\binom{\ell}{2}$ edges~$e $ such that $\{a_{e, 1}, a_{e,4}\}, \{a_{e,2}, a_{e, 3}\} \in M_2$.
  Using the above arguments it follows that the endpoint of these edges form a multicolored clique.
\end{proof}

Containment of this problem in~XP is an intriguing open question; in other words, is there 
a polynomial-time algorithm if the number of master lists is constant? 

\subsubsection{\textsc{Incremental Stable Marriage with Ties}}\label{sec:ISMT-masterLists}

Recall that \ISMT is polynomial-time solvable if agents have complete preferences derived from one weak master list (\Cref{pr:ISRTML}). 
Motivated by this, we now ask whether \mbox{\ISMT} parameterized by the number of master lists is fixed-parameter tractable for agents with complete preferences (similar to \Cref{sec:ISR-master-lists} for \ISR).
Using a similar but slightly more involved reduction than for \Cref{thm:isr-master-lists} for \ISR, we show that this problem is W[1]-hard parameterized by the number of master lists. 

\begin{restatable}{theorem}{ISMWML}
\label{pr:ISMWML}
\ISMT is W[1]-hard parameterized by the minimum number~$p$ such that in~$\mathcal{P}_2$ the preferences of each agent can be derived from one of $p$~weakly ordered preference lists, even if in $\mathcal{P}_1$ as well as in~$\mathcal{P}_2$ all agents have complete preferences.
\end{restatable}
Similar to \ISR, it remains open whether \ISMT for a constant number of master lists is polynomial-time solvable or NP-hard. 

To show \Cref{pr:ISMWML}, we adapt the reduction from above for \ISR.
  The underlying idea is that removing the pairs $\{a_{p, 2}^c , a_{e, 1}\}$ for some $c\in [\ell]$, $p\in [\nu m]$, and $e\in E$ (and some pairs which appear only in the $\pend$ part of the preferences) from the constructed \ISR instance in \Cref{sec:ISR-master-lists} already results in an \ISM instance. 
  Moreover, adding pairs~$\{a_{p, 2}^c, a_{e, 4}\}$ to the resulting \ISM instance still maintains bipartiteness.
  Thus, if we achieve that $a_{e, 4}$ is matched ``badly'' if and only if $a_{e, 1}$ is matched ``badly'', then $a_{e,1}$ together with $a_{e, 4}$ can perform the role of~$a_{e,1 }$ in our previous reduction.
  To achieve this, we change the preferences of the edge gadget such that $a_{e, 2}$ is indifferent between $a_{e, 1}$ and $a_{e, 3}$, agent~$a_{e, 3}$ is indifferent between $a_{e, 2}$ and $a_{e, 4}$, and $a_{e, 4} $ prefers~$a_{e, 3}$ to $a_{e, 1}$.
  
  However, these changes lead to another problem:
  A stable matching could now match $a_{e,1}$ or $a_{e, 4}$ to an agent from a vertex-selection gadget.
In order to prevent this, we will use another gadget which we call \emph{Forbidden Pairs Gadget}.
The basic idea behind this gadget is the same as behind the forced pair gadget from \Cref{th:ISMFE}:
We can replace a pair~$\{a,a'\}$ by a long ``path'', which makes adding~$\{a,a'\}$ to~$M_2$ result in many pairs from $M_1$ not being contained in $M_2$, and thus, it is too ``expensive'' to add~$\{a,a'\}$ to ~$M_2$.
The forbidden pairs gadget is very similar to the forced pair gadget from \Cref{th:ISMFE};
however, it consists of~$s$ (where $s$ is a sufficiently large number) instead of $k + 1$ many repetitions of the parallel pairs gadget by Cechl{\'{a}}rov{\'{a}} and Fleiner~\cite{DBLP:journals/talg/CechlarovaF05}.
Matching~$M_1$ will contain pairs~$\{a^{\rb}, a^{\rrm}\}$ and $\{a^{\lt}, a^{\lm}\}$ for every of the parallel pairs gadgets.
Thus, taking the pair modeled by the forbidden pairs gadget into~$M_2$ results in $s$ less pairs shared with $M_1$.

We will later replace for each $c\in [\ell]$, for each $e\in E$, and $p\in [\nu m]$ the pairs $\{a_{e,1}, a_{p, 1}^c\}$  and $\{a_{e, 4}, a_{p, 2}^c\}$ by such a forbidden pair gadget. 
We now show how to construct such a forbidden pairs gadget and that  a constant number of master lists are enough to cover all the agents from the introduced forbidden pairs gadgets.

\begin{lemma}\label{lem:forbidden-edges}
Let $\mathcal{I} = (A, \mathcal{P}_1, \mathcal{P}_2, M_1,k)$  be an \ISMT instance where in $\mathcal{P}_1$ and $\mathcal{P}_2$ all agents have complete preferences.
Let~$F = \{\{v_1, w_1\}, \dots, \{v_r, w_r\}\}$ be contained in a set of pairs such that $F \cap M_1 = \emptyset$. 
Further assume that the preferences of~$v_i$ for each~${i\in [r]}$ are derived from the same master list in $\mathcal{P}_2$, the preferences of $w_i$ for each $i \in [r]$ are derived from the same master list in $\mathcal{P}_2$, and for each $k\in [r]$, it holds that $v_i \succ_{w_k} v_j$  and $w_i \succ_{v_k} w_j$ for all~$i<j\in [r]$. 
Then, by adding a so-called forbidden pairs gadget, one can construct an \ISMT instance $\mathcal{I}' = (A', \mathcal{P}_1', \mathcal{P}_2', M_1',k')$  using only~$\mathcal{O} (1)$ additional master lists such that the following holds: 

There is a stable matching $M_2$ in $\mathcal{P}_2$ with $M_2\cap F = \emptyset$ and $|M_1\triangle M_2|\leq k$ if and only if there is a stable matching $M'_2$ in $\mathcal{P}'_2$ with $|M'_1\triangle M'_2|\leq k'$.
\end{lemma}

\begin{proof}
The forbidden pairs gadget consists of $s$ parts; for our purposes it suffices to set $s:=|A|+1$. 
 The forbidden pairs gadget contains agents $a_{i, j}^{\lm}$, $a_{i, j}^{\lb}$, $a_{i, j}^{\lt}$, $a_{i, j}^{\rrm}$, $a_{i, j}^{\rb}$, and $a_{i, j}^{\rt}$ for~$i \in [r]$ and $j\in [s]$.
For each $i\in [r]$, the preferences of the agents from the forbidden pairs gadget in~$\mathcal{P}_2'$ are as follows (see \Cref{fig:forbidden-edges} for an illustration).
 \begin{align*}
    a^{\lt}_{i, j} &: \mathcal{A}^{\lt}_1 \succ \mathcal{A}^{\lt}_2 \succ \dots \succ \mathcal{A}^{\lt}_{r} \pend  & j \in [s]\\
    \mathcal{A}^{\lt}_i & : a^{\rt}_{i, 1} \succ a^{\lm}_{i, 1} \succ a^{\rt}_{i, 2} \succ a^{\lm}_{i , 2}\succ \dots \succ a^{\rt}_{i, s} \succ a^{\lm}_{r, s};\\
    a^{\lm}_{i, 1} &: a^{\lt}_{1, 1} \succ v_1 \succ a^{\lb}_{1, 1} \succ  a^{\lt}_{2, 1} \succ v_2 \succ a^{\lb}_{2, 1} \succ \dots \succ  a^{\lt}_{r, 1} \succ v_r \succ a^{\lb}_{r, 1}\pend; & 
    \\
    a^{\lm}_{i, j} &: \mathcal{A}^{\lm}_1 \succ \mathcal{A}^{\lm}_2 \succ \dots \succ \mathcal{A}^{\lm}_{r} \pend, \qquad & j \in \{2, 3, \dots, s\}\\
    \mathcal{A}^{\lm}_i &: a^{\lt}_{i, 2} \succ a^{\rrm}_{i, 1} \succ 
    a^{\lb}_{i, 2} \succ a^{\lt}_{i, 3} \succ a^{\rrm}_{i, 2} \succ 
    a^{\lb}_{i, 3} \succ \dots \succ a^{\lt}_{i,s} \succ a^{\rrm}_{i, s-1} \succ 
    a^{\lb}_{i, s} \pend &\\
    a^{\lb}_{i, j} &: \mathcal{A}^{\lb}_1 \succ \mathcal{A}^{\lb}_2 \succ \dots \succ \mathcal{A}^{\lb}_{r} \pend, \qquad & j \in [s];\\
    \mathcal{A}^{\lb}_i & : a^{\lm}_{i, 1} \succ a^{\rb}_{i, 1} \succ a^{\lm}_{i, 2} \succ a^{\rb}_{i, 2}\succ \dots \succ a^{\lm}_{i, s} \succ a^{\rb}_{i, s} \pend &\\
    a^{\rt}_{i, j} &: \mathcal{A}^{\rt}_1 \succ \mathcal{A}^{\rt}_2 \succ \dots \succ \mathcal{A}^{\rt}_{r} \pend, \qquad & j \in [s];\\
    \mathcal{A}^{\rt}_i &: a^{\rrm}_{i, 1} \succ a^{\lt}_{i, 1} \succ a^{\rrm}_{i, 2} \succ a^{\lt}_{i, 2} \succ \dots \succ a^{\rrm}_{i, s} \succ a^{\lt}_{i, s} \pend, \qquad \\
    a^{\rrm}_{i, j} &: \mathcal{A}^{\rrm}_1 \succ \mathcal{A}^{\rrm}_2 \succ \dots \succ \mathcal{A}^{\rrm}_{r} \pend, \qquad & j \in [s -1];\\
    \mathcal{A}^{\rrm}_i & : a^{\rb}_{i, 1} \succ a^{\lm}_{i, 2} \succ 
    a^{\rt}_{i, 1} \succ a^{\rb}_{i, 2} \succ a^{\lm}_{i, 3} \succ 
    a^{\rt}_{i, 2} \succ \dots \succ a^{\rb}_{i, s-1}  \succ a^{\lm}_{i, s} \succ 
    a^{\rt}_{i, s -1} \pend, \qquad \\
    a^{\rrm}_{i, s} &: a^{\rb}_{1, s} \succ w_1 \succ 
    a^{\rt}_{1, s} \succ a^{\rb}_{2, s} \succ w_2 \succ 
    a^{\rt}_{2, s} \succ \dots \succ a^{\rb}_{r, s}  \succ w_r \succ 
    a^{\rt}_{r, s} \pend; &\\
    a^{\rb}_{i, j} &:  \mathcal{A}^{\rb}_1 \succ \mathcal{A}^{\rb}_2 \succ \dots \succ \mathcal{A}^{\rb}_{r} \pend, \qquad & j \in [s ];\\
    \mathcal{A}^{\rb}_i &: a^{\lb}_{i, 1} \succ a^{\rrm}_{i, 1} \succ a^{\lb}_{i, 2} \succ a^{\rrm}_{i, 2} \succ \dots \succ a^{\lb}_{i, s} \succ a^{\rrm}_{i, s}.
 \end{align*}

  Furthermore, for each $i\in [r]$, agent~$a^{\lm}_{i, 1}$ replaces $w_i$ in the preferences of $v_i$, and $a^{\rrm}_{i, s}$ replaces~$v_i$ in the preferences of~$w_i$.
  For each agent~$a \in A$, its preferences are completed by adding all agents which are not contained in the preferences of~$a$ so far at the end of the preferences of~$a$ (in an arbitrary order).
  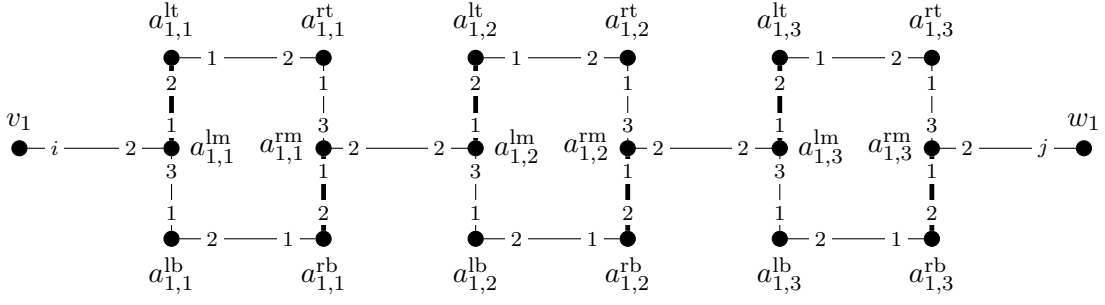
\begin{figure}[bt]
    \begin{center}
        \begin{tikzpicture}[xscale =2 , yscale = 1.2]
        \node[vertex, label=90:$w_1$] (w) at (6, 0) {};

        \draw (1, 0) edge node[pos=0.2, fill=white, inner sep=2pt] 
        {\scriptsize
            ${2}$}  node[pos=0.76, fill=white, inner sep=2pt]
        {\scriptsize $2$} (2,0);
        \draw (3,0) edge node[pos=0.2, fill=white, inner sep=2pt] 
        {\scriptsize
            ${2}$}  node[pos=0.76, fill=white, inner sep=2pt]
        {\scriptsize $2$} (4,0);

        \node[vertex, label=90:$v_1$] (m) at (-1, 0) {};
        \node[vertex, label=0:$a_{1,1}^{\lm}$] (wf) at (0, 0) {};
        \node[vertex, label=90:$a_{1,1}^{\lt}$] (mt) at (0, 1) {};
        \node[vertex, label=270:$a_{1,1}^{\lb}$] (mb) at (0, -1) {};
        \node[vertex, label=90:$a_{1,1}^{\rt}$] (wt) at (1, 1) {};
        \node[vertex, label=270:$a_{1,1}^{\rb}$] (wb) at (1, -1) {}; 
        \node[vertex,
        label=180:$a_{1,1}^{\rrm}$] (mf) at (1, 0) {};
        \draw (m) edge node[pos=0.2, fill=white, inner sep=2pt] 
        {\scriptsize
            $i$}  node[pos=0.76, fill=white, inner sep=2pt] {\scriptsize $2$}
        (wf);
        \draw (mf) edge node[pos=0.2, fill=white, inner sep=2pt] {\scriptsize
            ${3}$}  node[pos=0.76, fill=white, inner sep=2pt]
        {\scriptsize $1$} (wt);
        \draw[bedge] (mf) edge node[pos=0.2, fill=white, inner sep=2pt] {\scriptsize
            $1$}  node[pos=0.76, fill=white, inner sep=2pt] {\scriptsize $2$}
        (wb);
        \draw (wf) edge node[pos=0.2, fill=white, inner sep=2pt] {\scriptsize
            $3$}  node[pos=0.76, fill=white, inner sep=2pt] {\scriptsize $1$}
        (mb);
        \draw[bedge] (wf) edge node[pos=0.2, fill=white, inner sep=2pt] {\scriptsize
            $1$}  node[pos=0.76, fill=white, inner sep=2pt] {\scriptsize $2$}
        (mt);
        \draw (wt) edge node[pos=0.2, fill=white, inner sep=2pt] 
        {\scriptsize
            $2$}  node[pos=0.76, fill=white, inner sep=2pt] {\scriptsize $1$}
        (mt);
        \draw (wb) edge node[pos=0.2, fill=white, inner sep=2pt] 
        {\scriptsize
            $1$}  node[pos=0.76, fill=white, inner sep=2pt] {\scriptsize $2$}
        (mb);

        \begin{scope}[xshift = 2cm]
        \node[vertex, label=0:$a_{1,2}^{\lm}$] (wf) at (0, 0) {};
        \node[vertex, label=90:$a_{1,2}^{\lt}$] (mt) at (0, 1) {};
        \node[vertex, label=270:$a_{1,2}^{\lb}$] (mb) at (0, -1) {};
        \node[vertex, label=90:$a_{1,2}^{\rt}$] (wt) at (1, 1) {};
        \node[vertex, label=270:$a_{1,2}^{\rb}$] (wb) at (1, -1) {}; 
        \node[vertex,
        label=180:$a_{1,2}^{\rrm}$] (mf) at (1, 0) {};
        \draw (mf) edge node[pos=0.2, fill=white, inner sep=2pt] {\scriptsize
            ${3}$}  node[pos=0.76, fill=white, inner sep=2pt]
        {\scriptsize $1$} (wt);
        \draw[bedge] (mf) edge node[pos=0.2, fill=white, inner sep=2pt] {\scriptsize
            $1$}  node[pos=0.76, fill=white, inner sep=2pt] {\scriptsize $2$}
        (wb);
        \draw (wf) edge node[pos=0.2, fill=white, inner sep=2pt] {\scriptsize
            $3$}  node[pos=0.76, fill=white, inner sep=2pt] {\scriptsize $1$}
        (mb);
        \draw[bedge] (wf) edge node[pos=0.2, fill=white, inner sep=2pt] {\scriptsize
            $1$}  node[pos=0.76, fill=white, inner sep=2pt] {\scriptsize $2$}
        (mt);
        \draw (wt) edge node[pos=0.2, fill=white, inner sep=2pt] 
        {\scriptsize
            $2$}  node[pos=0.76, fill=white, inner sep=2pt] {\scriptsize $1$}
        (mt);
        \draw (wb) edge node[pos=0.2, fill=white, inner sep=2pt] 
        {\scriptsize
            $1$}  node[pos=0.76, fill=white, inner sep=2pt] {\scriptsize $2$}
        (mb);

        \end{scope}

        \begin{scope}[xshift = 4cm]
        \node[vertex, label=0:$a_{1,3}^{\lm}$] (wf) at (0, 0) {};
        \node[vertex, label=90:$a_{1,3}^{\lt}$] (mt) at (0, 1) {};
        \node[vertex, label=270:$a_{1,3}^{\lb}$] (mb) at (0, -1) {};
        \node[vertex, label=90:$a_{1,3}^{\rt}$] (wt) at (1, 1) {};
        \node[vertex, label=270:$a_{1,3}^{\rb}$] (wb) at (1, -1) {}; 
        \node[vertex,
        label=180:$a_{1,3}^{\rrm}$] (mf) at (1, 0) {};
        \draw (mf) edge node[pos=0.2, fill=white, inner sep=2pt] {\scriptsize
            ${3}$}  node[pos=0.76, fill=white, inner sep=2pt]
        {\scriptsize $1$} (wt);
        \draw[bedge] (mf) edge node[pos=0.2, fill=white, inner sep=2pt] {\scriptsize
            $1$}  node[pos=0.76, fill=white, inner sep=2pt] {\scriptsize $2$}
        (wb);
        \draw (wf) edge node[pos=0.2, fill=white, inner sep=2pt] {\scriptsize
            $3$}  node[pos=0.76, fill=white, inner sep=2pt] {\scriptsize $1$}
        (mb);
        \draw[bedge] (wf) edge node[pos=0.2, fill=white, inner sep=2pt] {\scriptsize
            $1$}  node[pos=0.76, fill=white, inner sep=2pt] {\scriptsize $2$}
        (mt);
        \draw (wt) edge node[pos=0.2, fill=white, inner sep=2pt] 
        {\scriptsize
            $2$}  node[pos=0.76, fill=white, inner sep=2pt] {\scriptsize $1$}
        (mt);
        \draw (wb) edge node[pos=0.2, fill=white, inner sep=2pt] 
        {\scriptsize
            $1$}  node[pos=0.76, fill=white, inner sep=2pt] {\scriptsize $2$}
        (mb);

        \draw (mf) edge node[pos=0.2, fill=white, inner sep=2pt] 
        {\scriptsize
            ${2}$}  node[pos=0.76, fill=white, inner sep=2pt]
        {\scriptsize $j$} (w);
        \end{scope}

        \end{tikzpicture}

    \end{center}
    \caption{The replacement for $s = 3$ and $F=\{\{v_1,w_1\}\}$, where $v_1$ ranks 
    $w_1$ at the $i$-th position and $w_1$ ranks $v_1$ at the $j$-th position.
    Edges in~$M_1'$ are bold.
    }\label{fig:forbidden-edges}
\end{figure}
  
  Concerning the matching $M_1'$, we add to $M_1$ the pairs $\{\{a_{i, j}^{\lt}, a_{i,j}^{\lm}\}, \{a_{i, j}^{\rrm}, a_{i, j}^{\rb}\} : j \in [s], i\in [r]\}$. 
  Moreover, we extend the matching to be a perfect matching where each currently unmatched agent is matched to an agent that appears in the $\pend$ part of its preferences in $\mathcal{P}'_2$.
  We define~$\mathcal{P}'_1$ in such a way that $M_1'$ is a stable matching in $\mathcal{P}'_1$. 
  From $k$ we can directly compute the number of pairs $t$ in which $M_1$ and a stable matching $M_2$ in $\mathcal{P}_2$ need to overlap such that $|M_1\triangle M_2|\leq k$. 
  Let $t':=t+rs$.  
  We set $k'$ such that $M'_1$ and a stable matching $M'_2$ in $\mathcal{P}'_2$ need to overlap in $t'$ pairs. 
  
  Given a matching $M_2$ in $\mathcal{P}_2$ with $F\cap M_2=\emptyset$ and $|M_1\cap M_2|\geq t$, we construct $M'_2$ as follows. For every $i \in [r]$,
  \begin{itemize}
   \item if $v_i$ prefers $w_i$ to $M_2(v_i)$, then we add to $M_2$ pairs $\{a_{i,j}^{\lt}, a_{i, j}^{\lm}\}$, $\{a_{i,j}^{\rt}, a_{i,j}^{\rrm}\}$, and~$\{a_{i,j}^{\lb}, a_{i,j}^{\rb}\}$ for $j\in [s]$, and
   \item otherwise, we add pairs $\{a_{i,j}^{\lt}, a_{i, j}^{\rt}\}$, $\{a_{i,j}^{\rrm}, a_{i,j}^{\rb}\}$, and~$\{a_{i,j}^{\lb}, a_{i,j}^{\lm}\}$ for $j\in [s]$.
  \end{itemize}
  We get from $|M_1\cap M_2|\geq t$ that $|M'_2\cap M'_1|\geq t'$. 
 Showing by induction on~$i \in [r]$ that $a_{i,\cdot}^{\cdot}$ is not contained in a blocking pair shows that $M'_2$ is stable in $\mathcal{P}'_2$. 

  Given a stable matching $M'_2$ in $\mathcal{P}_2'$,
  it is easy to verify (by induction on~$i$) that for any~$i \in [r]$, matching~$M'_2$ contains either
  \begin{enumerate}
   \item pairs $\{a_{i,j}^{\lt}, a_{i, j}^{\rt}\}$, $\{a_{i,j}^{\lb}, a_{i,j}^{\rb}\}$ for $j\in [s]$, pair $\{v_i, a_{i, 1}^{\lm}\}$, pair $\{a_{i,j}^{\rrm}, a_{i, j+ 1}^{\lm}\}$ for $j \in [s -1]$ and pair~$\{a_{i,s}^{\rrm}, w_i\}$. 
   \item pairs $\{a_{i,j}^{\lt}, a_{i, j}^{\lm}\}$, $\{a_{i,j}^{\rt}, a_{i,j}^{\rrm}\}$, and~$\{a_{i,j}^{\lb}, a_{i,j}^{\rb}\}$ for $j\in [s]$, or
   \item pairs $\{a_{i,j}^{\lt}, a_{i, j}^{\rt}\}$, $\{a_{i,j}^{\rrm}, a_{i,j}^{\rb}\}$, and~$\{a_{i,j}^{\lb}, a_{i,j}^{\lm}\}$ for $j\in [s]$.
  \end{enumerate}
  
 Note that it is further easy to see that no agent from the forbidden pairs gadget can be matched to an agent from the $\pend$ part of its preferences in a stable matching in $\mathcal{P}'_2$. 
 Thus, if the first case applies we directly get that $|M'_1\cap M'_2|\leq |A|+(r-1)s<t'$, a contradiction. 
 Thus the first case never applies. 
 Let $M_2$ be the matching $M'_2$ restricted to the agents from $A$. 
 It is easy to see that $M_2$ is stable in $\mathcal{P}_2$. 
 Moreover, as edges from the forbiddedn edge gadget can contribute at most $rs$ to $|M'_1\cap M'_2|$ from $|M'_1\cap M'_2|\geq t+rs$ it follows that $|M_1\cap M_2|\geq t$. 

\end{proof}

We will call the gadget constructed in \Cref{lem:forbidden-edges} ``forbidden pairs gadget''.

Having described the forbidden pairs gadget,
we now formally describe how we adapt the construction from \Cref{thm:isr-master-lists} to prove an analgous results for \ISMT. 
For this, let $\mathcal{I}' = (A', \mathcal{P}'_1, \mathcal{P}'_2, M'_1)$ be the \ISR instance constructed in the proof of \Cref{thm:isr-master-lists}. 
We will now describe how we modify $\mathcal{I}'$ to arrive at an instance $\mathcal{I}$ of 
\ISMT.
The bipartition in $\mathcal{I}$ will be $U = \{a_{i, 1}^c, a_{i,3}^c : i \in [\largeConst], c \in [\ell]\} \cup \{a_{e, 2}, a_{e, 4} : e \in E\}$ and $W = \{a_{i, 2}^c, a_{i,4}^c : i \in [\largeConst], c \in [\ell]\} \cup \{a_{e, 1}, a_{e, 3} : e \in E\}$.

We start by describing how to construct the preferences of the agents in the second preference profile. 
For this, we start with their preferences in $\mathcal{P}'_2$.
 For each $i\in [m]$, let $e_i=\{v,w\}\in E$ with $v\in V^{c_i}$ and $w\in V^{d_i}$.
 We replace the master lists of $a_{e_i, 1}$, $a_{e_i,2}$, $a_{e_i,3}$, and $a_{e_i,4}$ as follows:
  \begin{align*}
    a_{e, 1} & : a_{e_1, 2} \succ a_{1, 1}^{c_1} \succ a_{1, 1}^{d_1}
    \succ a_{e_1, 4} \succ a_{e_2, 2} \succ a_{\nu + 1, 1}^{c_2} 
    \succ a_{\nu + 1, 1}^{d_2} 
    \succ a_{e_2, 4} \succ \dots \succ  a_{e_m, 2} \\
    & \succ a_{\nu(m-1) + 1,1}^{c_m}
    \succ a_{\nu (m-1) + 1,1}^{d_m}
    \succ a_{e_m, 4} \pend\\
    a_{e, 2} & : a_{e_1, 3} \sim a_{e_1, 1} \succ a_{e_2, 3} \sim a_{e_2, 1} \succ \dots \succ a_{e_m, 3} \sim a_{e_m, 1} \pend\\
    a_{e, 3} & : a_{e_1, 4} \sim a_{e_1, 2} \succ a_{e_2, 4} \sim a_{e_2, 2} \succ \dots \succ a_{e_m, 4} \sim a_{e_m, 2} \pend\\
    a_{e, 4} & : a_{e_1, 3} \succ a_{\nu, 2}^{c_1}
    \succ a_{\nu, 2}^{d_1}
    \succ a_{e_1, 1} \succ a_{e_2, 3} \succ a_{2\nu, 2}^{c_2}
    \succ a_{2\nu, 2}^{d_2}
    \succ a_{e_2, 1} \succ \dots \succ a_{e_m, 3} \\
    & \succ a_{\nu m, 2}^{c_m}
    \succ a_{\nu m, 2}^{d_m}
    \succ a_{e_m, 1} \pend
  \end{align*}
  That is, compared to the construction from \Cref{thm:isr-master-lists}, we delete all agents $a_{i \cdot \nu, 2}^{c_i}$ and $a_{i \cdot \nu, 2}^{d_i}$ from the preferences of $a_{e, 1}$ and add them to the preferences of $a_{e, 4}$.
  Furthermore, $a_{e,2}$ is indifferent between $a_{e_i, 1}$ and $a_{e_i, 3}$ and $a_{e, 3}$ is indifferent between $a_{e_i, 2} $ and $a_{e, 4}$.
  Finally, $a_{e, 4}$ now prefers~$a_{e_i, 3}$ to $a_{e_i, 1}$.

  Regarding the vertex-selection gadgets, for each $c\in [\ell]$ and $i\in [\nu m]$, the preferences of $a_{i, 1}^c$, $a_{i,3}^c $ and $a_{i, 4}^c$ remain unchanged, while agent~$a_{i, 2}^c$ replaces~$a_{e, 1}$ by $a_{e, 4}$ in its preferences.
  Note that the modified instance is now indeed bipartite and that we denote the constructed preference profile as $\mathcal{P}^*_2$ (in the following, we will use this preference profile to construct an intermediate \ISMT instance $\mathcal{I}^*$).
  Moreover, we construct the initial matching $M^*_1$ in the modified instance such that it contains the pairs $\{ \{a_{e,1}, a_{e, 4} \} : e \in E\}$ and such that all other agents are matched to agents that appear in the $\pend$ part of their preferences. 
  As the difference between the first and second preference profile can be unbounded, we construct $\mathcal{P}^*_1$ such that $M^*_1$ is stable.
  As all preferences are complete, the second matching needs to be perfect. 
  Thus, we can set the budget $k^*$ in such a way that $M^*_1$ and the solution need to share at least $\binom{\ell}{2}$ edges.
  Thereby, we arrive at an \ISMT instance $\mathcal{I}^*=(A,\mathcal{P}^*_1,  \mathcal{P}^*_2, M^*_1, k^*)$. 
  Note in particular that in $\mathcal{P}^*_2$ the preferences of all agents can still be derived from $\mathcal{O}(\ell)$ master lists. 
  
  We now continue by modifying $\mathcal{I}^*$ by forbidding some edges. 
  Specifically, for every $c\in [\ell]$, $e\in E$, and $p\in [\nu m]$, we ``forbid'' the pairs~$\{a_{e, 1}, a_{p,1}^{c}\}$ and $\{a_{e, 4}, a_{p, 2}^c\}$. 
  We do this by iteratively applying \Cref{lem:forbidden-edges} to $\mathcal{I}^*$. 
  Specifically, for increasing $c\in [\ell]$,  we tranform the previously construct \ISMT instance by applying \Cref{lem:forbidden-edges} with $F=\{\{a_{e, 1}, a_{p,1}^{c}\}, \{a_{e, 4}, a_{p, 2}^c\}  \mid p\in [\nu m], e\in E\}$. 
  We call the resulting \ISMT instance $\mathcal{I}=(A, \mathcal{P}_1, \mathcal{P}_2, M_1, k)$, which is now our final constructed \ISMT instance. 
  Note that by applying \Cref{lem:forbidden-edges} $\ell$ times we have only added $\mathcal{O}(\ell)$ many master lists, thus still arriving at overall $\mathcal{O}(\ell)$ master lists. 
  By \Cref{lem:forbidden-edges}, we can directly conclude the following: 
  \begin{lemma}\label{lem:equiv}
There is a stable matching $M^*_2$ in $\mathcal{P}^*_2$ without a pair from  
$\{\{a_{e, 1}, a_{p,1}^{c}\}, \{a_{e, 4}, a_{p, 2}^c\}  \mid p\in [\nu m], e\in E, c\in [\ell]\}$ 
for which $|M^*_1\triangle M^*_2|\leq k^*$ if and only if there is a stable matching $M_2$ in $\mathcal{P}_2$ with $|M_1\triangle M_2|\leq k$.
  \end{lemma}

We are now ready to prove the correctness of the reduction, thereby proving \Cref{pr:ISMWML} stated in the beginning of this section:

\ISMWML*
\begin{proof}
 We already described how to adapt the reduction from the proof of \Cref{thm:isr-master-lists}.
It remains to show its correctness.
  
  \paragraph*{Forward Direction}
  Let $X$ be a multicolored clique.
  We construct a stable matching~$M^*_2$ the same way as in the proof of \Cref{thm:isr-master-lists}.
  Note that in fact $M^*_2$ only contains pairs that occur in $\mathcal{I}^*$ and thus in particular respects that the instance is bipartite. 
  Further,  using the same arguments as in \Cref{thm:isr-master-lists}, we can prove that $M^*_2$ is stable in $\mathcal{P}^*_2$ and that it overlaps with $M^*_1$ in at least ${\ell}\choose{2}$ edges. 
  Thus we have that $|M^*_1\triangle M^*_2|\leq k^*$. 
  As $M^*_2$ contains no pairs from $\{\{a_{e, 1}, a_{p,1}^{c}\}, \{a_{e, 4}, a_{p, 2}^c\}  \mid p\in [\nu m], e\in E, c\in [\ell]\}$, using \Cref{lem:equiv} it follows that $\mathcal{I}$ is a yes instance. 
  
  \paragraph*{Backward Direction}
    Let $M_2$ be a stable matching in $\mathcal{P}_2$ with $|M_1\triangle M_2|\leq k$. 
    Then by \Cref{lem:equiv}, we get a stable matching in $\mathcal{P}^*_2$ with $|M^*_1\triangle M^*_2|\leq k^*$ which does not contain any edges from $\{\{a_{e, 1}, a_{p,1}^{c}\}, \{a_{e, 4}, a_{p, 2}^c\}  \mid p\in [\nu m], e\in E, c\in [\ell]\}$.
    Specifically, by the design of $k^*$ we have that $M^*_1$ and $M^*_2$ overlap in 
  ${\ell}\choose{2}$ edges. 
  
  Before we proceed, we show that in $M^*_2$ no agent is matched to an agent that appears in the   $\pend $ part of its preferences in $\mathcal{P}_2^*$. 
  \begin{lemma}\label{lem:pend}
    Let $M^*_2$ be a stable matching in $\mathcal{P}^*_2$ with~$|M^*_1 \triangle M^*_2| \le k^*$.
    Then any agent $a \in A^*$ is matched to an agent appearing before $\pend $ in its preferences.
    Moreover, for any edge~$ e\in E$, matching~$M_2^*$ contains either~$ \{a_{e, 1}, a_{e, 2}\} $ and $\{a_{e, 3}, a_{e,4}\} $ or $\{a_{e, 1}, a_{e, 4}\}$ and $\{a_{e, 2}, a_{e,3}\}$.
  \end{lemma}

  \begin{proof}
   We first show the second part of the statement.
   By \Cref{lem:forbidden-edges}, for every~$ r \in [m]$ and $e' \in E$, $a_{e', 1}$ is matched to neither~$a^{c_r}_{\nu (r-1) + 1, 1}$ nor~$a^{d_r}_{\nu (r-1) + 1, 1}$.
   Similarly, for every~$r\in [m]$ and $e' \in E$, $a_{e', 4} $ is matched to neither~$a^{c_r}_{\nu r, 2}$ nor~$a^{d_r}_{\nu r, 2}$. 
   Assume towards a contradiction that there exists some~$a_{e_i, j}$ for~$j \in \{1, 3\}$ with $M^*_2 (a_{e_i, j}) \notin \{a_{e_i, 2}, a_{e_i, 4}\}$, where $i$ is minimal.
   Then using our above observations, there also exists some~$j' \in \{2, 4\}$ with $M^*_2 (a_{e_i, j'}) \notin \{a_{e_i, 1}, a_{e_i, 3}\}$.
   Then $\{a_{e_i, j}, a_{e_i, j'}\} $ blocks~$M_2^*$, a contradiction to the stability of~$M^*$.
   
   To prove the first part, it remains to consider agents~$a^c_{p, j}$.
   The proof is analogous to the case of an agent~$a_{e_i,j}$:
   By \Cref{lem:forbidden-edges},  $a^c_{p, 1}$ is not matched to~$a_{e, 1}$ for every~$c\in [\ell]$, $ p \in [\nu m]$ and $e \in E$.
   Similarly, $a^c_{p, 2} $ is not matched to~$a_{e, 4}$ for every~$c\in [\ell]$, $p \in [\nu m]$ and $e \in E$.
   Assume towards a contradiction that there exists some~$a^c_{p, j}$ for~$j \in \{1, 3\}$ with $M^*_2 (a^c_{p, j}) \notin \{a^c_{p, 2}, a^c_{p, 4}\}$, where $p$ is minimal.
   Then, using our above observations, there also exists some~$j' \in \{2, 4\}$ with $M^*_2 (a^c_{p, j'}) \notin \{a^c_{p, 1}, a^c_{p, 3}\}$.
   Then $\{a^c_{p, j}, a^c_{p, j'}\} $ blocks~$M_2^*$, a contradiction to the stability of~$M^*$.
  \end{proof}
  Using this and our above observation on forbidden edges it follows that $M^*_2$ contains only edges inside vertex-selection and inside edge gadgets.
  Note that for each edge~$e_r = \{v^c_p, v^d_q\}$, matching~$M^*_2$ can contain~$\{a_{e_r, 1}, a_{e_r, 4}\}$ and $\{a_{e_r ,2 }, a_{e_r, 3}\}$ only if agent~$a^c_{\nu (r-1 ) + 1, 1}$ is matched at least as good as $a^c_{\nu (r-1 ) + p, 2}$ (otherwise~$\{a^c_{\nu (r-1) + 1, 1}, a_{e_r, 1}\}$ would block~$M_2^*$) and agent~$a^c_{\nu r, 2}$ is matched at least as good as $a^c_{\nu r - p + 1, 1}$ (otherwise~$\{a^c_{\nu r, 2}, a_{e_r, 4}\}$ would block~$M_2^*$). Symmetrically, $a^d_{\nu (r-1 ) + 1, 1}$ must be matched at least as good as $a^d_{\nu (r-1 ) + q, 2}$ (otherwise~$\{a^d_{\nu (r-1) + 1, 1}, a_{e_r, 1}\}$ would block~$M_2^*$) and agent~$a^d_{\nu r, 2}$ is matched at least as good as $a^d_{\nu r - q + 1, 1}$ (otherwise~$\{a^d_{\nu r, 2}, a_{e_r, 4}\}$ would block~$M_2^*$).
  From here on, the proof is along the lines of the proof of \Cref{thm:isr-master-lists}.
\end{proof}
 
\section{Conclusion}

Among others, answering two open questions of Bredereck et al.~\cite{DBLP:conf/aaai/BredereckCKLN20}, we have contributed to the study of the computational complexity of adapting stable matchings to changing preferences. 
From a broader algorithmic perspective, in particular, the ``propagation'' technique from our XP-algorithm for the number of swaps, and the study of the number of different preference lists/master lists as a new parameter together with the needed involved constructions for the two respective hardness proofs could be of interest. 

There are several possibilities for future work. As direct open questions, for the parameterization by the number of outliers, we do not know whether \ISMT or \ISRT
are fixed-parameter tractable. Moreover, it remains open whether \ISR or \ISMT with complete preferences is polynomial-time solvable for a constant number of master lists.

While we have already considered some (new) parameters measuring the distance of a preference profile from the case of a master list (e.g., the number of outliers and the number of master lists), there also exist various other possibilities to measure the similarity (or, more generally, the structure) of the preferences.
Possible additional parameters to measure the similarity to a master lists that might be worthwhile to explore further include, for instance, the maximum or average distance of an agent's preference list from the master list.  

 \subsection*{Acknowledgments}
 NB was supported by the DFG project MaMu (NI 369/19)
	and by the DFG project ComSoc-MPMS (NI 369/22). KH
	was supported by the DFG Research Training Group 2434
	``Facets of Complexity'' and by the DFG project FPTinP (NI
	369/16).

\end{document}